\documentclass[twoside,11pt]{article}

\usepackage{jmlr2e}
\usepackage{amsmath}
\usepackage{amsfonts}
\usepackage{bbm}
\usepackage{enumerate}
\usepackage{float}
\usepackage{color}
\usepackage{graphics}
\usepackage[justification=centering]{caption}
\usepackage{mathrsfs}
\usepackage{psfrag,epsf}
\usepackage{url} % not crucial - just used below for the URL
\newtheorem{algorithm}{Algorithm}
\def\qed{\hfill\BlackBox}
\newcommand{\ba}{\begin{array}}
	\newcommand{\ea}{\end{array}}
\newcommand{\ee}{\end{eqnarray}}

\newcommand{\es}{\end{split}\end{align}}

\newcommand{\F}{\mathcal{F}}

\newcommand{\X}{\mathcal{X}}
\newcommand{\Y}{\mathcal{Y}}

\newcommand{\Ep}{{\mathrm{E}}}

\renewcommand{\Pr}{{\mathrm{P}}}

\def\RR{ {\Bbb{R}}}

\newcommand{\nOn}[1]{{#1}^{-}}

\definecolor{DSgray}{cmyk}{0,1,0,0}

\firstpageno{1}

\begin{document}
	
	\title{Shape-Enforcing Operators for Generic Point and Interval  Estimators of Functions}
	
	\author{\name Xi Chen \email xc13@stern.nyu.edu\\
		\addr Stern School of Business\\
		New York University, New York, NY 10012, USA
		\AND
		\name Victor Chernozhukov \email vchern@mit.edu\\
		\addr Department of Economics + Center for Statistics and Data Science
		\\
		Massachusetts Institute of Technology, Cambridge, MA 02142, USA
		\AND
		\name Iv\'an Fern\'andez-Val \email ivanf@bu.edu\\
		\addr Department of Economics\\
		Boston University, Boston, MA 02215-1403, USA
		\AND
		\name Scott Kostyshak \email skostyshak@ufl.edu\\
		\addr Department of Economics\\
		University of Florida, Gainesville, Florida 32611-7140, USA
		\AND 
		\name Ye Luo \email kurtluo@hku.hk\\
		\addr HKU Business School\\
		The University of Hong Kong, Pok Fu Lam, Hong Kong}
	
	\editor{}
	
	\maketitle
	
	\begin{abstract}%   <- trailing '%' for backward compatibility of .sty file
		
		A common problem in econometrics, statistics, and machine learning is to estimate and make inference on functions that satisfy shape restrictions. For example, distribution functions are nondecreasing and range between zero and one, height growth charts  are nondecreasing in age, and production functions are nondecreasing and quasi-concave in input quantities. We propose a method to enforce these restrictions \textit{ex post} on generic unconstrained point and interval estimates  of the target function by applying functional operators. The interval estimates could be either frequentist confidence bands or Bayesian credible regions. If an operator has reshaping,  invariance, order-preserving, and distance-reducing  properties,  the shape-enforced point estimates are closer to the target function than the original point estimates and the shape-enforced interval estimates have greater coverage and shorter length than the original interval estimates. We show that these properties hold for six different operators that cover commonly used shape restrictions in practice: range, convexity, monotonicity, monotone convexity, quasi-convexity, and monotone quasi-convexity, with the latter two restrictions being of paramount importance.  The main attractive property of the post-processing approach is that it works in conjunction with \textit{any generic} initial point or interval estimate, obtained using  any of parametric, semi-parametric or nonparametric learning methods, including recent methods that are able to exploit either smoothness, sparsity, or other forms of structured parsimony of target functions. The post-processed point and interval estimates automatically inherit and provably improve these properties in finite samples, while also enforcing qualitative shape restrictions brought by scientific reasoning.  We illustrate the results with two empirical applications to the estimation of a height growth chart for infants in India and a production function for chemical firms in China.
	\end{abstract}
	
	
\section{Introduction}

A common problem in econometrics, statistics, and machine learning is to estimate and make inference on functions that satisfy shape restrictions. These restrictions might arise either from the nature of the function and variables involved or from theoretical reasons. Examples of the first case include distribution functions, which are nondecreasing and range between zero and one, and height growth charts, which are nondecreasing in age. Examples of the second case include demand functions of utility-maximizing individuals, which are nonincreasing in price according to consumer demand theory; production functions of  profit-maximizing firms, which are nondecreasing \textit{and} quasi-concave in input quantities according to production theory \textit{and} can also be concave in industries that exhibit diminishing returns to scale; bond yield curves, which are monotone and concave in time to maturity; and American and European call option prices, which are concave \textit{and} monotone in the underlying stock price and increasing in volatility, according to the arbitrage pricing theory.\footnote{Different, but similar shape restrictions apply to put prices,  with the American put price being log-concave in the stock price, for example.} 

We propose a method to enforce shape restrictions \textit{ex post} on any initial generic point and interval estimates of  functions by applying functional operators. If an operator has reshaping, invariance, order-preserving, and distance-reducing properties,  enforcing the shape restriction improves the point estimates and improves the coverage property of the interval estimates. Thus,  the shape-enforced  point estimates are closer to the target  function than the original point estimates under suitable distances, and the shape-enforced interval estimates have greater coverage and shorter length under suitable distances than the original interval estimates.  We show that these properties hold for six different operators that enforce the following restrictions: range, convexity, monotonicity, joint convexity and monotonicity, quasi-convexity, and joint quasi-convexity and monotonicity, as well as for combinations of range with all of the above.   We impose the range restriction with a natural operator that thresholds the estimates to the desired range. The double Legendre-Fenchel transform enforces convexity by transforming the estimates into their greatest convex minorants.  We focus on the monotone rearrangement to enforce monotonicity (though projection on isotone class can also be used in all composition results, as well as convex combinations of isotone projection with rearrangement). We further develop a new operator to enforce quasi-convexity---a shape that has not been well explored in the literature, although it is common in applications.  We also show that the compositions of the monotone rearrangement with the double Legendre-Fenchel and the new quasi-convexity operators yield monotone convex and monotone quasi-convex estimates, respectively.  In other words, the application of the convex and quasi-convex operators does not affect the monotonicity of the function. We further demonstrate how to modify the operators to deal with concavity, quasi-concavity, their composition with the monotonicity and range operators, and shape restrictions on transformations of the function.

Our method is generic in that it can be applied to any point or interval estimator of the target function. For example, it works in combination with parametric, semi-parametric and nonparametric approaches to model and estimate the target function. It works with modern machine and deep learning methods that are able to exploit either smoothness or structured parsimony (e.g., approximate sparsity) of target functions.  
Hence our post-processed point and interval estimates automatically inherit the rates of convergence of these estimators and provably improve these properties in finite samples, while also enforcing qualitative shape restrictions brought by scientific reasoning.  
Moreover, our method applies without modification to any type of function including reduced form statistical objects such as conditional expectation, conditional density, conditional probability and conditional quantile functions, or causal and structural objects such as dose response, production, supply and demand functions identified and estimated using instrumental variable or other methods. The only requirement to obtain consistent point estimators or valid confidence bands is that the source point estimators be consistent or the source confidence bands be valid. There are many existing methods to construct such point estimators and confidence bands under general sampling conditions, including obtained through frequentist, Bayesian or approximate Bayesian methods.\footnote{ Bayesian methods are often used to quantify the uncertainty of complicated methods where the frequentist quantification is intractable, for example, in deep learning problems. Like in the classical approach, one may impose constraints directly during the estimation, though this is often quite cumbersome and is rarely done in practice. The post-processing can be applied to the unconstrained estimates and  be justified on pragmatic grounds, ease of computation,  or desire to analyze data without restrictions and accept  a menu of restrictions ex-post only after validating them.} Under misspecification these requirements may not be satisfied, but the shape-enforcing operators will bring improvements to the point estimators and confidence bands in a sense that we will make precise. To implement our method, we develop  algorithms to compute the Legendre-Fenchel transform of multivariate functions and the new quasi-convexity enforcing operator.

We illustrate the theoretical results with two empirical applications to the estimation of a height growth chart for infants in India and a production function for chemical firms in China.  In the case of the growth chart, we impose natural monotonicity in the effect of age, together with concavity that is plausible during early childhood. In the case of the production function, we enforce that a firm's output is nondecreasing and quasi-concave in labor and capital inputs according to standard production theory. We also consider imposing concavity in the effect of the inputs. In both applications we use series least squares methods to flexibly estimate the conditional expectation functions of interest, and construct confidence bands using bootstrap.   We quantify the size of strict improvements that imposing shape restrictions bring to point and interval estimates in small samples through numerical simulations calibrated to the empirical applications.

\smallskip

\paragraph{\textbf{Literature Review.}}
Due to the wide range of applications of shape restrictions, shape-constrained estimation and inference have received a lot of attention in the statistics community. Classical examples include \cite{hildreth1954point}, \cite{ayer1955empirical}, \cite{brunk1955maximum}, \cite{vaneeden}, \cite{grenander1956theory}, \cite{groeneboom2001estimation}, and \cite{mammen1991estimating}. We refer to \cite{barlowstatistical} and \cite{robertsonorder} for classical references on isotonic regression for monotonicity restrictions, and to \cite{Koenker:10} for the work on log-concave density estimation.
In terms of  risk bounds for estimation, please refer to \cite{Zhang02}, \cite{GuntuAnnIso}, \cite{han2017isotonic}, and references therein for recent developments in isotonic regression; and \cite{kuosmanen2008representation}, \cite{seijo2011nonparametric}, \cite{Guntuboyina:15}, and \cite{Han:16:convex} for convex regression.   \cite{bellec2018sharp} established sharp oracle inequalities for least squares estimators, when only shape restrictions are known to hold. Moreover, \cite{hengartner1995finite}, \cite{dumbgen2003}, and \cite{anevski2006} considered the construction of confidence bands for univariate functions under monotonicity or convexity restrictions.  Please refer to the book \cite{groeneboom2014nonparametric} and the survey paper \cite{guntuboyina2017nonparametric} for more comprehensive reviews on estimation and inference under shape constraints.

Most existing works developed constrained methods via  maximum likelihood methods for regression or density estimation that impose only  shape restrictions and  produce constrained estimates without further restrictions.  We remark here that these direct approaches deliver advantages over our approach when such target functions are known to satisfy only the qualitative shape constraints. By contrast, our post-processing approach delivers advantages when \textit{any generic} target function, in addition to satisfying qualitative constraints, satisfies smoothness or other structured parsimony restrictions (e.g., sparsity).  Indeed, our method applies to generic problems, and is not tied to statistical parameters such as regression or density estimation.  To explain where the advantages arise, we note that the direct isotone  univariate regression converges to the true regression function at the $n^{-1/3}$ rate, which
is minimax optimal for the parameter space of monotone functions. If the target function is known to lie in  the space of smooth functions (H\"older or Sobolev with smoothness $s>1$), the better and optimal rate  \textcolor{black}{$n^{-s/(2s+1)}$} can be achieved by an unconstrained estimator  (e.g., \cite{Stone:80}), making the pure isotonic regression suboptimal in this case.  To fix the direct isotonic regression in this case, we would need to impose the smoothness constraints in the estimation directly, which ordinarily is not done in practice, let alone theoretically analyzed. (One exception here is \cite{Chernozhukov:15:constrained} that considered testing shape restrictions in Banach spaces, with the target
function being (possibly partially) identified by general conditional moment condition problems, where shape restrictions induce a lattice structure on the space). Smooth cases and other problems, where unconstrained estimators achieve optimal rates, provide the chief motivation for our approach: in such cases, our method automatically inherits the optimal rate and improves the finite sample properties of the estimator through the distance-reducing properties.  \textcolor{black}{On the other hand, unlike constrained estimators, unconstrained estimators require delicate choices of tuning parameters to achieve the optimal rate, although adaptive estimation and inference of smooth functions is possible using the method of  \cite{l91} \citep[e.g.,][]{ls97,gn10b,gn10,cck14}.} It is worthwhile noting that \cite{cl2018} investigated asymptotic properties of smoothed isotonic estimators, and 
\cite{jw2009}  studied the method of rearrangements for obtaining discrete monotone distributions. However, these works only consider very specific classes of shape-constrained estimators, i.e., isotonic and/or discrete estimators.

Another recurrent problem with imposing shape restrictions in estimation is that the derivation of the statistical properties of the constrained estimators is involved and specific to the estimator and shape restriction. As a consequence, there exist very few distributional results, mainly for univariate functions.  The results available for the Grenander and isotonic regression estimators show that these estimators exhibit non-standard asymptotics (including relatively slow rates, since smoothness conditions are not exploited); see \cite{guntuboyina2017nonparametric} for a recent review. 
\textcolor{black}{Moreover, \cite{Horowitz_Lee_2017} and \cite{Freyberger2017InferenceUS} have recently pointed out the difficulties of developing inference methods from shape-constrained estimators with good uniformity properties with respect to the data generating process. They showed that for shape restrictions defined by inequalities, the distribution of the constrained estimator depends on where the inequalities are binding, which is unknown a priori. Inference based on this distribution  therefore becomes sensitive to how close the inequalities are to binding relative to the sample size.}  We avoid all of these complications arising from the constrained estimators by enforcing the restrictions \textit{ex post} and therefore relying on the distribution of the unconstrained estimators (whenever it is available) to construct the confidence bands.  Our confidence interval method can also be applied on top of a different constrained estimator to provide potential improvements when the end-point functions of the generated confidence band do not themselves satisfy the restriction. It is worthwhile noting that the idea of \textit{ex post} confidence bands was mentioned in Section 4.2 in \cite{dumbgen2003}, which only discussed two cases on the univariate monotone function and univariate convex function. Moreover, the construction of confidence bands for the convex case in \cite{dumbgen2003} is quite different from ours (e.g., their lower bound is not necessarily a convex function).

Our paper generally follows the approach introduced in  \cite{CFG09}, which focused on producing improved generic point and interval estimates of monotone functions using the monotone rearrangement. The class of shape enforcing operators covered by our paper is much bigger and much more useful, with analysis being much more challenging, and we view both aspects as a substantial contribution of our paper.  Some of the operators that we consider have been analyzed previously in the literature. \cite{Dette:08} apply a smoothed rearranged operator to kernel estimators for monotonization purposes and derive pointwise limit theory. \cite{CFG10} applied the monotone rearrangement to deal with the quantile crossing problem and \cite{BCCF} to impose monotonicity in conditional quantile functions estimated using series quantile regression methods and construct monotonized uniform confidence bands.  \cite{beare2017weak} used the double Legendre-Fenchel transform to construct point and interval estimates of univariate concave functions on the non-negative half-line.  Other applications of the double Legendre-Fenchel transform include \cite{delgado2012distribution}, \cite{beare2015nonparametric}, \cite{beare2016empirical}. \cite{CFG10} and \cite{beare2017weak}  used an alternative approach to make inference on shape-constrained functions. Instead of applying the shape-enforcing operator to a confidence band constructed from an unconstrained estimator, they constructed confidence bands from the estimator after applying the shape-enforcing operator. To do so, they characterized the distribution of the constrained estimator from the distribution of the unconstrained via the delta method, after showing that the shape-enforcing operator is Hadamard or  Hadamard directional differentiable. This approach usually yields narrower confidence bands than ours,  but it is computationally more involved and requires additional assumptions and non-standard methods. For example, \cite{beare2017weak} showed that the bootstrap is inconsistent for the distribution of constrained estimators after applying the double Legendre-Fenchel transform when the target function is not strictly concave.  Finally, we refer to \cite{matzkin1994restrictions}, and \cite{chetverikov2017econometrics} for excellent, insightful up-to-date surveys on the use of shape restrictions in econometrics.

\smallskip

Relative to the literature, we summarize the major contributions of this paper as follows.  (1) We introduce an operator to enforce quasi-convexity and deliver  improved point and interval estimates of general multivariate quasi-convex functions. Quasi-convexity extends the notion of unimodality to multiple dimensions and generalizes convexity constraints. Despite its importance, the shape restriction of quasi-convexity has not been well studied in the literature and \cite{guntuboyina2017nonparametric} listed quasi-convexity as an open area in shape-constrained estimation. (2) We extend the use of the Legendre-Fenchel transform to construct improved point and interval estimates of general multivariate convex functions. (3) We show that the composition of the monotone rearrangement with the Legendre-Fenchel transform can be used to construct improved point and interval estimates of monotone convex functions. (4) We show that the composition of the monotone rearrangement with our quasi-convex operator can be used to construct improved point and interval estimates of monotone quasi-convex functions.  (The third and fourth contributions proved to be the most challenging and important steps, where the importance stems from shape restrictions  often being a composition of monotonicity with convexity or quasi-concavity). (5) We provide a new algorithm to compute the Legendre-Fenchel transform of multivariate functions. (6) We develop an algorithm to compute our quasi-convex operator. The main advantage of our approach is that it works in conjunction
with any generic point estimate (e.g., including recent machine and deep learning methods),  or any generic interval estimate (that can be a frequentist confidence band or a Bayesian credible region).  Because of genericity, it is able to exploit smoothness or other forms of structured parsimony through the use of the appropriate initial estimator.  It inherits the rate properties of the initial estimator, while delivering better finite sample properties through distance-reducing inequalities.\\

\paragraph{\textbf{Outline.}} The rest of the paper is organized as follows. Section~\ref{sec:op} introduces the functional shape-enforcing operators and their properties, together with examples of operators that enforce the shape restrictions of interest. Section~\ref{sec:inf} discusses the use of shape-enforcing operators to obtain improved point and interval estimates of functions that satisfy shape restrictions. Section~\ref{sec:alg} provides  algorithms to compute the shape-enforcing estimators. Section~\ref{sec:num} reports the results of two empirical applications and numerical simulations calibrated to the applications. Section~\ref{sec:con} concludes the paper. The proofs of the main results are gathered in the Appendix.

\smallskip

\paragraph{\textbf{Notation. }} For any measurable function $f: \X \to \RR$ and $p \geq 1$,  let $\|f \|_{p} := \left\{\int_{\X} |f(x)|^p dx\right\}^{1/p}$, the  $L^p$-norm of $f$, with $\|f \|_{\infty} := \sup_{x \in \X} |f(x)|$, the $L^{\infty}$-norm or sup-norm of $f$.  We drop the subscript $p$ for the Euclidean norm, i.e., $\|x\| := \|x\|_2$.  For $p\geq 1$, let $\ell^p(\mathcal{X}) := \{f: \X \to \RR : \|f\|_p < \infty  \}$, the class of all measurable functions defined on $\mathcal{X}$ such that the $L^p$-norms of these functions is finite. For $x,x'\in \mathbb{R}^k$, we say $x\geq x'$ if every entry of $x$ is no smaller than the corresponding entry in $x'$.  For two functions $f$ and $g$ that map $\X \to \RR$
we say that $f \leq g$ if $f(x) \leq g(x)$ for all $x \in \X$.
We also use $a\vee b:=\max(a,b)$ and $a\wedge b:=\min(a,b)$ for any $a,b \in \RR$. For two scalar sequences $a_n$ and $b_n$, the notation $a_n \sim b_n$ means that $a_n/b_n \to 1$ as $n \to \infty$. For an operator $\mathbf{O}: \ell^{\infty}_0(\X) \to \ell^{\infty}_0(\X)$, we use $\mathbf{\nOn{O}}f:=-\mathbf{O}(-f)$. For two operators $\mathbf{O_1}: \ell^{\infty}_0(\X) \to \ell^{\infty}_0(\X)$ and $\mathbf{O_2}: \ell^{\infty}_0(\X) \to \ell^{\infty}_0(\X)$, we define $\mathbf{O_{1}O_{2}}$ to be the composition $\mathbf{O_{1}} \circ \mathbf{O_{2}}$.

\section{Functional Shape-Enforcing Operators}\label{sec:op}

\subsection{Properties of Shape-Enforcing Operators}

Assume that the function of interest, $f$, is real-valued with domain $\mathcal{X} \subset \RR^k$, for some positive integer $k$.  Let $\ell^{\infty}(\X)$ be the set of bounded \textit{measurable} functions.
%\footnote{Measurability is needed in almost all results.} 
Let $\ell^{\infty}_0(\X)$ and  $\ell^{\infty}_1(\X)$ be two subspaces of $\ell^{\infty}(\X)$,  such that $\ell^{\infty}_1(\X) \subset \ell^{\infty}_0(\X)$; and let $\mathbf{O}: \ell^{\infty}_0(\X) \to \ell^{\infty}_0(\X)$ be a functional operator.  In our case, $\ell^{\infty}_0(\X)$ will be the class of unconstrained functions and $\ell^{\infty}_1(\X)$ will be the subclass of functions that satisfy some shape restriction. We first introduce three properties that an operator must satisfy to be considered a shape-enforcing estimator.

\begin{definition}[Shape-Enforcing Operator]\label{def:SE_operators} We say that an operator $\mathbf{O}$ is  $\ell^{\infty}_1$-enforcing with respect to $\ell^{\infty}_0(\X)$ if it satisfies the following properties:
	
	\begin{enumerate}
		\item Reshaping:  the output of the operator is a function that satisfies the shape restriction:
		\begin{equation}
			\mathbf{O} f \in \ell^{\infty}_1(\X), \textrm{ for any } f \in \ell^{\infty}_0(\X).
		\end{equation}
		
		\item Invariance: the operator should do nothing when the input function has already satisfied the shape restriction:
		\begin{equation}
			\mathbf{O} f =f, \textrm{ for any } f\in \ell^{\infty}_1(\X).
		\end{equation}
		
		\item Order Preservation: the output functions preserve original order:
		\begin{equation}
			\mathbf{O}f  \leq \mathbf{O} g, \textrm{ for any } f,g \in \ell^{\infty}_0(\X) \textrm{ such that } f\leq g.
		\end{equation}
	\end{enumerate}
	
\end{definition}

In addition to these properties, we consider the following
``distance contraction" property.

\begin{definition}[Distance-Reducing Operator]\label{def:dc} Let $\rho$ be a distance or semi-metric function on $\ell^\infty(\X)$. We say that an operator $\mathbf{O}$ is a $\rho$-distance contraction if the output functions are weakly closer than input functions
	under the $\rho$:
	\begin{equation}
		\rho(\mathbf{O} f,\mathbf{O} g)\leq \rho(f,g) \ \textrm{ for any } f,g \in \ell^{\infty}_0(\X).
	\end{equation}
\end{definition}

Some particularly interesting cases of constrained classes $\ell^{\infty}_1(\X)$ are  subsets of functions that satisfy shape restrictions. In this paper, we focus on seven types of shape restrictions:  (1) range, (2) convexity,  (3) monotonicity, (4) monotone convexity, (5) quasi-convexity, (6) monotone quasi-convexity, and (7) compositions of range with all of the above. Our methods also apply to the restrictions of concavity and quasi-concavity by noting that if $f$ is concave (quasi-concave), then $-f$ is convex (quasi-convex). In the case of monotonicity we focus on the case of monotonically nondecreasing functions. The  methods  also apply to monotonically  nonincreasing functions noting that if $f$ is nondecreasing then $-f$ is nonincreasing.

\textcolor{black}{
	\begin{remark}[Counterexample] There are operators to impose shape restrictions that do not satisfy the conditions of Definitions  \ref{def:SE_operators} and \ref{def:dc}. For example, $\mathbf{O} f = f/\|f\|_1$ is a natural operator to enforce that a non-negative function integrates to one in density estimation. This operator satisfies reshaping and invariance, but does not satisfy order preservation nor distance contraction for any $L^p$-norm.\end{remark}}

\subsection{Range Restrictions} We first consider the subset of range-constrained functions $\ell^{\infty}_{R}(\mathcal{X}):=\{f \in \ell^{\infty}(\mathcal{X}): \underline{f} \leq f(x)\leq \overline f \textrm{ for all } x \in \X\}$ for some constants $\underline{f} \leq \overline{f}$.  A natural range-enforcing operator is as follows.

\begin{definition}[$\mathbf{R}$-Operator] For any set $\X \subset \mathbb{R}^k$, the range operator
	$\mathbf{R} : \ell^{\infty}(\mathcal{X}) \to \ell^{\infty}(\mathcal{X})$ is defined by
	thresholding the values of the function $f$ to $[\underline{f}, \overline{f}]$.
	\begin{equation}
		\mathbf{R}f(x) := [\underline{f} \vee f(x)] \wedge \overline{f}, \ \textrm{for any } x \in \X.
	\end{equation}
\end{definition}

%Define $L_{[a,b]}(\mathcal{X}):=\{f\in L^{\infty}\{\mathcal{X}\}: a\leq f(x)\leq b \textrm{ for all } x\}$.
Let $d_p$ be the distance measure induced by the $L^p$-norm, i.e., $d_p(f,g) = \| f - g\|_p$ for any $f,g \in \ell^p(\X)$ and $p\geq1$. The following theorem shows that $\mathbf{R}$ is indeed range-enforcing and distance-reducing with respect to $d_p$.

\begin{theorem}[Range-Enforcing Operator]\label{RR}
	The operator $\mathbf{R}$ is $\ell^{\infty}_{R}$-enforcing with respect to $\ell^{\infty}(\X)$ and a  $d_p$-distance contraction  for any $p\geq1$.
\end{theorem}

%\xnote{by the proof, it seems to me that this operator is o distance contraction with respect to the discrepancy function $L^{p}(\cdot,\cdot)$ for any $p> 0$ (instead of $p \geq 1$). But of %course, $L_p$ for $0<p<1$ will not be a metric space.}

\subsection{Convexity} Let $\mathcal{X}$ be a convex subset of $\Bbb{R}^k$. Let $\ell^{\infty}_{S}(\X) := \{ f \in \ell^{\infty}(\mathcal{X}) : \liminf_{x' \to x}  f(x') \geq f(x) \textrm{ for all }  x \in \X  \}$, the set of bounded lower semi-continuous functions on $\X$, and  $\ell^{\infty}_C(\X) := \{f \in \ell_S^{\infty}(\mathcal{X}): f(\alpha x + (1-\alpha) x') \leq \alpha f(x) + (1-\alpha) f(x') \textrm{ for all } x,x' \in \X, \alpha \in [0,1] \}$, the subset
of convex functions on $\X$. We consider the Double Legendre-Fenchel (DLF)  transform as a convexity-enforcing operator. To define this operator, we first recall the definition of the  Legendre-Fenchel transform (see, e.g.,  \cite{ConvAna:01}).

\begin{definition}[Legendre-Fenchel transform]\label{def:FL}
	For any convex set $\X \subset \mathbb{R}^k$ and $f \in \ell^{\infty}(\mathcal{X})$, let $f^*(\X) := \{\xi \in \RR^k: \sup_{x \in \X}\{\xi'x  - f(x) \} < \infty \}$. The Legendre-Fenchel transform $\mathbf{L}_{\X} : \ell^{\infty}(\mathcal{X}) \to \ell^{\infty}(f^*(\mathcal{X}))$  is defined by
	$$f^*(\xi) := \mathbf{L}_{\X} f (\xi):=\sup_{x \in \X}\{\xi'x-f(x)\}, \; \textrm{for any } \xi \in f^*(\mathcal{X}).$$
\end{definition}

%\xnote{Scott: Ye, one of the W above should be something else. Can you take a look?}

The function $\xi \mapsto f^*(\xi)$ is a closed convex function (see Lemma~\ref{lemma:LF} in the Appendix) which is also called the convex conjugate of $f$, and the Legendre-Fenchel transform $\mathbf{L}_{\X}$ is also called the conjugate operator.
%Without special announcement, we abbreviate $\mathscr{L}$ for $\mathscr{L}_{V,W}$.
The Legendre-Fenchel transform is a
functional operator that maps any function $f$ to a function of its
family of tangent planes, which is often referred to as the dual function of $f$.

\begin{definition}[$\mathbf{C}$-Operator]\label{def:DFL} For any convex set $\X \subset \mathbb{R}^k$, the double Legendre-Fenchel operator $\mathbf{C} : \ell^{\infty}_{S}(\X) \to \ell^{\infty}_{S}(\X)$ is defined by the repeated application of the Legendre-Fenchel transform twice:
	$$
	\mathbf{C} f := \mathbf{L}_{f^*(\X)} \circ \mathbf{L}_{\X} f.
	$$
	%$\mathscr{C}:=\mathscr{L}_{\mathbb{R}^k, V}\circ
	%\mathscr{L}_{V,\mathbb{R}^k}$ is defined as a mapping from
	%$\mathcal{F}_V$ to $\mathcal{F}_V$.
\end{definition}

%Define the set of bounded convex functions on $\mathcal{X}$ as $C\{\mathcal{X}\}$.

Lemma \ref{lemma:DLF-basic} in the Appendix shows that the double Legendre-Fenchel operator maps any lower semi-continuous function $f$ to its greatest convex minorant, i.e., the largest  function $g \in \ell^{\infty}_C(\X)$ such that $g \leq f$.
The left panel of Figure~\ref{figCM} illustrates this property with a graphical example. We apply the  operator $\mathbf{C}$  to the function $f(x) = [10 \exp(3x/2) - \lfloor 10 x \rfloor - 10]/25$ on $\X = [0,1]$, where $\lfloor \cdot \rfloor$ is the floor function. The convexity-enforced function $\mathbf{C} f$ is the greatest convex minorant of $f$, $\mathbf{C} f(x) = 10[f(\overline{x})(x - \underline{x}) + f(\underline{x})( \overline{x} - x)]$ if $\overline{x} \neq \underline{x}$ and $\mathbf{C} f(x) = f(x)$ otherwise, where $\overline{x} = \lceil 10 x \rceil/10$, $\underline{x} = \lfloor 10 x \rfloor/10$, and $\lceil \cdot \rceil$ is the ceiling function.

%\textcolor{black}{[Ye: please check expression of $\mathbf{C} f(x)$]} 

\begin{figure}[!t]
	\centering
	\includegraphics[width=.49\textwidth]{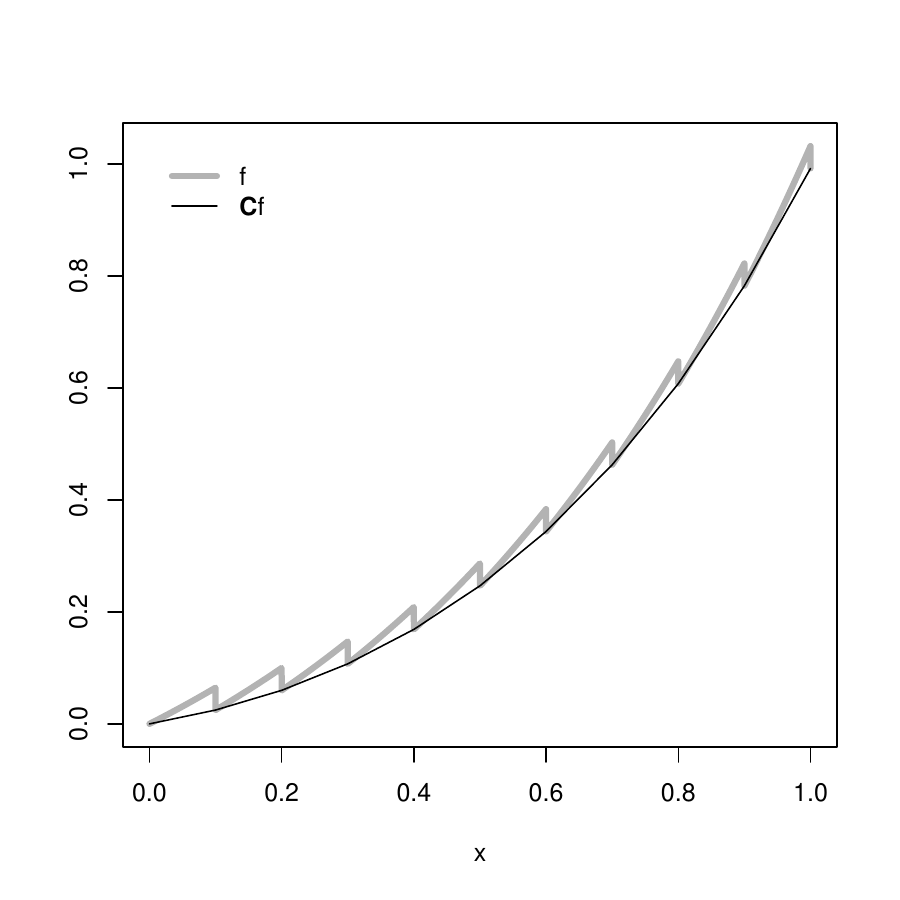}
	\includegraphics[width=.49\textwidth]{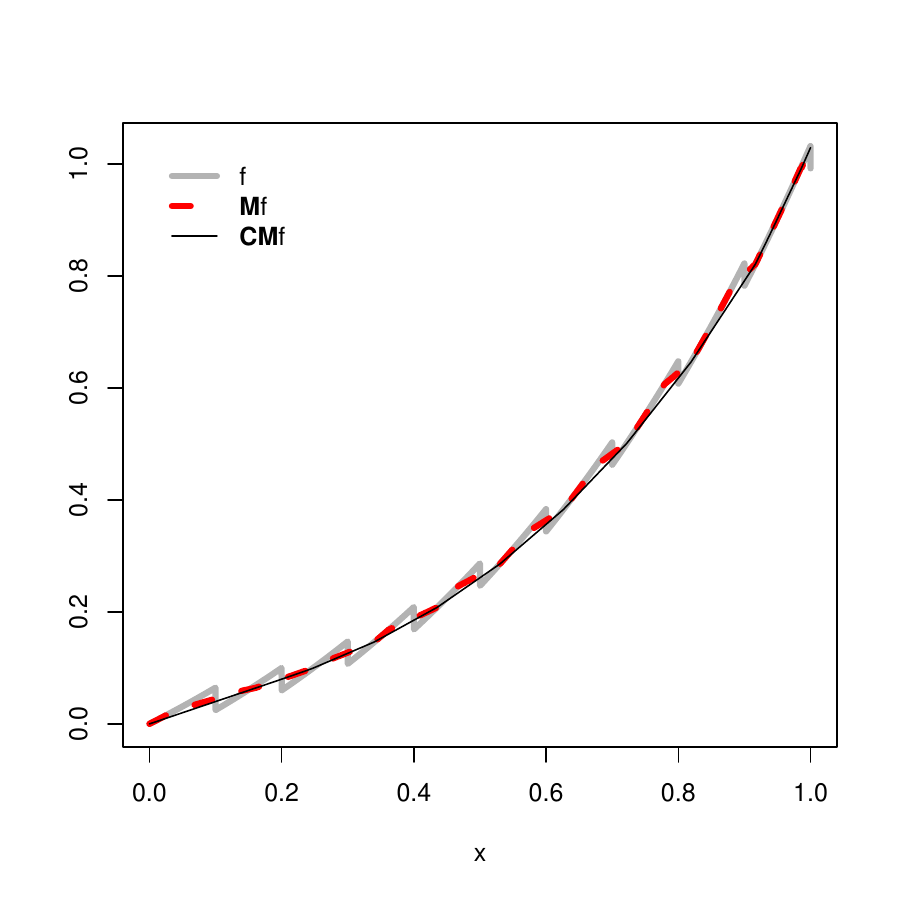}
	\caption{Left: unconstrained and convexity-enforced functions. Right: unconstrained, monotonicity-enforced, and (convexity and monotonicity)-enforced functions. The original function is $f(x) = [10 \exp(3x/2) - \lfloor 10 x \rfloor - 10]/25$.} \label{figCM}
\end{figure}
%\xnote{The line of $Mf$ in the right figure is too light to see. We should use a thicker line?}

The following theorem is an immediate consequence of applying known results from convex analysis.

\begin{theorem}[Convexity-Enforcing Operator]\label{DFL}
	For any convex set $\X$, the operator $\mathbf{C}$ is $\ell^{\infty}_C$-enforcing with respect to $\ell^{\infty}_{S}(\X)$ and a $d_{\infty}$-distance  contraction.
\end{theorem}

\textcolor{black}{
	\begin{remark}[Shifted Convexity-Enforcing Operator]\label{remark:sc} The convexity-enforced function is a minorant of the original function (see fig. \ref{figCM}). When the original function is estimated, the application of the $\mathbf{C}$-operator might introduce downward bias, especially in small samples. A way of reducing bias is by shifting the $\mathbf{C}$-operator, that is
		$$
		\mathbf{SC} f(x) := \mathbf{C} f(x) + \lambda(\X)^{-1}  \int_{\X} [f(x) - \mathbf{C} f(x)] dx,
		$$ 
		where $\lambda(\X)$ is the Lebesgue measure of $\X$. The $\mathbf{SC}$-operator is reshaping and invariant, but does not preserve order nor reduce distance. We compare the $\mathbf{C}$-operator with the $\mathbf{SC}$-operator in the numerical examples of Section \ref{sec:num}.
\end{remark}}

%\textbf{VC: Give Graphical Example}
%\inote{Ye: is $\ell^{\infty}_C(\X) \subset \ell_S^{\infty}(\X)$ when $\X$ is convex? Otherwise we need to change the definition of  $\ell^{\infty}_C(\X)$ to $\ell^{\infty}_C(\X) := \{f \in \ell_S^{\infty}(\mathcal{X}): f(\alpha x + (1-\alpha) x') \leq \alpha f(x) + (1-\alpha) f(x') \textrm{ for all } x,x' \in \X, \alpha \in [0,1] \}$ }

%add empirical version C_n.

\subsection{Monotonicity}\label{M} Let $\ell^{\infty}_M(\X) := \{f \in \ell^{\infty}(\mathcal{X}) : f(x') \leq f(x) \textrm{ for all }  x,x' \in \X \textrm{ such that } x' \leq x \}$, the set of bounded nondecreasing, measurable functions on $\X$.   
We consider the multivariate monotone rearrangement  of \cite{CFG09} as a monotonicity-enforcing operator.

%\xnote{This subsection lacks the motivation/application and some notations are not defined, e.g., $\Lambda$}
%For any function $f\in L^p(\mathcal{X})$ and a probability measure $\mu$ on $\mathcal{X}$, define $\Lambda(y;\mu):=\int_{x\in \mathcal{X}} 1(f(x)\leq y)d \mu(x)$ as the CDF function of the random variable $f(X)$ with $X\sim \mu$.
%
%The rearrangement operator is defined as follows:

\begin{definition}[$\mathbf{M}$-Operator]\label{def:M} For any rectangular set  $\mathcal{X}$ that is regular (i.e., has non-empty interior in $\Bbb{R}^k$), the multivariate increasing rearrangement operator $\mathbf{M}: \ell^{\infty}(\X) \to \ell^{\infty}_M(\X)$ is defined by
	$$
	\mathbf{M} f := \frac{1}{|\Pi|} \sum_{\pi \in \Pi}  \mathbf{M}_{\pi} f,
	$$
	where $\pi = (\pi_1, \ldots, \pi_k)$ is a permutation of the integers $1, \ldots, k$, $\Pi$ is a non-empty subset of all possible permutations $\pi$, and
	$
	\mathbf{M}_{\pi} f :=  \mathbf{M}_{\pi_1} \circ \cdots \circ \mathbf{M}_{\pi_k} f,
	$
	where $\mathbf{M}_{j} f$ is the one-dimensional increasing rearrangement applied to the function $x(j) \mapsto f(x(j), x(-j))$ defined by
	$$
	\mathbf{M}_{j} f (x):= \inf\left\{y \in \RR : \int_{\X(j)} 1\{f(t, x(-j)) \leq y\} dt \geq x(j)  \right\},
	$$
	the one-dimensional increasing rearrangement applied to the function $x(j) \mapsto f(x(j), x(-j))$. Here, we use $f(x(j), x(-j))$ to denote the dependence of $f$ on $x(j)$, the $j^{th}$-component of $x$, and all other arguments, $x(-j)$, and $\X(j)$ to denote the domain of $x(j) \mapsto f(x(j), x(-j))$.
\end{definition}

%\textbf{[CHECK ABOVE. THE DEFINITION PREVIOUSLY WAS MESSED UP][IVAN: LOOKS FINE TO ME]}

%Note that $\Pi$ needs to be only a subset, but not the entire set of all possible permutations $\pi$. 
Proposition 2 of \cite{CFG09} showed that the multivariate increasing rearrangement is monotonicity-enforcing and distance-reducing with respect to $d_p$ for any $p \geq 1$. We state this result as a theorem for the purpose of completeness.

\begin{theorem}[Monotonicity Operator]\label{RA}
	For any regular rectangular set  $\mathcal{X}$, the operator $\mathbf{M}$ is $\ell^{\infty}_M$-enforcing with respect to $\ell^{\infty}(\X)$ and is a $d_{p}$-distance  contraction for any $p \geq 1$.
\end{theorem}

\begin{remark}[Isotonization Operators] Isotonization operators, i.e., projections on the set of weakly increasing functions, can also be considered in place of rearrangement and the results below apply to them. Here we 
	focus on the rearrangement for conciseness. \cite{CFG09} showed, for the one-dimensional case, that  isotonization and convex linear combinations of monotone rearrangement and isotonic regression are also  $\ell^{\infty}_M$-enforcing operators with respect to $\ell^{\infty}(\X)$ and $d_{p}$-distance  contractions for any $p \geq 1$. Extension to the multivariate case follows analogously  to \cite{CFG09} by an induction argument.

\end{remark}

\textcolor{black}{
	\begin{remark}[Multivariate Distributions] Multivariate distribution functions satisfy stronger shape restrictions than monotonicity. For example, in the bivariate case they are   2-increasing (supermodular) and grounded  \cite{nelsen07}. The grounded restriction can be enforced using a simple variation of the range operator. We are not aware of any operator that enforces supermodularity. 
	\end{remark}
}

\subsection{Convexity and Monotonicity}
Let  $\ell^{\infty}_{CM}(\X) := \ell^{\infty}_C(\X) \cap \ell^{\infty}_M(\X)$ be the set of bounded convex and nondecreasing functions on $\X$. We consider the composition of the $\mathbf{C}$ and $\mathbf{M}$ operators to enforce both convexity and monotonicity.

\begin{definition}[$\mathbf{CM}$-Operator] For any regular rectangular set  $\mathcal{X}$, the convex rearrangement operator $\mathbf{CM}: \ell^{\infty}_{S}(\X) \to \ell^{\infty}_{S}(\X)$ is defined by
	$$
	\mathbf{CM} f := \mathbf{C} \circ \mathbf{M} f.
	$$
\end{definition}

%\inote{Ye: is it guaranteed that $\mathbf{M} f \in \ell^{\infty}_{S}(\X)$? - For the rearragement operator, I think the answer is yes (by the assmption about - perhaps some assumption is required on the measure of $x$ that is applied in the rearrangement). -  but do we need a proof?}
%
%\inote{Ye: include in the next remark an intuitive explanation on why we need to add the assumption that the domain is a regular rectangle ||| Ye: added in the next paragraph}

\begin{remark}[Rectangular Domain]\label{RM} 
	%We have found examples where the operator $\mathbf{CM}$ is not $\ell^{\infty}_{CM}$-enforcing when $\X$ is not a rectangular set. The source of the problem is that  
	The $\mathbf{C}$-operator does not preserve monotonicity in general.\footnote{Let $\X \subset \mathbb{R}^2$ be a triangular set with vertices at $(-1,3)$, $(0,0)$ and $(3,-1)$. Then, the function $f(x) = 1+3(x_1+x_2)/2-|x_1-x_2|$ is increasing on $\X$, but its greatest convex minorant $\mathbf{C}f(x) = 1 - (x_1+x_2)/2$ is decreasing on $\X$.}  
	When $\X$ is a regular rectangle,   the $\mathbf{C}$-operator can be obtained by separate application to each face of the rectangle and does not affect the monotonicity of the function;  see Lemma~\ref{remark:Drectagle}  in the Appendix. 
	From a practical point of view, we do not find this assumption very restrictive because the domains usually have the product form $\X = [a_1,b_1] \times \cdots \times [a_k,b_k]$ in applications. If the domain of the target function is not rectangular, we can either restrict the analysis to a rectangular subset of the domain or  extend the function to a rectangular set that contains the domain.
\end{remark}

In the right panel of Figure~\ref{figCM}, we apply the  operators  $\mathbf{M}$ and $\mathbf{CM}$  to the function $f(x) = [10 \exp(3x/2) - \lfloor 10 x \rfloor - 10]/25$ on $\X = [0,1]$. The monotonicity-enforced function $\mathbf{M} f$ in dashed line is not convex, whereas the (convexity and monotonicity)-enforced function $\mathbf{CM} f$ is both monotone and convex. Indeed, $\mathbf{CM} f$ is the greatest convex minorant of $\mathbf{M} f$. 

\begin{theorem}[Convexity and Monotonicity Operator]\label{composition}
	For any regular rectangular set $\X$, the operator $\mathbf{CM}$ is  $\ell^{\infty}_{CM}$-enforcing with respect to $\ell^{\infty}_{S}(\X)$ and a $d_{\infty}$-distance  contraction.
\end{theorem}

\begin{remark}[Proof of Theorem~\ref{composition} and ordering of composition] The proof of Theorem~\ref{composition} does not follow from combining Theorems \ref{DFL} and~\ref{RA}. As indicated in Remark~\ref{RM}, the argument is more subtle as we need to verify that the application of the $\mathbf{C}$-operator preserves monotonicity. Moreover, the 
	order of the composition of the operators matters. Thus, $\mathbf{MC}:=\mathbf{M}\circ\mathbf{C}$ is not $\ell^{\infty}_{CM}$-enforcing  because the operator $\mathbf{M}$ does not preserve convexity in general.
\end{remark}

\begin{remark}[Concavity and  Monotonicity]\label{remark:concavity} Using the notation for inverse operators given in the introduction, we can construct composite operators for all the combinations of concavity/convexity and increasing/decreasing monotonicity restrictions. Thus, the operator $\mathbf{CM^-} f = \mathbf{C}[-\mathbf{M}(-f)]$ enforces convexity and decreasing monotonicity, $\mathbf{C^-M} f = - \mathbf{C}[-\mathbf{M}(f)]$ enforces concavity and increasing monotonicity, and $\mathbf{C^-M^-} f = - \mathbf{C}[\mathbf{M}(-f)]$ enforces concavity and decreasing monotonicity. It can be shown that these operators satisfy analogous properties to $\mathbf{CM}$ by a straightforward modification of the proof of Theorem~\ref{composition}.
\end{remark}

\subsection{Quasi-convexity}\label{qc} \textcolor{black}{Quasi-convexity is a global property of a function, which is weaker than convexity. A convex function must be quasi-convex, but a quasi-convex function is not necessarily convex. Intuitively, a function $f$ defined on a convex domain is quasi-convex if and only if all its level sets are convex. Please also see its definition in Eq.~\eqref{eq:quasi_def} below. Quasi-convexity is commonly used in economics because it is an ordinal property, preserved by monotone transformations, which represents well economic relationships such as  utility and production functions \cite{g04,Koenker:10,c17}. Moreover, quasi-convex functions have good optimization properties.}

Consider the set of bounded lower semi-continuous quasi-convex functions on $\X$: 
\begin{equation}\label{eq:quasi_def}
	\ell^{\infty}_Q(\X) := \{f \in \ell_S^{\infty}(\mathcal{X}): f(\alpha x + (1-\alpha) x') \leq \max\{f(x),f(x')\} \textrm{ for all } x,x' \in \X, \alpha \in [0,1] \}.
\end{equation}
We note that $\ell^{\infty}_C(\X) \subset \ell^{\infty}_Q(\X)$, and that for any $f \in \ell^{\infty}_Q(\X)$, the lower contour sets, $\mathcal{I}_f(y):=\{x \in \X : f(x) \leq y \}$,  are convex for all $y \in \RR$. For any set $\mathcal{Z} \subset \RR^k$, let $\mathrm{conv}(\mathcal{Z})$ denote the convex hull of $\mathcal{Z}$.  We consider the following new operator to impose quasi-convexity:

\begin{definition}[$\mathbf{Q}$-Operator]
	For any convex  and compact set $\X \subset \RR^k$, the quasi-convexity operator $\mathbf{Q}: \ell^{\infty}_S(\X) \to \ell_S^{\infty}(\X)$ is defined by
	\begin{equation}\label{def:Q}
		\mathbf{Q} f(x) := \min \left\{ y \in \RR : x \in \mathrm{conv}[\mathcal{I}_f(y)] \right\}.
	\end{equation}
	%where $\mathrm{conv}(A)$ denote the convex hull of points of set $A \subset \Bbb{R}$.
\end{definition}

%\textbf{HOW DO WE KNOW MIN IS ACTUALLY ATTAINED? IS THIS WELL-DEFINED
%FOR EACH VALUE OF $x$. SHOULD $f$ be constrained to be lower-semi-continuous? Please check.}

\begin{remark}[Existence of $\mathbf{Q}$-Operator] The restriction of the operator $\mathbf{Q}$ to $ \ell^{\infty}_S(\X)$, where $\X$ is convex and compact,  guarantees that the minimum in \eqref{def:Q} exists  (see Lemma~\ref{lemma:Q-def} in the Appendix). 
	When the set $\mathcal{X}$ is non-compact or the function $f\notin \ell^\infty_S(\mathcal{X})$, there exist counter examples such that the minimum in \eqref{def:Q} does not exists.\footnote{\textcolor{black}{For example, the function $f(x) = x[1(x \leq 0)/(2+x) - 1(x > 0)/(2-x)]$ if $|x| \neq 1$ and $f(1) = f(-1) = -1$ is not lower-semicontinuous on $\mathcal{X}=[-1,1]$. In this case,  for any $x \in (-1,1)$, $ \{y : x \in \mathrm{conv}[\mathcal{I}_f(y)] \} = (-1,+\infty)$, so that  $\mathbf{Q}(x)$ is not well-defined. The same problem arises if $\mathcal{X}:=(-1,1)$, i.e., the domain $\mathcal{X}$ is not compact.}} In such cases, one might still define $\mathbf{Q} f(x) := \inf \left\{ y \in \RR : x \in \mathrm{conv}[\mathcal{I}_f(y)] \right\}$, but this operator appears to lose the contraction property stated below. 
\end{remark}

%Let the domain of $x$, denoted as $\mathcal{X} \subseteq \mathbb{R}^k$,  be compact and convex. Define the set of bounded quasi-convex functions on $\mathcal{X}$ as $\mathcal{F}_Q$. For any function $f\in \mathcal{F}$, we consider the following operator that translates $f$ into a quasi-convex function $\in \mathcal{F}_Q$:
%
%Define the level set $\mathcal{I}_f(y):=\{x \in \mathcal{X}|f(x)\leq y\}$. For any set $Z\subset \mathbb{R}^k$, define $\mathscr{C}_0(Z)$ as the convex hull of $Z$. It is easy to see that for any $y_1 \leq y_2$, we have $\mathcal{I}_f(y_1) \subseteq \mathcal{I}_f(y_2)$ and $\mathscr{C}_0(\mathcal{I}_f(y_1)) \subseteq \mathscr{C}_0(\mathcal{I}_f(y_2))$. We define an operator $\mathscr{Q}(f)$ as following:
%\begin{equation}\label{eq:Q}
%\mathscr{Q}(f)(x):=\min_{y}\{x\in \mathscr{C}_0(\mathcal{I}_f(y))\}.
%\end{equation}
The operator $\mathbf{Q}$ transforms any bounded lower semi-continuous function into a quasi-convex  function. 
To see this, recall that a function is quasi-convex if its domain and all its lower contour sets are convex.
By construction,  $x \in \mathrm{conv}[\mathcal{I}_f(y)] $ if and only if $\mathbf{Q} f(x) \leq y$. Therefore, the lower contour  set of $\mathbf{Q} f$ at any level $y \in \RR$ is $\mathcal{I}_{\mathbf{Q} f}(y) = \{x \in \mathcal{X} : \mathbf{Q} f(x)\leq y\} = \mathrm{conv}[\mathcal{I}_f(y)]$, which is a convex set.

The left panel of Figure~\ref{figCMQ} shows a graphical example. We apply the operator $\mathbf{Q}$ to the function $f(x) = [10 \exp(3x/2) - \lfloor 10 x \rfloor - 10]/25$ on $\X = [0,1]$. Here we can see that the function $\mathbf{Q} f$ is the greatest quasi-convex minorant of $f$, $\mathbf{Q} f(x) = \min\{f(x), f(\lceil 10 x \rceil /10)\} $.

\begin{figure}[!t]
	\centering
	\includegraphics[width=.49\textwidth]{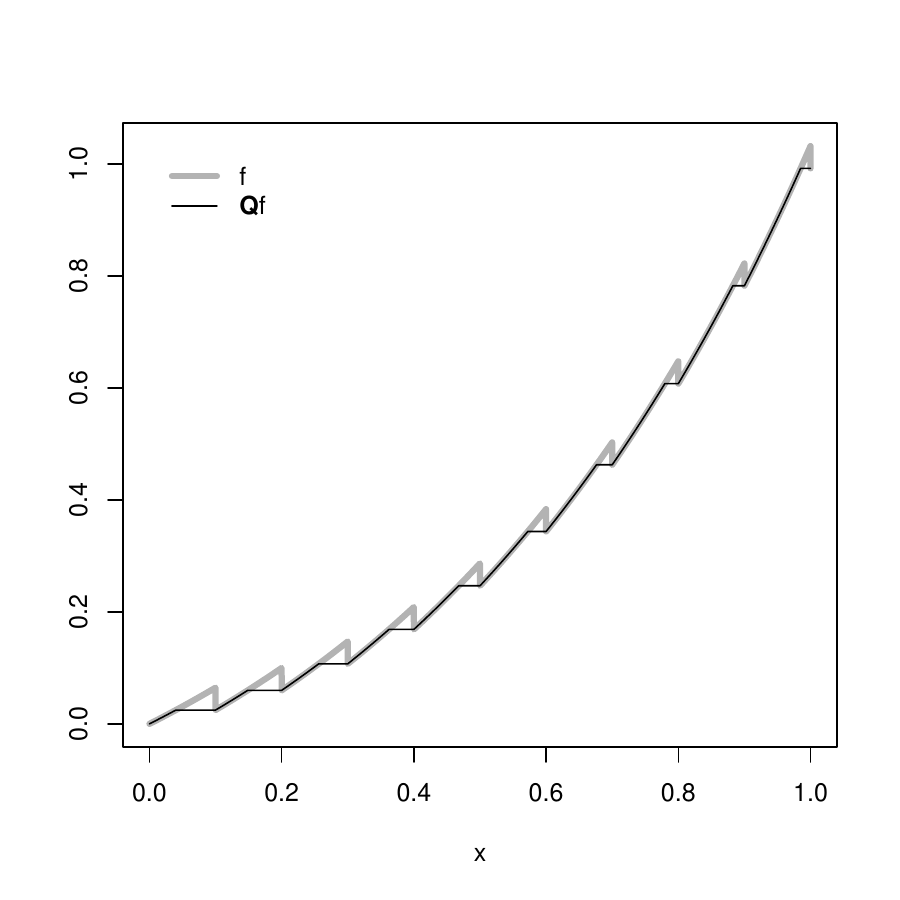}
	\includegraphics[width=.49\textwidth]{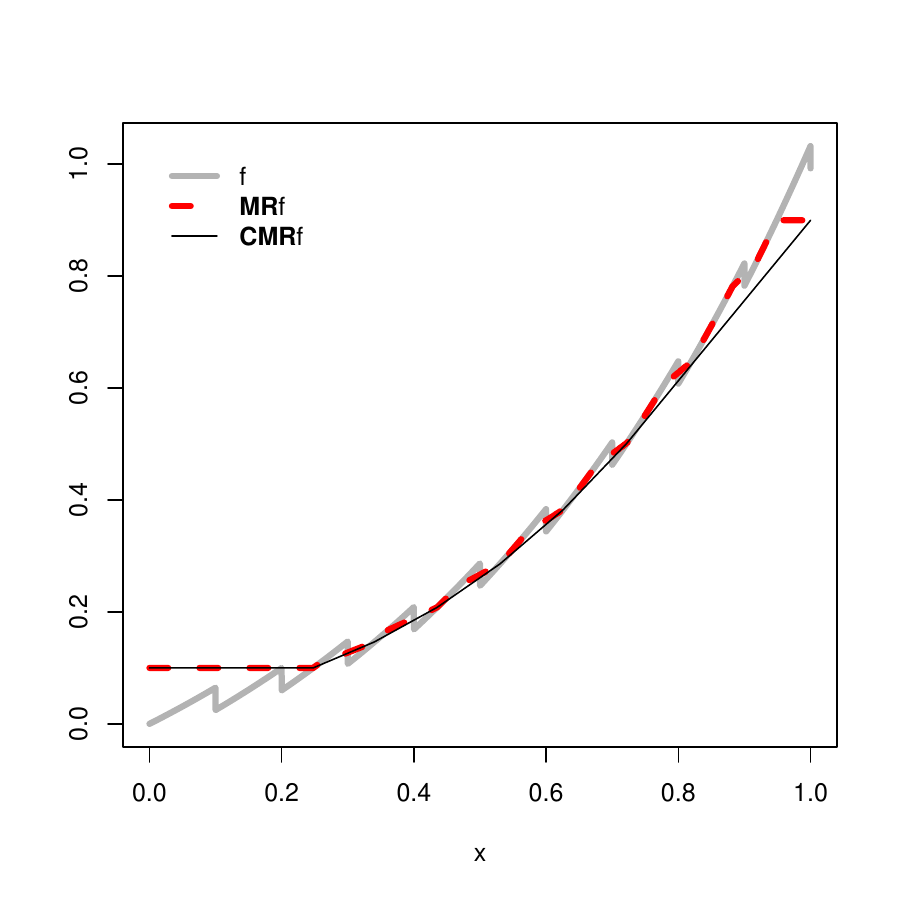}
	\caption{Left: unconstrained and quasiconvexity-enforced functions. Right: unconstrained and (convexity, monotonicity and range)-enforced functions. The original function is $f(x) = [10 \exp(3x/2) - \lfloor 10 x \rfloor - 10]/25$.} \label{figCMQ}
\end{figure}
%\xnote{similarly, it is bit difficult to see $MRf$ line in right of Figure 2}

%For any $f, g \in \ell_S^{\infty}(\mathcal{X})$, define
%\begin{equation}\label{eq:dp}
%d_{H,p}(f,g)= \left\{ \int_{\RR} d_H(\mathcal{I}_f(y), \; \mathcal{I}_g(y))^p dy \right\}^{1/p}, \ \ d_{H,\infty}(f,g) := \sup_{y\in \RR} d_H(\mathcal{I}_f(y), \; \mathcal{I}_g(y)),
%\end{equation}
%where $d_H(\mathcal{V},\mathcal{W})$ is the Hausdorff distance between the sets $\mathcal{V}$ and $\mathcal{W}$.

% and $\|\cdot\|_p$ denotes the functional $L_p$-norm (noting that $d_H(\mathcal{I}_f(y), \; \mathcal{I}_g(y))$ is a function in $y \in \mathbb{R}$).

\begin{theorem}[Quasi-Convexity Operator]\label{QC}
	For any convex and compact set $\mathcal{X}$, the operator $\mathbf{Q}$ is $\ell^{\infty}_Q$-enforcing with respect to $\ell_S^{\infty}(\X)$ and a $d_\infty$-distance contraction. 
	%(a) $d_{H,p}$-distance  contraction for any $p \geq 1$.
	%
	%(b) $d_\infty$-distance (e.g., the sup-norm of the difference between two functions) contraction.
	%the  is convex. The operator $\mathscr{Q}$ is $\mathcal{F}_Q$-enforcing, or ``quasi-convex'' enforcing, with respect to $L^{\infty}(\mathcal{X})$. This operator is distance contraction with respect to $d_p(\cdot,\cdot)$ for $p>0$.
\end{theorem}
%\xnote{Please confirm if it is $d_p(\cdot,\cdot)$ for $p\geq 0$ or $d_\infty$?}

%\inote{Ye: Do we need $\X$ to be compact?}
%
%\inote{Ivan: if $f\in l^\infty(\X)$, it does not need to be compact, and it does not affect the proofs if we remove the "compact" from the statement - so we should remove them in Lemma 5 and Definition 8.}

%\inote{Ye: is $\ell^{\infty}_Q(\X) \subset \ell_S^{\infty}(\X)$ when $\X$ is convex and compact? Otherwise we need to change the definition of $\ell^{\infty}_Q(\X)$ to $\ell^{\infty}_Q(\X) := \{f \in \ell_S^{\infty}(\mathcal{X}): f(\alpha x + (1-\alpha) x') \leq \min[f(x),f(x')] \textrm{ for all } x,x' \in \X, \alpha \in [0,1] \}$ }

\textcolor{black}{
	\begin{remark}[Shifted Quasi-Convexity-Enforcing Operator]\label{remark:sq} By similar reasons to  Remark \ref{remark:sc}, we introduce the shifted quasi-convexity-enforcing operator:
		$$
		\mathbf{SQ} f(x) := \mathbf{Q} f(x) + \lambda(\X)^{-1}  \int_{\X} [f(x) - \mathbf{Q} f(x)] dx,
		$$ 
		where $\lambda(\X)$ is the Lebesgue measure of $\X$. 
\end{remark}}

%\begin{remark}[Remark on $d_{H,p}$]
%The discrepancy function $d_{H,p}(f,g)$ measures the distance of the lower contour sets of the functions $f$ and $g$. For a sequence of functions $f_n$ that converges to $f$ uniformly over all $x\in \mathcal{X}$ as $n \to \infty$, the sequence of lower contour sets $\mathcal{I}_{f_n}(y)$ converges to $\mathcal{I}_f(y)$ uniformly over all possible values $y$ in the range of $f$. If $y$ is a regular value of $x \mapsto f(x)$, i.e., $\nabla f(x)\neq 0$ for all $x \in \{\tilde x \in \X: f(\tilde x)=y\}$,  then $\partial \mathcal{I}_f(y)$, the boundary of $\mathcal{I}_f(y)$,  is a $(k-1)$-dimensional manifold, and as $d_{\infty}(f_n, f) \to 0$,
%$$
%d_{H}(\mathcal{I}_{f_n}(y), \mathcal{I}_f(y)) \sim \sup_{x\in \partial \mathcal{I}_f(y) } \frac{|f_n(x)-f(x)|}{\|\nabla f(x)\|}.
%$$
%Therefore, when  $\|\nabla f(x)\| \neq 0$  for all $x\in \mathcal{X}$, as $d_{\infty}(f_n, f) \to 0$,
%$$
%d_{H,\infty}(f_n,f)  \sim \sup_{x\in \mathcal{X}}\frac{|f_n(x)-f(x)|}{\|\nabla f(x)\|},$$
%which is a weighted $d_{\infty}$-distance between $f$ and $f_n$.
%\end{remark}

%\textbf{NB. What does arrow $\to$ mean above?  $a_n \to b_n$ is ambigiuous and could
%mean $a_n - b_n \to 0$ or $a_n/b_n \to 1$. Could someone fix this appropriately here and in the proofs? The math needs to be rigorous. [IVAN: I HAVE REPLACED $\to$ BY THE OPERATOR $\sim$ DEFINED IN THE INTRODUCTION. YE - PLEASE CHECK]}

\subsection{Quasi-Convexity and Monotonicity}
Let  $\ell^{\infty}_{QM}(\X) := \ell^{\infty}_Q(\X) \cap \ell^{\infty}_M(\X)$ be the set of bounded quasi-convex and partially nondecreasing functions on $\X$. This case is only relevant when $k>1$ because univariate monotone functions are quasi-convex. We consider the composition of the $\mathbf{Q}$ and $\mathbf{M}$ operators to impose both quasi-convexity and monotonicity.

\begin{definition}[$\mathbf{QM}$-Operator] For any regular rectangular set $\X$, the quasi-convex rearrangement operator $\mathbf{QM}: \ell^{\infty}_S(\X) \mapsto \ell^{\infty}_{S}(\X)$ is defined by
	$$
	\mathbf{QM} f := \mathbf{Q} \circ \mathbf{M} f.
	$$
\end{definition}

\begin{theorem}[Quasi-Convexity and Monotonicity Operator]\label{composition2}
	For any regular rectangular set  $\mathcal{X}$, the operator $\mathbf{QM}$ is  $\ell^{\infty}_{QM}$-enforcing with respect to $\ell_S^{\infty}(\X)$ and a $d_\infty$-distance contraction.
\end{theorem}

The comments and example in Remark~\ref{RM} also apply to the $\mathbf{QM}$-operator. Thus, the assumption that $\X$ is a rectangular set is sufficient to guarantee that the  $\mathbf{Q}$-operator preserves monotonicity.

\begin{remark}[Quasi-Concavity and  Monotonicity] Similar to Remark~\ref{remark:concavity},  we can construct composite operators for all the combinations of quasi-concavity/quasi-convexity and increasing/decreasing monotonicity restrictions. Thus, the operator $\mathbf{QM^-} f = \mathbf{Q}[-\mathbf{M}(-f)]$ enforces quasi-convexity and decreasing monotonicity, $\mathbf{Q^-M} f = - \mathbf{Q}[-\mathbf{M}(f)]$ enforces quasi-concavity and increasing monotonicity, and $\mathbf{Q^-M^-} f = - \mathbf{Q}[\mathbf{M}(-f)]$ enforces quasi-concavity and decreasing monotonicity. It can be shown that these operators satisfy analogous properties to $\mathbf{QM}$ by a straightforward modification of the proof of Theorem~\ref{composition2}.
	%
	%It is also important to have the assumption that $\mathcal{X}$ is a regular rectangular in Theorem~\ref{composition2} in order to have the statement of the lemma to be valid, and the proof of Theorem~\ref{composition2} is non-trivial. When $\mathcal{X}$ is not a regular rectangular, there exists counter examples for $\mathbf{QM}$ to be non-shape preserving. Example \ref{domain} is one of such counter examples, if we replace $\mathbf{C}$ with $\mathbf{M}$ in the example.
	%
	%Take the function $f$ in Example \ref{domain}, where $f$ is a monotonically increasing function. By definition of $\mathbf{Q}$ operator in (\ref{def:Q}),  for $x^*=(1,1)$, it is easy to see that $\min_{y\in \mathbb{R}}\{x^*\in conv[\mathcal{I}_f(y)]\} = 0$, since $f\geq 0$ for all $x\in \mathcal{X}$, $\mathcal{I}_f(0)=\{A,C\}$, and $conv[\mathcal{I}_f(y)]=$ segment $AC\supset \{x^*\}$.
	%
	%It is easy to see that $\mathbf{Q}f(B) = f(B)=1$, as for any $y<1$, $B\notin conv[\mathcal{I}_f(y)]$. Hence, $\mathbf{Q}f(B)=1>\mathbf{Q}f(x^*)=0$. That said, $\mathbf{Q}f$ is not a monotonically increasing function.
\end{remark}

\subsection{Range and Other Shape Restrictions} The following theorem shows that the operator $\mathbf{R}$ can be composed with $\mathbf{C}$, $\mathbf{M}$ and $\mathbf{Q}$ to produce range-constrained convex, monotone or quasi-convex functions.  Let $\ell^{\infty}_{CR}(\X) := \ell^{\infty}_{C}(\X) \cap \ell^{\infty}_{R}(\X)$, $\ell^{\infty}_{MR}(\X) := \ell^{\infty}_{M}(\X) \cap \ell^{\infty}_{R}(\X)$, $\ell^{\infty}_{QR}(\X) := \ell^{\infty}_{Q}(\X) \cap \ell^{\infty}_{R}(\X)$, $\mathbf{CR}  := \mathbf{C} \circ \mathbf{R}$, $\mathbf{MR} := \mathbf{M} \circ \mathbf{R}$, and $\mathbf{QR} := \mathbf{Q} \circ \mathbf{R}$.

\begin{theorem}[Composition with Range Operator]\label{cRR}
	(i) For any convex set $\X$, the operator $\mathbf{CR}$ is  $\ell^{\infty}_{CR}$-enforcing with respect to $\ell^{\infty}_{S}(\X)$ and a $d_{\infty}$-distance  contraction; (ii) for any regular rectangular set $\X$, the operator $\mathbf{MR}$  is $\ell^{\infty}_{MR}$-enforcing with respect to $\ell^{\infty}(\X)$ and a $d_{p}$-distance  contraction for any $p \geq 1$; and (iii) for any convex and compact set $\mathcal{X}$, the operator $\mathbf{QR}$ is   $\ell^{\infty}_{QR}$-enforcing with respect to $\ell_S^{\infty}(\X)$ and a $d_{\infty}$-distance contraction.
\end{theorem}

The operator $\mathbf{MR}$ can be composed with $\mathbf{C}$ and $\mathbf{Q}$ to produce range-constrained monotone convex or quasi-convex functions. The properties of the resulting operators $\mathbf{CMR}  := \mathbf{C} \circ \mathbf{MR}$ and $\mathbf{QMR}  := \mathbf{Q} \circ \mathbf{MR}$ follow from combining Theorem~\ref{cRR} with Theorems~\ref{composition} and~\ref{composition2}, respectively. Let $\ell^{\infty}_{CMR}(\X) := \ell^{\infty}_{CM}(\X) \cap \ell^{\infty}_{R}(\X)$ and $\ell^{\infty}_{QMR}(\X) := \ell^{\infty}_{QM}(\X) \cap \ell^{\infty}_{R}(\X)$.

%\xnote{It might be better to provide the full statement of the properties, e.g., contraction with respect to which distance?}

\begin{corollary}[Composition with Range and Monotonicity Operators]
	(i) For any regular rectangular set $\X$, the operator $\mathbf{CMR}$ is  $\ell^{\infty}_{CMR}$-enforcing with respect to $\ell^{\infty}_{S}(\X)$  and a $d_{\infty}$-distance  contraction; and (ii) for any regular rectangular set $\X$, the operator $\mathbf{QMR}$ is  $\ell^{\infty}_{QMR}$-enforcing with respect to $\ell_S^{\infty}(\X)$  and a $d_{\infty}$-distance contraction.
\end{corollary}

%\begin{figure}[!htb]
%\centering
%\includegraphics[width=.49\textwidth]{../Analytical_example/figure-CMRf-ex3.pdf}
%\caption{Unconstrained and (convexity, monotonicity and range)-enforced functions. The original function is $f(x) = [10 \exp(3x/2) - \lfloor 10 x \rfloor]/4$.} \label{figCMR}
%\end{figure}

In the right panel of Figure~\ref{figCMQ}, we apply the  operators  $\mathbf{MR}$ and $\mathbf{CMR}$  to the function $f(x) = [10 \exp(3x/2) - \lfloor 10 x \rfloor - 10]/25$ on $\X = [0,1]$. We enforce that the range be in the interval $[0.1,0.9]$.  The (monotonicity and range)-enforced function $\mathbf{MR} f$ in dashed line satisfies the monotonicity and range restrictions but is not convex. The (convexity, monotonicity and range)-enforced function $\mathbf{CMR} f$ satisfies the three shape restrictions.

\subsection{Shape Restrictions on Transformations} The shape operators can be combined with other functions to enforce shape restrictions on transformations of the function $f$. An example is log-concavity where we assume that $\log f$ is concave.\footnote{\cite{bagnoli2005log}  discussed applications of log-concavity to economics  and statistics, and analyzed the log-concavity properties of common distributions. } Let $h$ be a real-valued bijection with inverse function $h^{-1}$. We consider the operator $\mathbf{O}_{h}$ that applies the operator $\mathbf{O}$ to the transformation $h(f)$ and then recovers the shape-constrained version of $f$ by inversion, that is
\[
\mathbf{O}_{h}f=h^{-1}\circ\mathbf{O}(h\circ f).
\]
For example, if $f$ is log-concave, then $h(x) = \log x$ and
$$
\mathbf{O}_{h}f= \mathbf{C}^{-}_{\log}f = \exp[\mathbf{C}^{-}(\log f)].
$$

The following theorem gives conditions under which the transformations $h$ and $h^{-1}$ preserve the properties of the operator $\mathbf{O}$. Define $\ell_{h,j}^{\infty}(\mathcal{X})=\{ f\in\ell_{0}^{\infty}(\mathcal{X}) : h \circ f\in\ell_{j}^{\infty}(\mathcal{X})\} $ for $j \in \{0,1\}$.

\begin{theorem}[Properties of  $\mathbf{O}_{h}$-Operator]\label{transformation} Let  $\mathbf{O}$ be a $\ell_1^\infty(\X)$-enforcing operator with respect to $\ell_{0}^{\infty}(\mathcal{X})$, and $y \mapsto h(y)$ be a real valued strictly monotonic bijection on the domain $\mathcal{Y} \subset \mathbb{R}$. Then,  $\mathbf{O}_{h}$ is a $\ell_{h,1}$-enforcing operator with respect to $\ell_{h,0}^{\infty}(\mathcal{X})$. Moreover, if $\mathbf{O}$ is a $\rho$-distance contraction, then  $\mathbf{O}_h$ is a $\rho_h$-distance contraction for $\rho_h(f,g) := \rho(h\circ f, h\circ g)$.
	%Moreover, if $\mathbf{O}$ is a $\rho$-distance contraction, $y \mapsto h(y)$ is Lipschitz continuous on $\mathcal{Y}$   with Lipschitz constant $C_{\mathcal{Y}}$,  $z \mapsto h^{-1}(z)$ is Lipschitz continuous on $\mathcal{Z}$ with Lipschitz constant $C_{\mathcal{Z}}$, and $C_{\mathcal{Y}} C_{\mathcal{Z}} \leq 1$, then  $\mathbf{O}_{h}$  is a $\rho$-distance contraction. 
\end{theorem}

%
%Lemma: Suppose that $\mathbf{O}$ is $\ell_{1}^{\infty}$-enforcing
%with respect to $\ell_{0}^{\infty}(\mathcal{X})$, and that $\ell_{2}^{\infty}(\mathcal{X})\subset\ell_{0}^{\infty}(\mathcal{X})$.
%For any $h:\mathcal{X}\rightarrow\mathbf{\mathbb{R}}$ with a single-valued
%inverse $h^{-1}$ that has the property that $h^{-1}\circ g\in\ell_{2}^{\infty}(\mathcal{X})$
%for any $g\in\ell_{1}^{\infty}(\mathcal{X})$, define the operator
%$\mathbf{O}_{h}$ such that
%\[
%\mathbf{O}_{h}f=h^{-1}\circ\mathbf{O}(h\circ f).
%\]
%Then, $\mathbf{O}_{h}$ is $\ell_{2}$-enforcing with respect to $\ell_{0}^{\infty}(\mathcal{X})$.
%
%Corollary: Define $\ell_{\log}^{\infty}(\mathcal{X})=\left\{ f\in\ell_{0}^{\infty}(\mathcal{X})|f>0,\log\circ f\in\ell_{C}^{\infty}(\mathcal{X})\right\} $.
%Then, the operator $\mathbf{C}_{\log}$ is $\ell_{\log}^{\infty}(\mathcal{X})$-enforcing.

\subsection{Other Ways of Generating Shape Enforcing Operators} 
When $\ell_0^\infty(\X)$ is a Hilbert space \textcolor{black}{and $\ell_1^\infty(\X)$ is a closed set in the $L^2$ metric}, it is possible to construct generic shape-enforcing operators via $L^2$-projection in $\ell_1^\infty(\X)$:
\begin{definition}[$\mathbf{P}$-Operator]  The $L^2$-projection operator on the Hilbert space $\ell_0^\infty(\X)$, 
	$\mathbf{P} : \ell_0^{\infty}(\mathcal{X}) \to \ell_0^{\infty}(\mathcal{X})$, is defined by
	\begin{equation}
		\mathbf{P}f(x) := \arg\min_{g \in \ell_1^\infty(\X)} \|f - g \|_2.
	\end{equation}
\end{definition}

The $\mathbf{P}$-operator involves an infinite dimensional optimization program that can be computationally challenging except for special cases. For example, when $\ell_1^{\infty}(\mathcal{X}) = \ell_M^{\infty}(\mathcal{X})$, the $\mathbf{P}$-operator corresponds to the isotonization operator that can be computed using the pool adjacent violators algorithm described in \cite{barlowstatistical}.  The following result, discussed on p.45 of \cite{chetverikov2017econometrics} and stated here as a theorem for the purpose of completeness, shows  that $\mathbf{P}$ is range-enforcing and distance-reducing with respect to $d_2$ under some conditions on $\ell_1^\infty(\X)$.

\begin{theorem}[$L^2$-Projection Operator]\label{PP}
	If $\ell_0^\infty(\X)$ is a Hilbert space, $\ell_1^\infty(\X)$ is a closed and convex set, and for any $f_1, f_2 \in \ell_1^\infty(\X)$ the pointwise maximum and minimum of $f_1$ and $f_2$ belongs to $\ell_1^\infty(\X)$, then the operator $\mathbf{P}$ is $\ell^{\infty}_{1}$-enforcing with respect to $\ell_0^{\infty}(\X)$ and a  $d_2$-distance contraction.
\end{theorem}

\textcolor{black}{
	The condition that $\ell_0^\infty(\X)$ is a Hilbert space is satisfied when $\X$ is bounded.  The sets $\overline{\ell}_C^\infty(\X)$, $\overline{\ell}_M^\infty(\X)$, $\overline{\ell}_Q^\infty(\X)$ and their intersections are convex and closed.\footnote{\textcolor{black}{The set $\overline{\ell}_O^\infty(\X)$ denotes the intersection of $\ell_O^\infty(\X)$ with the set of uniformly bounded functions, for $O \in \{C,M,Q\}$. The intersection ensures that $\overline{\ell}_O^\infty(\X)$  is closed in the $L^2$ metric. Alternatively, we can ensure that $\ell_O^\infty(\X)$ is closed in the $L^2$ metric by restricting  $\X$ to be a countable set.}} The condition on the maximum and minimum is a more restrictive Hilbert lattice property. It  is satisfied by  $\overline{\ell}_M^\infty(\X)$, but not by $\overline{\ell}_C^\infty(\X)$ and $\overline{\ell}_Q^\infty(\X)$. Theorem~\ref{PP} therefore covers the isotonization operator, but not the convex and quasi-convex projections. }

\begin{remark}[Proof of Theorem \ref{PP}]  \cite{chetverikov2017econometrics}  referred to Lemma 2.4 in \cite{no2012} for the order-preserving  and to Lemma 46.5.4 in \cite{zeidler1984} for the distance contraction. Reshaping and invariance hold trivially by the definition of the operator.
\end{remark}

\medskip

\section{Improved Point and Interval Estimation} \label{sec:inf}

We show how to use shape-enforcing operators to improve point and interval estimators of a shape-constrained function.  Let $f_0 : \X \to \RR$ be the target function, which is known to satisfy a shape restriction, i.e.,  $f_0 \in \ell_1^{\infty}(\X)$. Assume we have a point estimator $f$ of $f_0$, and an interval estimator or uniform confidence band $[f_l,f_u]$ for  $f_0$. These estimators are unconstrained and therefore do not necessarily satisfy the shape restrictions, i.e.,  $f,f_l,f_u \in \ell^{\infty}_0(\X)$ but $f,f_l,f_u \not\in \ell^{\infty}_1(\X)$ in general. 

There are many different ways to obtain these initial estimators, ranging from parametric to modern adaptive nonparametric methods \textcolor{black}{\citep[e.g., ][]{fan96,li07,hastie09}. These methods can be tailored to properties of the target function such as smoothness or sparsity.}  A common frequentist confidence band for the function $f_0$ is constructed as
$$
f_l(x) = f(x) - c_p s(x), \ \ f_u(x) = f(x) + c_p s(x),
$$
where $s(x)$ is the standard error of $f(x)$ and $c_p$ is a critical value chosen such that
$$
\Pr(f_0 \in [f_l,f_u]) \geq p,
$$
for some confidence level $p$, where  event $f_0 \in [f_l,f_u] $ means $\{f_0(x) \in [f_l(x),f_u(x)] \text{ for all } x \in \X\} $. \cite{Wasserman:2006} provides an excellent overview of methods for constructing the critical value; \textcolor{black}{see also \cite{gn10} and \cite{cck14} for constructions of adaptive confidence bands in low-dimensional smooth nonparametric models and \cite{bkks20} for a recent proposal  in high-dimensional generalized additive models.}   With a slight abuse of notation,  an initial Bayesian credible region  $[f_l,f_u]$  can be constructed similarly with
the constant $c_p$ determined such that
$$
\Pi \{f_0 \in [f_l,f_u]   \mid S\} \geq p,
$$
where $S$ denotes data (can be a set of statistics derived from data in robust Bayes procedures, for example, means or empirical moment functions), $[f_l,f_u]$ is a measurable function of $S$,
and $\Pi (\cdot \mid S) $ denotes posterior distribution of parameter $f_0$ (viewed as a random element in the Bayesian approach), induced
by $S$ and a prior distribution over potential values $f_0$ can take.  We give empirical and numerical examples in Section 5. 

To enforce the shape restriction, we apply a suitable shape-enforcing operator to the original point estimator and end-point functions of the confidence band.
The resulting estimator, $\mathbf{O}f$, and confidence band, $[\mathbf{O} f_l, \mathbf{O} f_u]$, improve over $f$ and $[f_l, f_u]$ in the sense that 
$f$ lies weakly closer to $f_0$ and  the width of the band $[\mathbf{O} f_l, \mathbf{O} f_u]$ is weakly smaller than that of $[f_l, f_u]$, while the coverage is weakly greater.
These properties of  $\mathbf{O}f$ and  $[\mathbf{O} f_l, \mathbf{O} f_u]$ are a  corollary of Definition~\ref{def:SE_operators}:

\begin{corollary}[Improved Point and Interval Estimators]\label{CI}
	Suppose we have a target function $f_0 \in \ell^{\infty}_1(\X)$, an estimator $f \in  \ell^{\infty}_0(\X)$ a.s., and a confidence band $[f_l,f_u]$ such that $f_l,f_u \in \ell^{\infty}_0(\X)$ a.s.  If the operator $\mathbf{O}$ is $\ell^{\infty}_1$-enforcing with respect to $\ell^{\infty}_0(\X)$, then a.s.
	
	(1) the $\ell^{\infty}_1$-enforced confidence band $[\mathbf{O} f_l, \mathbf{O} f_u]$ has weakly greater coverage than $[f_l,f_u]$:
	$$ 1\{f_0 \in [\mathbf{O} f_l,\mathbf{O} f_u] \} \geq 1\{(f_0 \in [f_l,f_u] \}.$$
	
	If in addition $\mathbf{O}$ is a $\rho$-distance contraction, then a.s.
	
	(2) the $\ell^{\infty}_1$-enforced estimator $\mathbf{O} f$ is weakly closer to $f_0$  than $f$  with respect to the distance $\rho$,
	$$\rho(\mathbf{O}f, f_0) \leq \rho(f,f_0);$$
	
	(3) and the $\ell^{\infty}_1$-enforced confidence band $[\mathbf{O} f_l, \mathbf{O} f_u]$ is weakly shorter than $[f_l,f_u]$ with respect to the distance $\rho$,
	$$\rho(\mathbf{O}f_l,\mathbf{O} f_u) \leq \rho(f_l,f_u).$$
	
	%Suppose $\mathscr{T}$ is $\mathcal{F}-$enforcing with respect to $\mathcal{F}_0$, and $\mathscr{T}$ is distance contraction with respect to $\rho(\cdot,\cdot)$. Assume that the confidence bands $[f_l,f_u]$ covers $f$ with probability $ = 1-\alpha$. Then, $[\mathscr{T}(f_l),\mathscr{T}(f_u)]$ covers $\mathscr{f}$ with probability $\geq 1-\alpha$. If $f\in \mathcal{F}$, then
	%
	%(1) $[\mathscr{T}(f_l),\mathscr{T}(f_u)]$ covers $f$ with probability $\geq 1-\alpha$.
	%
	%(2) $\rho(\mathscr{T}(f_l),\mathscr{T}(f) )\leq \rho(f_l,f)$, $\rho(\mathscr{T}(f_u) , \mathscr{T}(f))\leq \rho(f_u,f)$ and $\rho(\mathscr{T}(f_l),\mathscr{T}(f_u))\leq \rho(f_l,f_u)$.
\end{corollary}

Part (1) shows that  $[\mathbf{O}f_{l},\mathbf{O}f_{u}]$ provides a  coverage improvement over  $[f_{l},f_{u}]$ in that $[\mathbf{O}f_{l},\mathbf{O}f_{u}]$  contains $f_0$ whenever $[f_{l},f_{u}]$ does.  Part (2) shows that the shape-enforced point estimator improves over the original estimator in terms of estimation error measured by the $\rho$-distance between the estimator and the target function. Parts (1) and (3) show that the shape-enforced confidence band not only has greater coverage but also is shorter with respect to the $\rho$-distance than the original band. These improvements apply to any sample size. In particular, they imply that enforcing the shape restriction preserves the statistical properties of the point and interval estimators. Thus, the shape-enforced estimator inherits the rate of consistency of the original estimator, and the shape-enforced confidence band has coverage at least $p$ in large samples if the original band has coverage $p$ in large samples. Corollary~\ref{CI} can therefore be coupled with Theorems~\ref{RR}--\ref{PP} to yield improved inference on a function that satisfies any of the shape restrictions considered in the previous section. It is also worthwhile noting that further quantifying the exact size of improvement depends on $f_0$ and properties of the obtained estimators $f$ and $[f_l, f_u]$. 

%\xnote{Should we comment on why it is difficult to characterize the rate of the improvement　(just in case reviewers force us to quantify the rate)?}

\begin{remark}[Model Misspecification]  \textcolor{black}{Let $f_{\infty}$ denote the probability limit of the estimator $f$, provided that the limit exists. Model misspecification occurs when $f_{\infty}$ is different from the target function $f_0$.} In this case the results of Corollary~\ref{CI} still apply. Moreover,  if $f_{\infty}$ does not satisfy the shape restriction, $f_{\infty} \not\in \ell_1^{\infty}(\X)$, then enforcing this restriction also improves estimation and inference on $\mathbf{O} f_{\infty}$. Thus, the probability limit of the shape-enforced estimator, $\mathbf{O} f_{\infty}$, is closer to $f_0$ in $\rho$-distance than $f_{\infty}$, and the shape-enforced confidence band, $[\mathbf{O} f_l, \mathbf{O} f_u]$, covers $\mathbf{O} f_{\infty}$ with at least the same probability as $[f_l, f_u]$ covers $f_{\infty}$ and $[\mathbf{O} f_l, \mathbf{O} f_u]$ is shorter than $[f_l, f_u]$ in $\rho$-distance.
\end{remark}

%\begin{remark}
%Corollary~\ref{CI}, together with Theorems~\ref{RR}-\ref{QCM} imply that, for all the operators we introduced in Section 2, as they are shape-enforcing operators as described in Definition~\ref{def:SE_operators}, we can perform conservative inference for the shape-enforced functions. The confidence bands will become smaller under the corresponding discrepancy measure $\rho(\cdot,\cdot)$ as specified for each operator.
%\end{remark}

\medskip

\section{Implementation Algorithms} \label{sec:alg}
We provide implementation algorithms for the different shape-enforcing operators based on a sample or grid of $n$ points $\X_n = \{x_1, \ldots, x_n\}$ with corresponding values of $f$ given by the array $\Y_n  = \{y_1, \ldots, y_n\}$ with $y_i = f(x_i)$. Computation of the $\mathbf{R}$-operator is trivial, as it amounts to thresholding the elements of $\Y_n$ to be between $\underline{f}$ and $\overline{f}$, i.e.,
$$
\mathbf{R} f(x_i) = (\underline{f} \vee y_i) \wedge \overline{f}.
$$
When $k=1$, \cite{CFG09} showed that the  $\mathbf{M}$-operator sorts the elements of $\Y_n$. Thus, assume that $x_{1} \leq x_{2} \leq \ldots \leq x_{n}$ and let $y_{(1)} \leq y_{(2)} \leq \ldots \leq y_{(n)}$ denote the sorted array of $\Y_n$. Then,
$$
\mathbf{M} f(x_i) =  y_{(i)}.
$$
When $k>1$, each $\mathbf{M}_{j}$-operator in Definition~\ref{def:M} can be computed by applying the same sorting procedure to the dimension $j$ sequentially for each possible value of the other dimensions.   We refer to \cite{CFG09} for more details on computation. We next develop new algorithms for the $\mathbf{C}$ and $\mathbf{Q}$ operators.

\subsection{Computation of  $\mathbf{C}$-Operator}
When $k=1$, we can obtain the greatest convex minorant  using the standard method based on the pool adjacent violators algorithm described in \cite{barlowstatistical}. We provide an algorithm for the case where $k>1$.
By Definitions \ref{def:FL} and~\ref{def:DFL}, the DFL transform of $f$ is the solution to
$$
\mathbf{C} f(x) = \sup_{\xi \in f^*(\X)} \inf_{\tilde x \in \X} \{ \xi'(x - \tilde x) + f(\tilde x) \}.
$$
This is a saddle point problem that might be difficult to tackle directly. However, when $\X$ is replaced by the finite grid $\X_n$, the problem has a convenient linear programming representation:
\begin{eqnarray}\label{eq:DLFatx}
	\mathbf{C} f(x) = &\max_{v \in \mathbb{R}, \xi \in \mathbb{R}^k} & v \\
	& \textrm{s.t. } & v+\xi'(x_i-x) \leq f(x_i), \quad i=1,2,\ldots,n. \nonumber
\end{eqnarray}
This program can be solved using standard linear programming methods. In particular, the computational complexity of the standard interior point method for solving \eqref{eq:DLFatx} is $O((k+1)(n+(k+1))^{1.5})$, where $k+1$ is the number of decision variables and $n$ is the number of constraints.
%
%
%To implement the double Legendre Transform, we propose the following algorithm to compute the value of $\mathscr{C}(f)$:
%
%Define $X$ as a $n\times k$ matrix with $x_1,\ldots,x_n$ being its rows. Define $Y$ as a $n\times 1$ vector with $Y=(f(x_1),\ldots,f(x_n))'$. For any $p\in \mathbb{R}^k$, define $g(p):=\min_{i=1,2,\ldots,n}\{(Y-(X-x)p)_{i}\}$. By definition, $\mathscr{C}(f)(x)=\sup_{p\in \mathbb{R}^k} g(p)$. It is also easy to see that $g(p)$ is a concave function of $p$.
%
%\xnote{Scott: should we introduce $x$ above as "For any $p\in \mathbb{R}^k$ and $x\in \mathbb{R}^k$", or define $g$ as a function of also $x$? }
%
%For any set of points $(x_1,\ldots,x_n)$, since we can compute the value of $g(p)$ given any $p\in \mathbb{R}^k$. Therefore by maximizing $g(p)$, we are able to obtain the value of $\mathscr{C}(f)(x)$. In practice, we maximize $g(p)$ with $p$ being constrained to $[A,-A]^k$ to approximate $\mathscr{C}(f)(x)$, where $A$ is set to be a positive real number large enough. Such a restriction makes sure that there always exists a solution as the maxima of $g(p)$.
%
%Therefore, we can summarize the above procedure as a linear programming problem, with $v\in \mathbb{R}$ and $p\in \mathbb{R}^k$ are unknown variables. Given point $x$ and function $f(\cdot)$, the $DLF$ of $f$ at point $x$ is the solution of the following problem:
%
%\begin{eqnarray}\label{eq:DLFatx}
%&\max_{v \in \mathbb{R}, p \in \mathbb{R}^k} & v \\
%& \textrm{s.t. } & v+(X_i-x)'p\leq Y_i, \quad i=1,2,\ldots,n. \nonumber
%\end{eqnarray}

The following algorithm summarizes the computation of the $\mathbf{C}$-Operator.
\begin{algorithm}[$\mathbf{C}$-Operator]\label{Algo:C}
	(1) Pick a dense enough grid of size $n$ in $\mathcal{X}$, denoted as $\mathcal{X}_n$. One natural choice is the set of values of $x$ observed in the data. (2) For each $x\in \mathcal{X}_n$, solve the linear programming problem stated in (\ref{eq:DLFatx}) to obtain $\mathbf{C} f(x)$. \end{algorithm}

\subsection{Computation of  $\mathbf{Q}$-Operator}
We propose a method to compute the $\mathbf{Q}$ operator based on solving problem \eqref{def:Q} on a finite grid, namely
$$
\mathbf{Q} f(x) = \min \left\{ y \in \Y_n : x \in \mathrm{conv}[\mathcal{I}_{f,n}(y)] \right\},
$$
where $\mathcal{I}_{f,n}(y)  = \{ x_i : y_i = f(x_i) \leq y, i = 1,2,\ldots,n\}$. We find the solution to the program using the following bisection search algorithm:

%In practice, suppose we observe $(x_1,y_1),\ldots,(x_n,y_n)$ as i.i.d. sample, with $y_i = f(x_i)$. We are interested in constructing a quasi-convex function based on these observed points. We define $\mathscr{Q}_n(f)$ as follows:
%\begin{equation}\label{eq:Q_n}
%\mathscr{Q}_n(f)(x)=min_{y}\{x\in \mathscr{Q}_0(\mathcal{I}_{f,n}(y))\},
%\end{equation}
%where $I_{f,n}:=\{x_i|y_i\leq y, i=1,2,\ldots,n\}$ as the set of $x_i$ such that $f(x_i)\leq y$.
%
%%add empirical version Q_n...
%
%The algorithm for computing $\mathscr{Q}_n(f)(x)$ is elaborated as follows:

\begin{algorithm}[$\mathbf{Q}$-Operator]\label{Algo:QC}
	For a given $x \in \X_n$: (1) Initialize $y_L=y_{(1)}$ and $y_U=y_{(n)}$. (2) Find the median  of $\{y \in \Y_n: y_L \leq y \leq y_U\}$ and assign it to $y^*$. (3) Compute the lower contour set $\mathcal{I}_{f,n}(y^*)$. (4)  If $x \in \mathrm{conv}[\mathcal{I}_{f,n}(y^*)]$ (which indicates  $y^* \geq \mathbf{Q} f(x)$), set $y_U=y^*$; otherwise, set $y_L=y^*$. (5) Repeat (2)--(4) until $y_U=y_L$ and report $\mathbf{Q} f(x)  = y_U$.
	%
	%Given $\{x_1, \ldots, x_n\}$, let $y=(f(x_1), \ldots, f(x_n))$ and the sorted array $y_{(1)} \leq y_{(2)} \leq \ldots \leq y_{(n)}$. Moreover let $y_L=y_{(1)}$ and $y_U=y_{(n)}$. The algorithm is as follows:
	%\begin{enumerate}
	%  \item Find the median $y_{(p)}$ of all the $y$'s between $y_L$ and $y_U$.
	%  \item Compute the empirical level set $\mathcal{I}_{f,n}(y_{(p)})$ and its convexification  $\mathscr{C}(\mathcal{I}_{f,n}(y_{(p)})$
	%  \item Check if $x \in \mathcal{I}_{f,n}(y_{(p)})$\footnote{When $d \leq 2$, the MATLAB has the function  convhull, when $d >2$, this problem can be solved by an feasibility problem in linear programming}. If $x \in \mathcal{I}_{f,n}(y_{(p)})$ (which indicates  $y_{(p)} \geq \mathscr{Q}_n(f)(x)$ according to (\ref{eq:Q_n})), set $y_U=y_{(p)}$; else set $y_L=y_{(p)}$.
	%  \item Stop the procedure until $y_U=y_L$ and report the value of $y_U=y_L$ as $\mathscr{Q}_n(f)(x)$.
	%\end{enumerate}
\end{algorithm}

The binary search algorithm for the $\mathbf{Q}$-Operator runs in \textcolor{black}{$O(\log(n))$} iterations. The major computational cost within each iteration is the check of whether $x$ is in the convex hull in step (4). This check does not require construction of the actual convex hull, which is computationally expensive especially in high dimensions. Instead, it is sufficient to check the existence of a feasible solution of a linear program.

We further note that each of the above two algorithms can be run in parallel across the grid points,
because the output of the algorithm for one grid point does not depend on the
output for any other grid point. This parallelizability allows for efficient
computation on nontrivial grids.

\medskip

\section{Numerical Examples} \label{sec:num}

%\subsection{Monte-Carlo illustration}
%
%\subsubsection{Examples for double Legendre Fenchel Transform $\mathscr{C}$}
%
%{\color{red}Scott: add this section}

%\subsubsection{Examples for Quasi-convexification Operator $\mathscr{Q}$}
%
%{\color{red} Scott: add this section}

\subsection{Univariate Case} We consider an empirical application to growth charts and a calibrated simulation where the target function $f_0$ is univariate.

\subsubsection{Height Growth Charts for Indian Children}
Since their introduction by Quetelet in the 19th century, reference
growth charts have become common tools to assess an individual's
health status. These charts describe the evolution of individual
anthropometric measures, such as height, weight, and body mass
index, across different ages. See \cite{cole} for a classical work
on the subject, and \cite{koenker:charts} for a detailed analysis
and additional references.
Here we consider the estimation of height growth charts  imposing monotonicity and concavity restrictions. These restrictions are plausible, since an individual's height is nondecreasing in age at a nonincreasing growth rate during early childhood; see, e.g., \cite{twt66} and the  growth child standards of the World Health Organization at \url{https://www.who.int/childgrowth/en/}.

%\inote{It would be nice to to find a reference for concavity of height growth charts during early childhood}

We use the data from \cite{fnh14} and \cite{koenker2011} on childhood malnutrition in India. These data include a measure of height in centimeters, $Y$,  age in months, $X$,  and 22 covariates, $Z$, for 37,623 Indian children.  All of the children have ages between 0 and 5 years, i.e., $X\in \X = \{0,1,\ldots,59\}$. The covariates $Z$ include the mother's body mass index, the number of months the child was breastfed, and the mother's age (as well as the square of the
previous three covariates); the mother's years of education and the
father's years of education; indicator variables for the child's sex,
whether the child was a single birth or part of a multiple birth, whether
the mother was unemployed, whether the mother's residence is
urban or rural, and whether the mother has each of: electricity, a
radio, a television, a refrigerator, a bicycle, a motorcycle, and
a car; and factor variables for birth order of the child, the mother's
religion and quintiles of wealth.

\medskip

We assume a partially linear model for the conditional expectation of $Y$ given $X$ and $Z$, namely
$$
\Ep[Y \mid X=x, Z = z] = f_0(x) + z'\gamma.
$$
The target function is the conditional average growth chart $x \mapsto f_0(x)$, which we assume to be nondecreasing  and concave.  Since $X$ is discrete, we can express $f_0(x) = P(x)'\beta$, where $P(x)$ is a vector of  indicators for each value in $\X$, i.e., $P(X) = [1(X=0), 1(X=1), \ldots, 1(X=59)]'$. We estimate $\beta$ and $\gamma$ by least squares of $Y$ on $P(X)$ and $Z$, and construct a confidence band for $f_0$ on $\X$ using weighted bootstrap with standard exponential weights and 200 repetitions \cite{praestgaard1993exchangeably,hahn1995bootstrapping}. Standard errors are estimated using bootstrap rescaled interquartile ranges \cite{chernozhukov2013inference}, and the critical value is the bootstrap $0.95$-quantile of the maximal $t$-statistic. Weighted bootstrap is computationally convenient in this application because it is less sensitive than empirical bootstrap to singular designs, which are likely to arise in the bootstrap resampling because $Z$ and $P(X)$ contain many indicators.

\medskip
%\inote{Ye: include knot sequence for B-splines}

%Growth charts measure the relationship between age and anthropometric measures such as  height, weight or body mass index. We apply our methodology to obtain growth charts of children with an Indian data. We are interested in obtaining the relationship between height (denoted as $Y$) and age measured by month (denoted as $X$). In the data, all children are chosen with age between 0-5 years old, i.e., $X\in \X = \{1,2,\ldots,60\}$. There are additional covariates $Z$ available for control, and we estimate the following partially linear model:
%\begin{equation}\label{eq:indian-growth}
%Y=f_0(X)+Z'\gamma+\epsilon,
%\end{equation}
%where $f$ is an unknown  function.
%
%
%
%
%Below we present the first set of results, with $P(X)$ being the vector of indicators for each value in $\X$, i.e., $P(X) = [1(X=2), \ldots, 1(X=60)]'$.

Figures \ref{figc1} and~\ref{figc2} report the point estimates and 95\% confidence bands of $f_0$ for the entire sample and a random extract with $1{,}000$ observations,  respectively.  We use the subsample to illustrate the deviations from the shape restrictions that are more apparent when the sample size is small.  The original estimates are displayed in the left panels,  and the estimates imposing monotonicity and concavity in the right panels. The original estimates in the entire sample are nondecreasing in age except at 45 months, and deviate from concavity in some areas. The  $\mathbf{M}$ and $\mathbf{\nOn{C}}$ operators correct these deviations. The estimates in the random extract of the data clearly show deviations from both monotonicity and concavity. The $\mathbf{M}$ and $\mathbf{\nOn{C}}$ operators fix these deviations and produce point estimates that are closer to the estimates in the entire sample.

\begin{figure}[!htb]
	\centering
	\includegraphics[width=.49\textwidth]{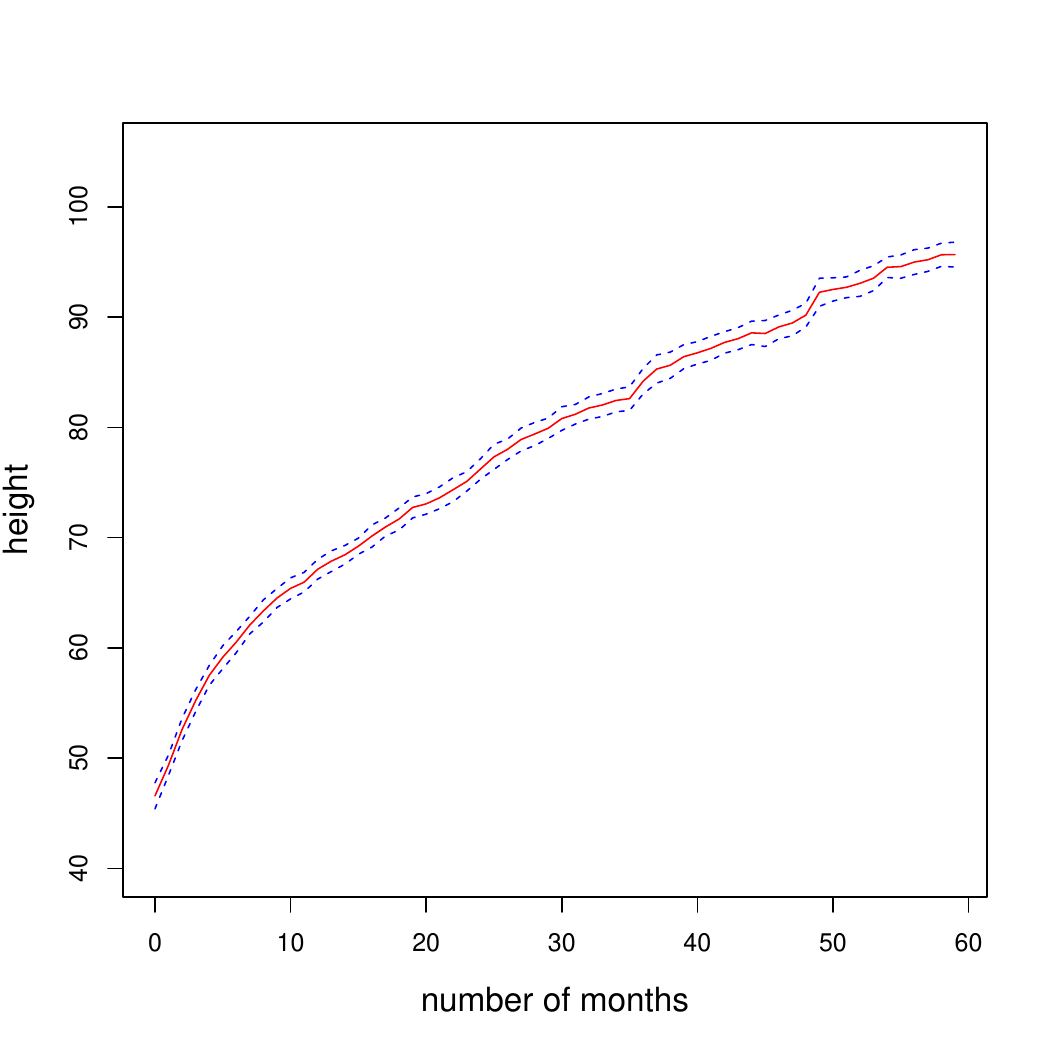}
	\includegraphics[width=.49\textwidth]{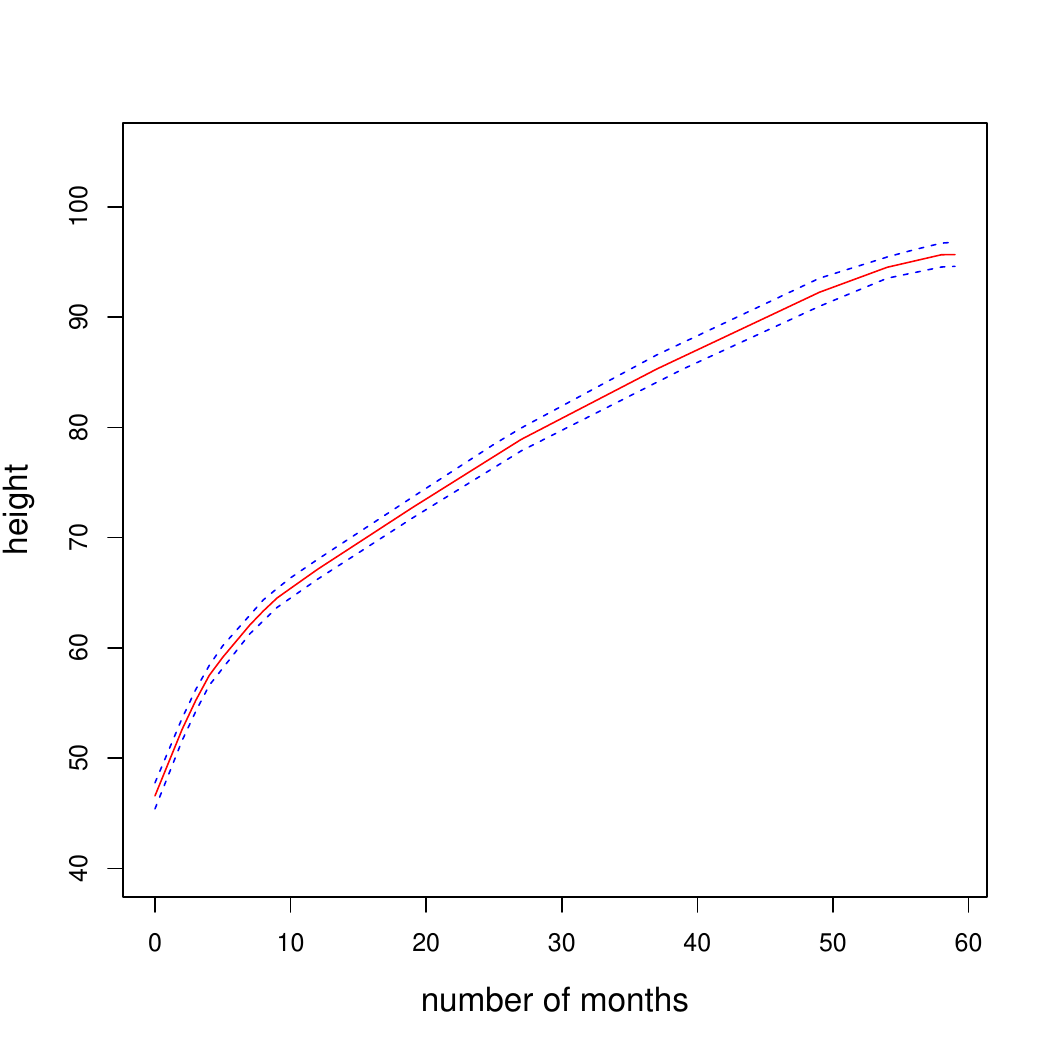}
	\caption{Entire Sample:  Estimates and 95\% confidence bands. Left: $ f$ and $[f_l,f_u]$.
		%Center:  $\mathbf{M} f$ and $[\mathbf{M} f_l, \mathbf{M} f_u]$.
		Right: $\mathbf{\nOn{C}M} f$ and  $[\mathbf{\nOn{C}M} f_l, \mathbf{\nOn{C}M} f_u]$} \label{figc1}
\end{figure}

\begin{figure}[!htb]
	\centering
	\includegraphics[width=.49\textwidth]{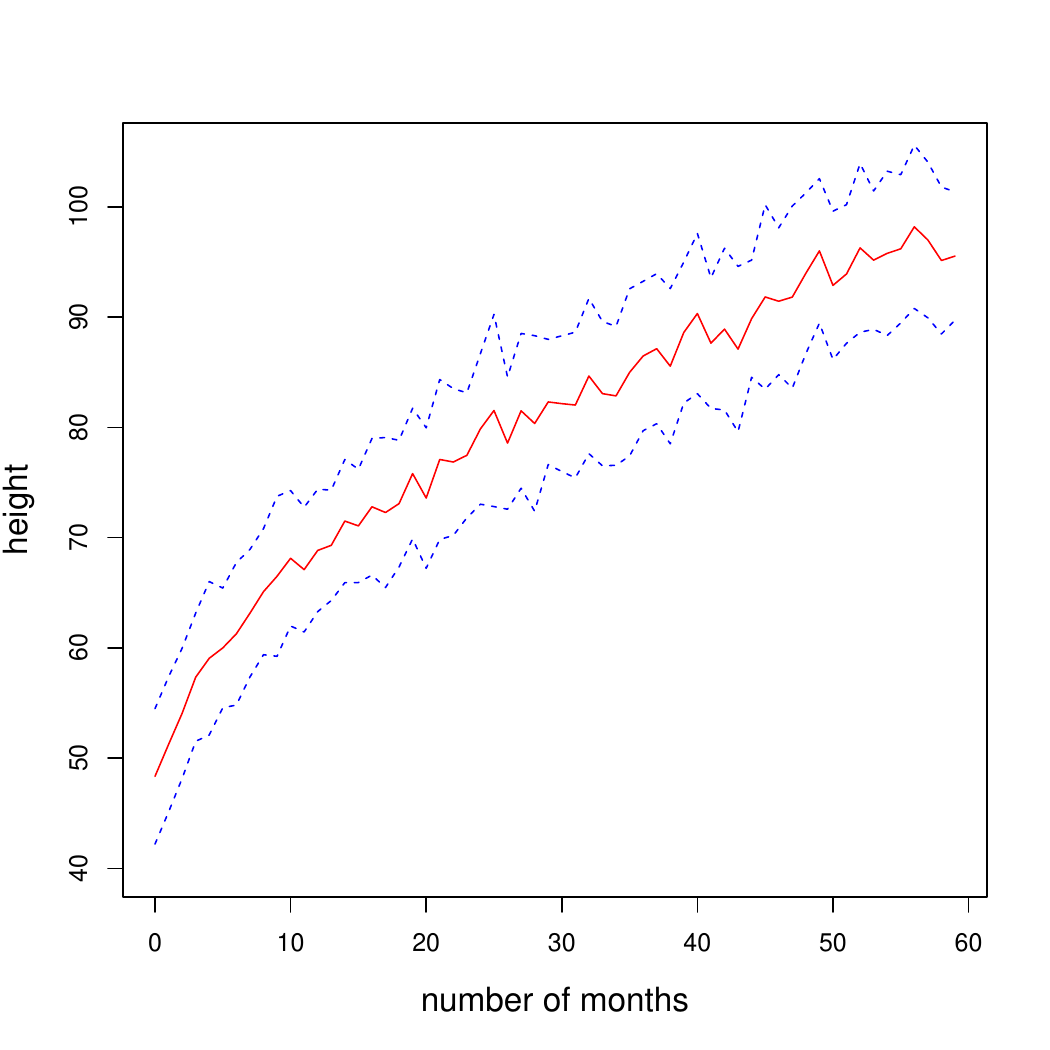}
	\includegraphics[width=.49\textwidth]{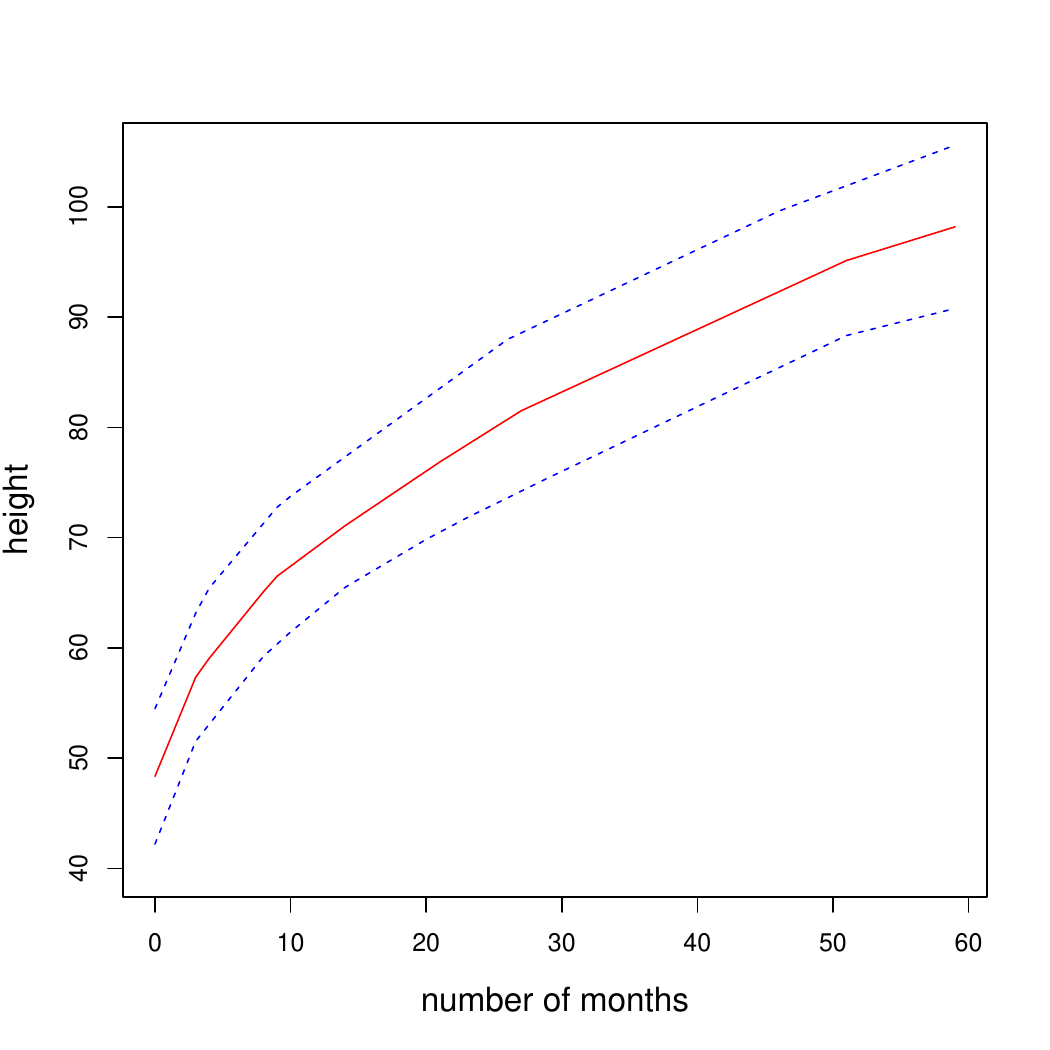}
	\caption{Subsample with 1,000 observations:  Estimates and 95\% confidence bands. Left: $ f$ and $[f_l,f_u]$.
		%Center:  $\mathbf{M} f$ and $[\mathbf{M} f_l, \mathbf{M} f_u]$.
		Right: $\mathbf{\nOn{C}M} f$ and  $[\mathbf{\nOn{C}M} f_l, \mathbf{\nOn{C}M} f_u]$} \label{figc2}
\end{figure}
%{\red The graphs here is commented.}
%
%%The second set of results consist estimating the growth function using B-splines, with degree of freedom $=20$.
%
%\begin{figure}[!htb]
%\centering
%\includegraphics[width=.48\textwidth]{./fig/original_bands_bs.pdf}
%%\includegraphics[width=5cm]{./fig/rearranged_bands_bs.pdf}
%\includegraphics[width=.48\textwidth]{./fig/DLF+REARR_bands_bs.pdf}
%\caption{B-Spline specification: Estimates and 95\% confidence bands. Left: $ f$ and $[f_l,f_u]$.
%%Center:  $\mathbf{M} f$ and $[\mathbf{M} f_l, \mathbf{M} f_u]$.
%Right: $\mathbf{C} f$ and  $[\mathbf{C} f_l, \mathbf{C} f_u]$} \label{figc2}\end{figure}

\subsubsection{Calibrated Monte Carlo Simulation}
We quantify the finite-sample improvement in the point and interval estimates of enforcing shape restrictions using simulations calibrated to the growth chart application. The child's height, $Y$, is generated by
$$
Y_i = \mathbf{\nOn{C}M} [P(X_i)'\hat \beta] +Z_i'\hat\gamma+ \hat\sigma \epsilon_i, \ \ i = 1,\ldots,n,
$$
where $P(X_i)$ is the vector of indicators for all the values of $\X$; $\hat \beta$, $\hat\gamma$ and $\hat\sigma$ are the least squares estimates of $\beta$, $\gamma$ and the residual standard deviation in the growth chart data; and $\epsilon_i$ are independent draws from the standard normal distribution.  The application of the $\mathbf{\nOn{C}M}$-operator guarantees that the target function $f_0(x) = \mathbf{\nOn{C}M} [P(x)'\hat \beta]$ is monotone and concave. We consider six sample sizes, $n \in \{500, 1{,}000, 2{,}000, 4{,}000, 8{,}000, 37{,}623\}$, where $n = 37{,}623$ is the same sample size as in the empirical application. The values of $X_i$ and $Z_i$ are randomly drawn from the data without replacement. The results are based on 500 simulations. In each simulation we construct point and band estimates of $f_0$ using the same methods as in the empirical application.

%In this subsection, we will utilize the estimated model from the previous subsection. Let $\hat f(X) = $ denote the least squares estimator of the partially linear model $Y=P(X)'\beta+Z'\gamma+\epsilon$. Also denote $\hat \gamma$ as the estimated $\gamma$ from the previous section.
%
%We consider to simulate the outcome variable $Y_i:= \mathbf{CM} [P(X_i)'\hat \beta] +Z_i'\hat\gamma+ \hat\sigma \epsilon_i$, where $\epsilon_i \sim N(0,1)$, with $\hat\sigma$ calibrated to data. Let $n$ be the simulated same size. We would like to draw $n$ out of the full sample and examine our algorithm as different levels of data size $n$.
%

\begin{table}[!t]
	\caption{Finite-Sample Properties}
	\label{bias_bw}
	%\footnotesize
	\begin{tabular}{lcccccc}
		\hline \hline
		& $\|f - f_0 \|_{\infty}$ & $\|f_u - f_l \|_{\infty}$ & $\Pr(f_0 \in [f_l,f_u])$ & $\|f - f_0 \|_{\infty}$ & $\|f_u - f_l \|_{\infty}$ & $\Pr(f_0 \in [f_l,f_u])$ \\ \hline \\
		& \multicolumn{3}{c}{$n = 500$} & \multicolumn{3}{c}{$n = 1{,}000$} \\ \\
		Original & 7.45 & 26.31 & 0.69   & 4.83 & 18.45 & 0.84  \\
		$\mathbf{\nOn{C}}$ & 6.28 & 21.77 & 0.79 & 4.16 & 16.95 & 0.89 \\
		$\mathbf{M}$ & 5.23 & 21.24 & 0.92 & 3.72 & 16.58 & 0.95 \\
		$\mathbf{\nOn{C}M}$ & 4.73 & 20.15 & 0.93 & 3.30 & 15.97 & 0.95 \\ \\
		& \multicolumn{3}{c}{$n = 2{,}000$} & \multicolumn{3}{c}{$n = 4{,}000$} \\ \\
		Original & 3.35 & 13.55 & 0.90   & 2.29 & 9.61 & 0.93 \\
		$\mathbf{\nOn{C}}$ & 2.93 & 13.21 & 0.94 & 1.98 & 9.53 & 0.96 \\
		$\mathbf{M}$ & 2.79 & 12.89 & 0.96 & 2.00 & 9.43 & 0.97 \\
		$\mathbf{\nOn{C}M}$ & 2.48 & 12.64 & 0.97 & 1.76 & 9.34 & 0.98 \\    \\
		& \multicolumn{3}{c}{$n = 8{,}000$} & \multicolumn{3}{c}{$n = 37{,}623$} \\ \\
		Original & 1.61 & 6.92 & 0.93  & 0.72 & 3.22 & 0.95  \\
		$\mathbf{\nOn{C}}$ & 1.40 & 6.90 & 0.97 & 0.64 & 3.21 & 0.97 \\
		$\mathbf{M}$ & 1.47 & 6.89 & 0.96 & 0.71 & 3.22 & 0.96  \\
		$\mathbf{\nOn{C}M}$ & 1.30 & 6.87 & 0.98  & 0.63 & 3.21 & 0.98 \\ \\
		\hline \hline
		\multicolumn{7}{l}{\footnotesize{Notes: Based on $1{,}000$ simulations. Nominal level of the confidence bands is $95\%$. } }\\
		\multicolumn{7}{l}{\footnotesize{Confidence bands constructed by weighted bootstrap with standard exponential weights and $200$ repetitions.} }
	\end{tabular}
\end{table}

%\begin{table}[ht]
%\caption{Finite-Sample Properties}
%\label{bias_bw}
%%\footnotesize
%\begin{tabular}{ccccccc}
%  \hline \hline
%  \multicolumn{1}{c}{$n$} & $\|f - f_0 \|_{\infty}$ & $\|f_u - f_l \|_{\infty}$ & $\Pr(f_0 \in [f_l,f_u])$ & $\|f^* - f_0 \|_{\infty}$ & $\|f_u^* - f^*_l \|_{\infty}$ & $\Pr(f_0 \in [f^*_l,f^*_u])$ \\ \hline \\
%   \multicolumn{7}{c}{A. Indicator Specification} \\ \\
%  1,000 & 9.049 & 9.127 &  0.946 & 8.832 & 7.405 &  0.962 \\
%  2,000 & 5.956 & 5.982 &  0.954 & 5.894 & 4.919  &  0.970  \\
%  5,000 &  3.934 &  4.011 &  0.946 &  3.917 & 3.584  &  0.968  \\
%  10,000 &  2.981  & 3.058  &  0.940 &  2.978 & 2.908   &  0.942  \\
%  37,623 & 1.818  & 2.040 & 0.958 & 1.818  & 2.093  & 0.960 \\
%  \\
%   \multicolumn{7}{c}{B. B-Spline Specification} \\ \\
%  1,000 & 7.350 & 8.152 &  0.960 & 6.891 & 8.134 &  0.968   \\
%  2,000 & 5.434 & 5.773 &  0.948 & 5.190 & 5.772  &  0.954  \\
%  5,000 &  3.619 &  4.117 &  0.982 &  3.500 & 4.117  &  0.982  \\
%  10,000 &  3.130  & 3.771  &  0.992 &  3.078 & 3.771   &  0.992  \\
%  37,623 & 1.473  & 1.616 & 0.982 & 1.446  & 1.616  & 0.982  \\   \\
%   \hline \hline
%\multicolumn{7}{l}{\footnotesize{Notes: $f^* = \mathbf{CM} f ,$ $f_l^* = \mathbf{CM} f_l,$ and $f_u^* = \mathbf{CM} f_u$. Based on 500 simulations.} }
%\end{tabular}
% \end{table}

Table~\ref{bias_bw} reports simulation averages of the $d_{\infty}$-distance between the estimates and target function, coverage of the target function by the confidence band and $d_{\infty}$-length of the confidence band for the original and shape-enforced estimators. We consider enforcing concavity with the $\mathbf{\nOn{C}}$-operator, monotonicity with the $\mathbf{M}$-operator, and both concavity and monotonicity with the $\mathbf{\nOn{C}M}$-operator. The improvements from imposing the shape restrictions are decreasing in the sample size, but there are substantial benefits in estimation error even with the largest sample size. Enforcing monotonicity has generally stronger effects than enforcing concavity, but both help improve the estimates. Thus,  the $\mathbf{\nOn{C}M}$-operator produces the best point and interval estimators for every sample size. For the smallest sample size, the reduction in estimation error is almost 37\%  and the improvement in length of the confidence band is more than 20\%. The gains in coverage probability are also substantial, especially for the smaller sample sizes. Overall, the simulation results clearly showcase the benefits of enforcing shape restrictions, even with large sample sizes.

% latex table generated in R 3.6.3 by xtable 1.8-4 package
% Sun Jan 31 13:28:29 2021
\begin{table}[ht]
	\centering
	\textcolor{black}{
	\caption{\textcolor{black}{Comparison  of Estimators}} \label{comparison}
	\begin{tabular}{lcccccc}
		\hline\hline
		& 500 & 1,000 & 2,000 & 4,000 & 8,000 & 37,623 \\ 
		\hline
		Piecewise Constant & 5.85 & 3.93 & 2.66 & 1.82 & 1.27 & 0.58 \\ 
		\quad $\mathbf{\nOn{C}}$ PC & 5.32 & 3.65 & 2.49 & 1.71 & 1.19 & 0.54 \\ 
		\quad   $\mathbf{S\nOn{C}}$ PC & 4.89 & 3.41 & 2.34 & 1.59 & 1.10 & 0.52 \\ 
		\quad   $\mathbf{M}$  PC& 3.73 & 2.83 & 2.13 & 1.55 & 1.15 & 0.57 \\ 
		\quad   $\mathbf{\nOn{C}M}$ PC & 3.54 & 2.67 & 2.01 & 1.46 & 1.08 & 0.53 \\ 
		\quad   $\mathbf{S\nOn{C}M}$ PC & 3.29 & 2.45 & 1.84 & 1.33 & 0.97 & 0.50 \\ 
		Conreg & 3.04 & 2.27 & 1.64 & 1.16 & 0.83 & 0.42 \\ 
		Isoreg & 3.82 & 2.95 & 2.21 & 1.60 & 1.17 & 0.57 \\ 
		\quad   $\mathbf{\nOn{C}}$ Isoreg & 3.52 & 2.75 & 2.08 & 1.50 & 1.10 & 0.53 \\ 
		\quad   $\mathbf{S\nOn{C}}$ Isoreg & 3.28 & 2.53 & 1.90 & 1.37 & 1.00 & 0.50 \\ 
		Locally Linear & 2.66 & 2.03 & 1.53 & 1.15 & 0.87 & 0.52 \\ 
		\quad   $\mathbf{\nOn{C}}$ LL & 2.64 & 2.02 & 1.52 & 1.14 & 0.87 & 0.52 \\ 
		\quad   $\mathbf{S\nOn{C}}$ LL & 2.61 & 1.99 & 1.51 & 1.13 & 0.87 & 0.53 \\ 
		\quad   $\mathbf{M}$ LL & 2.60 & 1.99 & 1.50 & 1.13 & 0.87 & 0.52 \\ 
		\quad   $\mathbf{\nOn{C}M}$-LL & 2.58 & 1.98 & 1.50 & 1.13 & 0.87 & 0.52 \\ 
		\quad   $\mathbf{S\nOn{C}M}$ LL & 2.54 & 1.94 & 1.48 & 1.12 & 0.86 & 0.52 \\ 
		\hline\hline
		\multicolumn{7}{l}{\footnotesize{Notes: Based on $5{,}000$ simulations.  Entries are $\|f - f_0 \|_{\infty}$.} }\\
	\end{tabular}
	}
\end{table}

\textcolor{black}{
	We compare the $d_{\infty}$-error of several estimators in the simplified design
	$$
	Y_i = \mathbf{\nOn{C}M} [P(X_i)'\hat \beta] + \hat\sigma \epsilon_i, \ \ i = 1,\ldots,n,
	$$
	where $P(X_i)$, $\hat \beta$, $\hat\sigma$, $\epsilon_i$ and $n$ are the same as for Table~\ref{bias_bw}. We consider unconstrained, shape-constrained, shape-enforced, and combinations of shape-enforced and shape-constrained estimators. The unconstrained estimators include the same estimator as in Table~\ref{bias_bw} (Piecewise Constant) and a locally linear estimator with data-driven choice of bandwidth (Locally Linear).\footnote{The piecewise constant estimator can be viewed as a locally constant estimator with bandwidth equal to zero. The locally linear estimator is computed using the package \texttt{KernSmooth} \cite{KernSmooth19} with the bandwidth chosen by the plug-in method of \cite{rsw95}.} We consider two classical shape-constrained estimators: the isotonic regression estimator (Isoreg) that imposes monotonicity and the concave regression estimator (Conreg) that imposes concavity.\footnote{We compute the isotonic regression using the \texttt{R} command \texttt{isoreg} \cite{r19}, and the concave regression using the package \texttt{cobs} \cite{cobs20}.} We illustrate how to combine shape-enforced operators with shape-constrained estimators by applying the $\mathbf{\nOn{C}}$-operator to the isotonic regression estimator to enforce monotonicity and concavity. Finally, we compare the $\mathbf{\nOn{C}}$-operator with the $\mathbf{S\nOn{C}}$-operator defined in Remark \ref{remark:sc}.}

\textcolor{black}{Table~\ref{comparison} shows the results based on $5{,}000$ simulations. The comparison between shape-constrained and shape-enforced estimators produces  mixed results, which vary with the estimator, shape restriction and sample size. Thus, the $\mathbf{M}$-operator outperforms Isoreg for both unconstrained estimators, whereas Conreg outperforms the $\mathbf{C}$-operator applied to Piecewise Constant. The unconstrained locally linear estimator outperforms Isoreg, Conreg and the shape-enforced estimators applied to Piecewise Constant for most sample sizes, despite the target function not being smooth. This finding highlights the benefit of using  estimators that exploit smoothness  when the  sample size is not large.  On the other hand, the shape-enforcing operators are more effective when applied to estimators such as Piece Constant and Isoreg that do not rely on smoothness. Shifting the $\mathbf{\nOn{C}}$-operator to deal with potential bias generally reduces estimation error for all the estimators considered.}

\subsection{Multivariate Case} We consider an empirical application to production functions and a calibrated simulation where the target function $f_0$ is bivariate.

\subsubsection{Production Functions of Chinese Firms}
The production function is a fundamental relationship in economics that maps the quantity of inputs, such as capital, labor and intermediate goods, to the quantity of output of a firm. When there are only two inputs, the law of diminishing marginal rate of technical substitution dictates that the production function of a firm is nondecreasing and quasi-concave in the inputs \cite{Hicks_Allen_1934}.  If in addition the industry exhibits diminishing returns to scale, then the production function is concave in the inputs. We use the data from \cite{Jacho-Chavez_Lewbel_Linton_2010},
and \cite{Horowitz_Lee_2017} to estimate the production function of Chinese firms in the chemical industry. These data contain information on real value added (output), real fixed assets (capital) and number of employees (labor) for 1,638 firms in 2001.\footnote{Following \cite{Jacho-Chavez_Lewbel_Linton_2010} and \cite{Horowitz_Lee_2017}, we drop observations with a capital-to-labor ratio below the $0.025$ sample quantile or above the $0.975$ sample quantile.}   We estimate a production function using these data and enforce the monotonicity and quasi-concavity restrictions.  We provide results from enforcing concavity  only for illustrative purposes because the chemical industry might exhibit increasing returns to scale at some levels of the inputs.

\begin{figure}
	\centering
	\includegraphics[width=\textwidth]{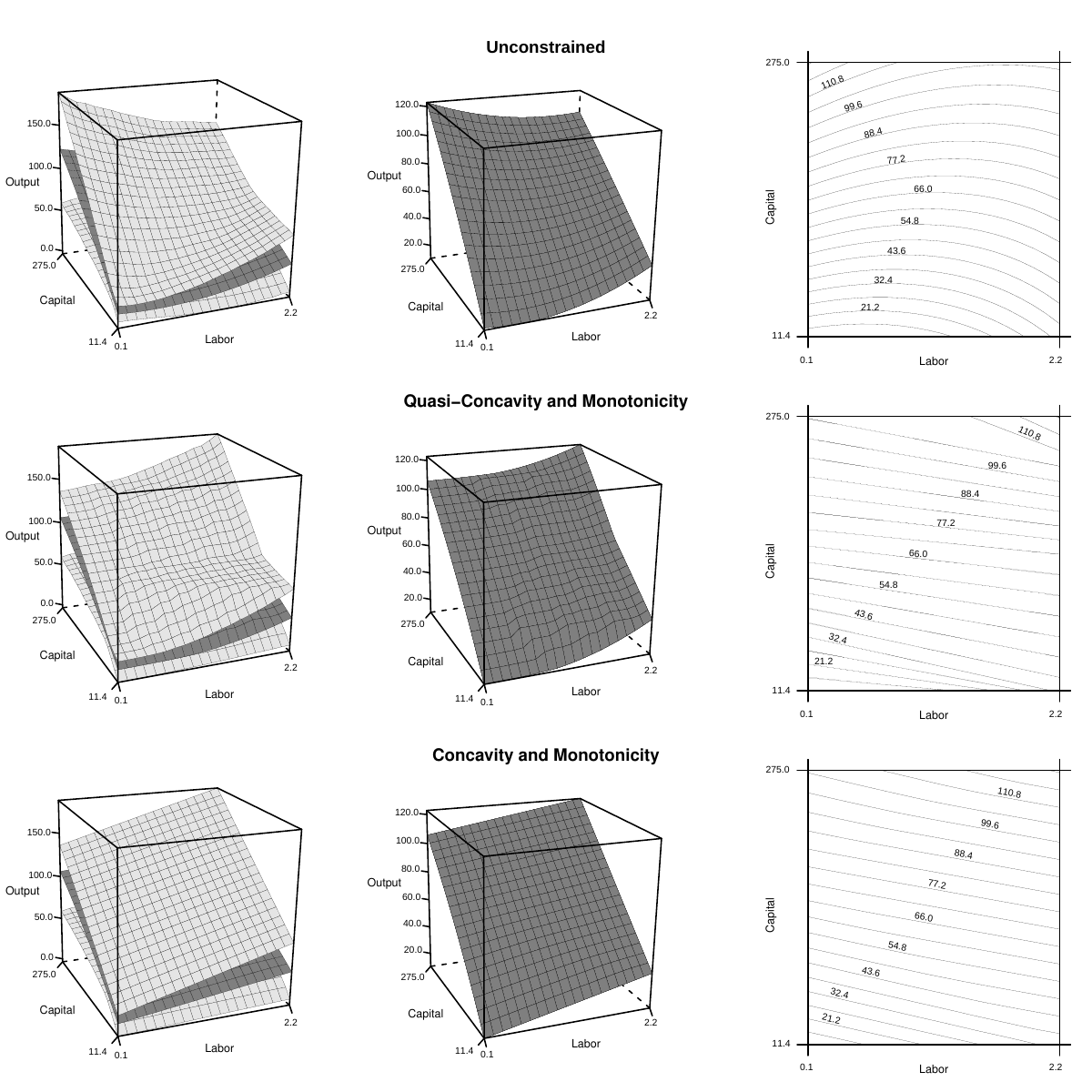}
	\caption{Confidence Bands and Countour Maps of Shape-Constrained and Unconstrained
		Estimates of the Production Function. Output and capital are measured in millions of 2000 yuan, and labor is measured in thousands of workers.}\label{fig_prodfn01}
\end{figure}

Figure~\ref{fig_prodfn01} shows 3-dimensional estimates and 95\% confidence
bands for the average production function, together with upper contour sets for the point estimates. The estimates and bands are displayed in a region defined by the tensor product of two grids for labor and capital. Each grid includes 20 equidistant points from the 10\% to the 90\% sample percentiles of the corresponding variable.  We obtain the unconstrained estimates from least squares with the tensor product of third-degree global polynomials as the two marginal bases for capital and labor. The confidence bands are constructed using  weighted bootstrap with standard exponential weights and 500 repetitions. Standard errors are estimated using bootstrap rescaled interquartile ranges and the critical value is the bootstrap $0.95$-quantile of the maximal $t$-statistic. Many of the upper contour sets of the unconstrained point estimates
are far from being convex, and thus imply a violation of quasi-concavity. In fact, violations of monotonicity occur over a considerable area---most notably, see the positive slopes of the contour curves at low levels of labor and high levels of capital (the upper-left region of the contour plot).

The second row of Figure~\ref{fig_prodfn01} shows the results after the
$\mathbf{\nOn{Q}M}$-operator is applied to the point estimates and to each end-point function of the
confidence band to ensure monotonicity and quasi-concavity.
%\footnote{To restrict $f(k,l)$ to quasi-concavity, it suffices to take the negative of the quasi-convexification of $-f(k,l)$.} 
The contour curves are convex by construction, and thus satisfy
the quasi-concavity restrictions. Finally, the third row of Figure~\ref{fig_prodfn01} shows the results after the $\mathbf{\nOn{C}M}$-operator is applied to enforce monotonicity and concavity. Although quasi-concavification of a production function estimate is always
reasonable, whether restriction to concavity is appropriate depends on prior
knowledge of the industry.

\subsubsection{Calibrated Monte Carlo Simulation}

\begin{table}[ht]
	\centering
	\caption{Finite-Sample Properties}
	\label{mc_prodfn}
	% latex table generated in R 3.4.4 by xtable 1.8-2 package
	% Tue Dec 18 17:27:08 2018
	\begin{tabular}{llccccc}
		\hline
		\hline
		\multicolumn{2}{l}{} & \multicolumn{5}{c}{$n$} \\
		&  & 100 & 200 & 500 & 1,000 & 1,638 \\ 
		\hline
		$\Pr(f_{0}\in[f_{l},f_{u}])$ & Original & 0.83 & 0.89 & 0.93 & 0.94 & 0.95 \\ 
		%   & $\mathbf{M}$ & 0.926 & 0.970 & 0.980 & 0.966 & 0.957 \\ 
		& $\mathbf{\nOn{Q}M}$ & 0.93 & 0.97 & 0.98 & 0.97 & 0.96 \\ 
		%   & $\mathbf{Se\nOn{Q}M}$ & 0.926 & 0.970 & 0.981 & 0.970 & 0.961 \\ 
		& $\mathbf{\nOn{C}M}$ & 0.93 & 0.97 & 0.99 & 0.98 & 0.97 \\ 
		%   & $\mathbf{Se\nOn{C}M}$ & 0.926 & 0.973 & 0.985 & 0.977 & 0.972 \\ 
		\\
		$\|f-f_{0}\|_{\infty}$ & Original & 297 & 88.9 & 36.5 & 20.4 & 13.2 \\ 
		%   & $\mathbf{M}$ & 270.582 & 67.134 & 26.603 & 17.624 & 12.702 \\ 
		& $\mathbf{\nOn{Q}M}$ & 271 & 67.1 & 26.6 & 17.6 & 12.7 \\ 
		& \textcolor{black}{$\mathbf{S\nOn{Q}M}$ }& \textcolor{black}{270 }& \textcolor{black}{66.7 }& \textcolor{black}{26.2 }& \textcolor{black}{17.4 }& \textcolor{black}{12.6 }\\ 
		& $\mathbf{\nOn{C}M}$ & 270 & 66.3 & 25.7 & 17.3 & 12.6 \\ 
		& \textcolor{black}{$\mathbf{S\nOn{C}M}$ }& \textcolor{black}{250 }& \textcolor{black}{63.0 }& \textcolor{black}{25.2 }& \textcolor{black}{17.3 }& \textcolor{black}{12.6 }\\ 
		\\
		$\|f_{u}-f_{l}\|_{\infty}$ & Original & 2297 & 484 & 190 & 109 & 72.0 \\ 
		%   & $\mathbf{M}$ & 1441.202 & 331.755 & 134.619 & 82.836 & 61.970 \\ 
		& $\mathbf{\nOn{Q}M}$ & 1441 & 332 & 135 & 82.8 & 62.0 \\ 
		%   & $\mathbf{Se\nOn{Q}M}$ & 1441.202 & 331.754 & 134.619 & 82.836 & 61.970 \\ 
		& $\mathbf{\nOn{C}M}$ & 1440 & 331 & 135 & 82.8 & 62.0 \\ 
		%   & $\mathbf{Se\nOn{C}M}$ & 1440.335 & 331.408 & 134.561 & 82.835 & 61.970 \\ 
		\hline
		\hline
		\multicolumn{7}{l}{\footnotesize{Notes: Based on \textcolor{black}{5,000} simulations. Nominal level of the confidence } }\\
		\multicolumn{7}{l}{\footnotesize{bands  is 95\%.  Confidence bands constructed by weighted bootstrap } }\\
		\multicolumn{7}{l}{\footnotesize{with standard exponential weights and 500 repetitions.} }
	\end{tabular}
\end{table}

Similar to the univariate case, we now explore the finite-sample improvements from enforcing shape restrictions via simulations calibrated to the production function application. The output,  $Y$, of each firm is generated by
$$
Y_i = \hat \gamma + \hat \beta_1 L_i + \hat \beta_2 K_i + \hat \sigma \epsilon_i, \ \ i = 1,\ldots,n,
$$
where $\hat \gamma$, $\hat \beta_1$, $\hat \beta_2$ and $\hat \sigma$ are calibrated to the least squares estimates and the residual standard deviation of this linear regression model in the production function data; $\epsilon_i$ are independent draws from the standard normal distribution; and $n$ is the sample size of the simulated data. The vector $(L, K)$ of labor and capital is drawn without replacement from the original data. The target function is $$
f_0(l,k) := \Ep[ Y \mid L = l, K =k] =  \hat \gamma +  \hat \beta_1 l +  \hat \beta_2 k,
$$
which is increasing and concave in the capital and labor inputs because $\hat \beta_1 > 0$ and $\hat \beta_2 > 0$.
We consider five sample sizes, $n \in \{100, 200, 500, 1{,}000, 1{,}638\}$, where $n = 1{,}638$ is the same sample size as in the empirical application. The results are based on 5,000 simulations. In each simulation we construct point and band estimates of $f_0$ using the same methods as in the empirical application.

Table~\ref{mc_prodfn} reports the same diagnostics as Table~\ref{bias_bw}. The operators $\mathbf{\nOn{Q}M}$ and $\mathbf{\nOn{C}M}$ perform similarly in this case. Both bring substantial gains in estimation and inference\textcolor{black}{, and the shifted variants bring additional gains. Shifting the $\mathbf{\nOn{C}M}$ operator in particular has a notable effect on estimation error for small sample sizes}. The operators reduce estimation error between 5\% and \textcolor{black}{31\%} and the width of the confidence band between 14\% and 37\% in the sup-norm, depending on the sample size.  
%These reductions are particularly remarkable for $\mathbf{\nOn{Q}M}$ because this operator is not a $d_{\infty}$-distance contraction.  
The operators also improve the coverage of the confidence bands, especially for the smaller sample sizes.   Indeed, enforcing the constraints compensates for the  undercoverage of the unconstrained estimates for most of the sample sizes considered.

%The concavity restriction performs slightly better than the quasi-concavity restriction, which is expected as the restriction is stronger. Not surprisingly, the operators composed with monotonicity peform the best, since they provide the strongest restrictions.

\section{Conclusion}\label{sec:con}

In this paper, we investigate a pool of shape-enforcing operators, including range, rearrangement, double Legendre-Fenchel, quasi-convexification, composition of rearrangement and double Legendre-Fenchel, and composition of rearrangement and quasi-convexification operators. We show that enforcing the shape restrictions through these operators improves point and interval estimators, and provide computational algorithms to implement these shape-enforcing operators. It would be useful to develop operators to enforce other shape restrictions, such as \textcolor{black}{supermodularity} or the Slutsky conditions for demand functions. We leave this extension to future research.

\section*{Acknowledgments}
We are very grateful to Simon Lee for kindly sharing the data for the production function application and to Roger Koenker for kindly making the Indian nutrition data available through his website. We thank the editor Garvesh Raskutti, two anonymous referees, Shuowen Chen and Hiroaki Kaido for comments. We gratefully acknowledge research support from the National Science Foundation and the Spanish State Research Agency MDM-2016-0684 under the Mar\'ia de Maeztu Unit of Excellence Program. Xi Chen is supported by NSF IIS-1845444. Part of this work was completed while Fern\'andez-Val was visiting CEMFI and NYU. He is grateful for their hospitality.

\newpage

\appendix
\section{Proofs}

\subsection*{Proof of Theorem~\ref{RR}} We first show that $\mathbf{R}$ satisfies the three properties of  Definition~\ref{def:SE_operators}.

(1) Reshaping: it holds because for any $f \in \ell^{\infty}(\X)$, $\mathbf{R} f \in \ell_{R}(\X)$ by construction.

(2) Invariance: it holds trivially because $\mathbf{R} f=f$ for any $f\in \ell_{R}(\X)$ by definition of $\mathbf{R}$.

(3) Order preservation: assume that $f,g \in \ell^{\infty}(\X)$ are such that $f\geq g$.  For any $x\in \mathcal{X}$ there are three possible cases.  (a) If $f(x)\geq \overline{f}$, then $\mathbf{R}f (x) = \overline{f} \geq \mathbf{R}g(x)$. (b) If $f(x)\leq \underline{f}$, then $g(x)\leq \underline{f}$, and $\mathbf{R}f (x)=\underline{f} =\mathbf{R}g(x)$. (c) if $\underline{f}<f(x)<\overline{f}$, then $\mathbf{R}f(x)= f(x)\geq \max(g(x),\underline{f}) = \mathbf{R}g(x)$ because $g(x) <\overline{f}$. Thus, $\mathbf{R}f(x)\geq \mathbf{R}g(x)$ for any $x\in \mathcal{X}$.

We next show Definition~\ref{def:dc} for $\rho = d_p$ for any $p \geq 1$. For any $f,g \in \ell^{\infty}(\X)$ assume without loss of generality  that  $f(x)\geq g(x)$ for some $x \in \X$. We need to show that $|\mathbf{R}f(x)-\mathbf{R} g(x)|\leq f(x)-g(x)$. There are five possible cases. (a) If $f(x)\geq g(x)\geq \overline{f}$, then $|\mathbf{R} f(x)-\mathbf{R} g(x)| = |\overline{f}-\overline{f}| = 0 \leq f(x)-g(x)$. (b) If $f(x)\geq \overline{f} \geq g(x)$, then $\mathbf{R}f(x) = \overline{f} \leq f(x)$, and $\mathbf{R} g(x)\geq g(x)$. By the order preservation property proved in (3), $0\leq |\mathbf{R}f(x)-\mathbf{R} g(x)|\leq f(x)-g(x)$.
(c) If $\overline{f}>f(x)\geq g(x)\geq \underline{f}$, then $\mathbf{R} f(x)=f(x)$ and $\mathbf{R}g(x)=g(x)$, and $|\mathbf{R} f(x)-\mathbf{R} g(x)| = f(x)-g(x)$. (d) If $\overline{f}>f(x)>\underline{f} \geq g(x)$, then $0\leq |\mathbf{R}f(x)-\mathbf{R}g(x)| = f(x)-\underline{f}\leq f(x)-g(x)$. (e) If $\underline{f}\geq f(x)\geq g(x)$, then $|\mathbf{R}f(x)-\mathbf{R}g(x)| = |\underline{f}-\underline{f}|=0\leq f(x)-g(x)$. \qed

\subsection*{Proof of Theorem~\ref{DFL}}

Before we prove Theorem~\ref{DFL}, we recall some useful geometric properties of the Legendre-Fenchel transform.

\begin{lemma}[Properties of Legendre-Fenchel transformation]\label{lemma:LF}
	Given a convex set $\X \subset \RR^k$, suppose that $f, g \in \ell^{\infty}_S(\X)$. Then:
	
	(1) Lower semi-continuity: $\mathbf{L}_{\X} f \in \ell^{\infty}_S(\RR^k)$.
	
	%\inote{Ye: check property (1) and add the proof}
	%
	%\inote{Ivan: already fixed}
	
	(2) Convexity: $\mathbf{L}_{\X} f$ is closed convex on $\mathbb{R}^k$.

	(3) Order reversing: If $f\geq g$, then
	$\mathbf{L}_{\X} f\leq
	\mathbf{L}_{\X}g $.
	
	(4) $d_{\infty}$-Distance reducing:
	$\|\mathbf{L}_{\X} f-
	\mathbf{L}_{\X}g \|_{\infty} \leq \|f-g\|_{\infty}$.
	
	% (4) [$L^\infty$ norm invariance for convex functions]: If f and g
	% are convex bounded functions on $V$, then
	% $|\mathscr{L}_{V,\mathbb{R}^k}(f)-
	% \mathscr{L}_{V,\mathbb{R}^k}(g)|_{\infty}=|f-g|_{\infty}$.
\end{lemma}

\begin{proof}[Proof of Lemma~\ref{lemma:LF}]
	(1) For any $\xi\in f^*(\X)$ and $\epsilon>0$, there must exist $x_0\in \mathcal{X}$ such that $\xi'x_0 - f(x_0) \geq \mathbf{L}_{\X}f(\xi) - \epsilon/2$. Then, for any $\xi_1$ such that $||\xi-\xi_1||_2\leq \min[1, \epsilon/(2||x_0||)]$, we have:
	\begin{multline*}
		\mathbf{L}_{\X}f(\xi_1)\geq \xi_1'x_0- f(x_0) = (\xi_1-\xi)'x_0+\xi'x_0-f(x_0)\geq -||\xi-\xi_1||_2 ||x_0|| + \mathbf{L}_{\X}f(\xi) - \frac{\epsilon}{2} \\
		\geq -\frac{\epsilon}{2} + \mathbf{L}_{\X}f(\xi) - \frac{\epsilon}{2} \geq \mathbf{L}_{\X}f(\xi) - {\epsilon}.
	\end{multline*}
	Hence, $ \mathbf{L}_{\X}f(\xi) $ is lower semi-continuous at $\xi$. Since $\xi$ can be arbitrary, we conclude that $\mathbf{L}_{\X}f$ is a lower semi-continuous function.

	Properties (2) and (3) are shown in Theorem 1.1.2 and  Proposition 1.3.1 in Chapter E of \cite{ConvAna:01}. For (4), it is easy to check that
	$\|\mathbf{L}_{\X} f - \mathbf{L}_{\X}g \|_\infty = \sup_{\xi \in \mathbb{R}^k} |\mathbf{L}_{\X}f (\xi)-\mathbf{L}_{\X} g(\xi)|\leq \sup_{\xi \in \mathbb{R}^k}\sup_{x\in \X}|\{\xi'x-f(x)\}-\{\xi'x-g(x)\}|\leq \|f-g\|_\infty$.
	%
	%
	%(1) For any $\epsilon>0$ and $p\in \mathbb{R}^k$, by definition of the Legendre-Fenchel operator, there exists a point $x$ such that $p'x-f(x)\geq \mathscr{L}_{V,\mathbb{R}^k}(f)(p)-\epsilon$. For any $p_1$ and $p_2$ $\in \mathbb{R}^k$ such that $p=\alpha_1 p_1+\alpha_2 p_2$ with $\alpha_1,\alpha_2\geq 0$ and $\alpha_1+\alpha_2=1$, by definition, $\mathscr{L}_{V,\mathbb{R}^k}(f)(p_1)\geq p_1'x-f(x)$ and $\mathscr{L}_{V,\mathbb{R}^k}(f)(p_2)\geq p_2'x-f(x)$. Therefore, $\alpha_1 \mathscr{L}_{V,\mathbb{R}^k}(f)(p_1)+\alpha_2 \mathscr{L}_{V,\mathbb{R}^k}(f)(p_2)\geq p'x-f(x)\geq \mathscr{L}_{V,\mathbb{R}^k}(f)(p)-\epsilon$.
	
	%Since $\epsilon$ can be arbitrarily small, $\mathscr{L}_{V,\mathbb{R}^k}(f)$ is convex.
	%
	%(2) By definition of $\mathscr{L}_{V,\mathbb{R}^k}$, if $f\geq g$, then for any $p\geq 0$, $\{p'x-f(x)\}\leq \{p'x-g(x)\}$. Thus, $\sup_{x\in V} \{p'x-f(x)\}\leq \sup_{x\in V}  \{p'x-g(x)\}$.
	%
	%(3)
\end{proof}

\begin{remark}
	The Legendre-Fenchel transform $\xi \mapsto \mathbf{L}_{\X} f(\xi)$ is locally Lipshitz, because any convex function is locally Lipshitz. Statement (4) of Lemma~\ref{lemma:LF} is known as Marshall's Lemma \cite{marshall1970}.
\end{remark}

% \begin{proof}
% I refer the proof to Rockafeller (1997).
% \end{proof}

%The above Lemma states the geometric properties of Legendre-Fenchel
%transformation. The function $\mathscr{L}_{V,\mathbb{R}^k}(f)$ is often referred as %the dual function of $f$.

Next, we derive some properties for the  Double Legendre-Fenchel transformation.

\begin{lemma}[Properties of $\mathbf{C}$-Operator]\label{lemma:DLF-basic}
	Given a convex set $\X \subset \RR^k$, suppose that $f \in \ell^{\infty}_S(\X)$. Then:
	
	(1) $\mathbf{C} f$ is the greatest convex minorant of $f$,
	i.e., the largest function $g \in \ell_C^{\infty}(\X)$ such that $g \leq f$.
	
	(2) If $\X$ is compact, for any $x\in\X$  there exist $d \leq k+2$ points
	$x_1,x_2,\ldots,x_d$ and scalars $(\alpha_1, \alpha_2, \ldots, \alpha_d) \in \Delta_{d-1}$, where  $ \Delta_{d-1}$ is the $(d-1)$-simplex $$\Delta_{d-1} := \left\{\alpha \in \mathbb{R}^{d}: \alpha_j\geq 0, j=1,2,\ldots,d, \sum_{j=1}^{d} \alpha_j =1\right\},$$  such that
	\begin{equation}
		\mathbf{C} f(x)=\sum_{i=1}^d \alpha_i f(x_i),
	\end{equation}
	where $x=\sum_{i=1}^d \alpha_i x_i$ and
	$f(x_i)=\mathbf{C}f (x_i)$, $1\leq i\leq d$.
	
	%\xnote{This statement is slightly different from the usual definition of ``convex minorant'', which is defined as $\mathscr{C}(f)(x)=\inf\left\{\sum_{i=1}^d \alpha_i f(x_i): \alpha_1>0,\alpha_2>0,\ldots, \alpha_d>0, \sum_{i=1}^d
	%\alpha_i=1, x=\sum_{i=1}^d \alpha_i x_d\right\}$? I am not sure if the inf is always attainable? }
	%\xnote{Ye: this is addressed, see the reference uploaded, Proposition 2.5.1. in Chapter B of \cite{ConvAna:01}}

	(3) We say that $f$ is convex at $x\in \X$ if there exists a supporting
	hyperplane with direction $\xi$ such that $f(\tilde x)\geq f(x)+\xi'(\tilde x-x)$ for
	all $\tilde x\in \X$. Then, $f(x)=\mathbf{C}f(x)$ if and only if $f$ is
	convex at $x$. Furthermore, if $f$ is convex at every $x\in \mathcal{X}$, then $f$ is a convex function.
	
	%\inote{Is this definition of convexity different to the one used in the definition of $\ell_C^{\infty}(\X)$ - Ye: it is a different definition "convex at a point", it is not globally convex}
\end{lemma}

\begin{proof}[Proof of Lemma~\ref{lemma:DLF-basic}]
	Statement (1): Recall that  $\mathbf{C}f:=\mathbf{L}_{f^*(\mathcal{X})}\circ \mathbf{L}_\mathcal{X}f$. We first show that $\mathbf{C}f \in \ell_C^{\infty}(\X)$ and $\mathbf{C}f \leq f$.  For any $g\in f^*(\mathcal{X})$,  $\mathbf{L}_{f^*(\mathcal{X})} g$ is a closed convex function by Lemma \ref{lemma:LF}(2), so that $\mathbf{C}f  \in \ell_C^{\infty}(\X)$. Let $f^*:=\mathbf{L}_\mathcal{X}f$. For any $x\in \mathcal{X}$, $\mathbf{C}f(x) = \mathbf{L}_{f^*(\mathcal{X})} f^*(x)=\sup_{\xi \in f^*(\mathcal{X})}\{\xi'x-f^*(\xi)\} \leq \sup_{\xi\in f^*(\mathcal{X})}(\xi'x- (\xi'x-f(x)))=f(x)$ because $f^*(\xi)\geq \xi' x-f(x)$ for any $\xi\in f^*(\mathcal{X})$ by definition of $f^*$.
	
	Next, we show that $\mathbf{C}f$ is the convex minorant of $f$, i.e. $\mathbf{C}f \geq h$ for any $h \in \ell_C^{\infty}(\X)$ such that $h \leq f$. If $h \in \ell_C^{\infty}(\X)$, 
	for any $x\in \mathcal{X}$, there exists $\xi\in f^*(\mathcal{X})$ such that 
	$h(\tilde x)\geq h(x)+\xi'(\tilde x-x)$ for all $\tilde x \in \mathcal{X}$. Since $h(\tilde x)\leq f(\tilde x)$, 
	\begin{equation}\label{eq:aux}
		\xi' x-h(x)\geq \xi' \tilde x-h(\tilde x)\geq \xi' \tilde x - f(\tilde x).
	\end{equation}
	By definition, $f^*(\xi) =\sup_{\tilde x \in \mathcal{X}} \{\xi'\tilde x-f(\tilde x)\}$, so that for any $\epsilon>0$, there must exist $\tilde x \in \mathcal{X}$ such that 
	$\xi'\tilde x-f(\tilde x)\geq f^*(\xi)-\epsilon$, which combined with \eqref{eq:aux} gives
	$\xi' x-h(x) \geq f^*(\xi)-\epsilon$, or rearranging terms, $\xi' x-f^*(\xi) \geq h(x) - \epsilon$.
	Then, $\mathbf{C}f(x)=\sup_{\tilde \xi \in f^*(\mathcal{X})}\{ \tilde \xi' x-f^*(\tilde \xi)\}\geq \xi' x - f^*(\xi)\geq h(x)-\epsilon$. Since $\epsilon$ can be arbitrarily small, we conclude that $\mathbf{C}f(x)\geq h(x)$.

	Statement (2): by Proposition 2.5.1. in Chapter B of \cite{ConvAna:01},
	\begin{equation}\label{eq:DFL}
		\mathbf{C}f(x) = \inf\left\{\sum_{j=1}^{k+2} \alpha_j f(x_j) :  \sum_{j=1}^{k+2}\alpha_j x_j = x, \alpha = (\alpha_1,\ldots,\alpha_{k+2})\in \Delta_{k+1}\right\},
	\end{equation}
	where $\Delta_{k+1}$ is the $(k+1)$-simplex.
	
	By \eqref{eq:DFL}, there exists a sequence $(x^t,\alpha^t)\in \mathcal{X}\times \Delta_k$ such that $\sum_{j=1}^{k+2}\alpha_j^t x_j^t = x$ and $\mathbf{C}f(x) \leq \sum_{j=1}^{k+2} \alpha_j^t f(x_j^t)+ \frac{1}{t}$. Since $\mathcal{X}\times \Delta_{k+1}$ is compact,  there must exist a limit point $(x^0,\alpha^0) \in \mathcal{X}\times \Delta_{k+1}$ of the sequence $(x^t,\alpha^t)$ such that $\sum_{j=1}^{k+2} \alpha^0_j x^0_j =\lim_{t\rightarrow \infty} \sum_{j=1}^{k+2}  \alpha^t_j x^t_j =x $, and by lower semi-continuity of $f$, $\sum_{j=1}^{k+2}  \alpha^0_j f(x^0_j) \leq \lim_{t\rightarrow \infty} \sum_{j=1}^{k+2}  \alpha^t_j f(x^t_j) = \mathbf{C}f(x)$. Then, it follows from \eqref{eq:DFL} that $\sum_{j=1}^{k+2}  \alpha^0_j f(x^0_j) =   \mathbf{C}f(x)$.
	% since $\mathbf{C}f(x) = \inf\{\sum_{j=1}^k \alpha_j f(x_j)| \sum_{j=1}^k\alpha_j x_j = x, \alpha = (\alpha_1,\ldots,\alpha_k)\in \Delta_k\}$.
	Equivalently, $\sum_{j=1}^{k+2} \alpha_j^0 x_j^0 =x$, and
	\begin{equation}\label{eq:DLF-1}\sum_{j=1}^{k+2} \alpha_j^0 f(x^0_j) =\mathbf{C}f(x).\end{equation}
	
	Let $(\alpha_1, \ldots, \alpha_d)$ and $(x_1,\ldots, x_d)$ denote the subsets of $(\alpha^0_1, \ldots, \alpha^0_{k+2})$ and $(x^0_1,\ldots, x^0_{k+2})$ corresponding to the components with $\alpha_j^0 >0$, where $j = 1,2,\ldots, k+2$ and $d\leq k+2$.  Next, we show that $f(x_j)=\mathbf{C}f (x_j)$, for $1\leq j\leq d$. By statement (1), since $\mathbf{C}f$ is the convex minorant of $f$, it follows that $x \mapsto \mathbf{C}f(x)$ is convex and $\mathbf{C}f(x)\leq f(x)$ for all $x\in \X$.  In particular,
	$$\mathbf{C}f(x)\leq \sum_{j=1}^{d} \alpha_j \mathbf{C}f(x_j)\leq \sum_{j=1}^{d} \alpha_j f(x_j).$$
	By (\ref{eq:DLF-1}), the two inequalities imply that $$\sum_{j=1}^{d} \alpha_j\mathbf{C}f(x_j) = \sum_{j=1}^{d} \alpha_j f(x_j).$$ Since $\mathbf{C}f(x_j)\leq f(x_j)$ for all $j=1,2,\ldots,d$, it follows that $f(x_j) = \mathbf{C}f(x_j)$  for all $j=1,2,\ldots,d$.
	
	This completes the proof of statement (2).
	
	%\inote{I do not see how this shows that $f(x_i)=\mathbf{C}f (x_i)$, $1\leq i\leq d$|||| Ye- fixed in the above}
	
	Statement (3):  $\mathbf{C}f(x)\leq f(x)$ for all $x\in \X$ by (1). If there exists  a $\xi \in \RR^k$ such that  $f(\tilde x)\geq f(x)+\xi'(\tilde x-x)$ for any $\tilde x\in \X$, then $g(\tilde x)=f(x)+\xi'(\tilde x-x)$ is a convex function that lies below $f$. By (1), $\mathbf{C}f(x)\geq g(x)=f(x)$. Therefore, $\mathbf{C}f(x)=f(x)$. On the other hand, suppose that $\mathbf{C}f(x)=f(x)$. Since $x \mapsto \mathbf{C}f(x)$ is convex on $\X$,  there must exist $\xi \in \mathbb{R}^k$ such that  $\mathbf{C}f(\tilde x) \geq \mathbf{C}f(x)+\xi'(\tilde x-x)$ for any $\tilde x \in \X$. By definition of greatest convex minorant, $f(\tilde x)\geq \mathbf{C}f(\tilde x) \geq \mathbf{C}f(x)+\xi'(\tilde x-x)= f(x)+\xi'(\tilde x-x)$ for any $\tilde x \in \X$. So  $\mathbf{C}f(x)=f(x)$ implies that $f$ is convex at $x$.
	
	If $f$ is convex at every $x\in \mathcal{X}$, then by the results above, $f(x)=\mathbf{C}f(x)$ for every $x\in \mathcal{X}$. That is, $f=\mathbf{C}f$, which implies that $f$ is convex on $\mathcal{X}$ because $\mathbf{C}f$ is convex on $\mathcal{X}$.
\end{proof}

Theorem~\ref{DFL} follows from the properties in Lemmas \ref{lemma:LF} and~\ref{lemma:DLF-basic}.
%\begin{proof}[Proof of Theorem~\ref{DFL}]
The properties (1) and (3) in Definition~\ref{def:SE_operators} are implied by properties (2) and (3) of Lemma~\ref{lemma:LF} applied to $\mathbf{L}_{\X}f$ and using that $\mathbf{L}_{\X}f \in \ell_S^{\infty}(\RR^k)$ by property (1) of Lemma~\ref{lemma:LF}.
The property (2) in Definition~\ref{def:SE_operators} is implied by property (3) of Lemma~\ref{lemma:DLF-basic}. Moreover, the $d_{\infty}$-contraction property is given by property (4) in Lemma~\ref{lemma:LF} again applied to $\mathbf{L}_{\X}f$ and using that $\mathbf{L}_{\X}f \in \ell_S^{\infty}(\RR^k)$ by property (1) of Lemma~\ref{lemma:LF}. \qed
%
% \inote{Does the proof presume  that $\mathbf{L}_{\X} f \in \ell^{\infty}_S(\RR^k)$|||| Ye- Yes, it is a prerequisite for Lemma 7 and Lemma 8.}
%
%
%
%
%(1) and (3) are implied by the order preserveing and $L^\infty$ contraction properties of $\mathscr{L}_{V,\mathbb{R}^k}$ stated in Lemma~\ref{lemma:LF}. (2) follows directly by the statement (2) of Lemma~\ref{lemma:DLF-basic}.
%\end{proof}

%\subsection*{Lemma 3}

%[RA] [Leave to Ivan, Refer to RA works by Chernozhukov, Fernandez-val and Galichon]
%refers to Rearrangement works...

\subsection*{Proof of Theorem~\ref{composition}}

%\begin{proof}[Proof of Lemma~\ref{composition}]

We start by demonstrating that the $\mathbf{C}$-operator on a rectangle can be computed separately at each face of the rectangle.

\begin{definition}[$\mathbf{C}$-Operator Restricted to a Face of a Rectangle]
	%Suppose the domain $\mathcal{X}$ is a rectangle with each faces being
	%either parallel or perpendicular to $e_1,e_2,\ldots,e_k$ in $\mathcal{R}^k$.
	%A regular rectangle can be characterized by the Cartesian product $[a_1,b_1]\times \ldots\times [a_k,b_k]$, which is a special case of rectangles in $\mathbb{R}^k$.
	
	% \xnote{Difference between rectangle and regular rectangle?}
	
	For any regular rectangular set $\mathcal{X}:=[a_1,b_1]\times \ldots\times [a_k,b_k]$,  a set  $\F_m$ is an $m$-dimensional face of $\mathcal{X}$ if there exists a  set of indexes $i(\F) \subset \{1,2,\ldots,k\}$ with $m$ elements such that  $\F_m= \{x \in \X: x_{j}\in [a_j, b_j], \textrm{ for any } j\in i(\F), x_{j} \in \{a_j,b_j\}, \textrm{ for any } j \notin i(\F)\}$. For every $x \in \F_m$, we can define the $\mathbf{C}$-operator restricted to the face $\F_m$ by applying  the Legendre-Fenchel transform only to each of the coordinates of $x$ that are in $i(\F_m)$. Thus, let $$\mathbf{L}_{\X \mid \F_m}f(\xi) :=  \sup_{x \in \F_m}\{\xi'x_{i(\F_m)} - f(x_{i(\F_m)},x_{i^c(\F_m)})\},$$  where we partition $x$ into the coordinates with indexes in $i(\F_m)$, $x_{i(\F_m)}$, and the rest of the coordinates, $i^c(\F_m)$. Then, the $\mathbf{C}$-operator restricted to the face $\F_m$ of $ f \in \ell_S^\infty(\X)$ is
	$$
	\mathbf{C}_{\X \mid \F_m} f(x)  := \mathbf{L}_{f^*(\X \mid \F_m)}  \circ \mathbf{L}_{\X \mid \F} f(x),
	$$
	where $f^*(\X \mid \F_m) := \{\xi \in \RR^m:\mathbf{L}_{\X \mid \F_m}f(\xi)  < \infty \}$.
	%$\mathscr{C}_{|F}$ can be viewed as the convexification operator which applies to functions defined on $\Pi_{j\in I}[a_j, b_j]$.
	Moreover, by Proposition 2.5.1 of \cite{ConvAna:01}, $\mathbf{C}_{\X \mid \F}  f(x)$ is a linear combination of the $f$-images of $m+2$ elements of $\F_m$,   that is
	$$\mathbf{C}_{\X \mid \F} f(x) = \inf \left\{\sum_{j=1}^{m+2} \alpha_j f(x_j) :  x_j\in \F_m, (\alpha_1, \ldots, \alpha_{m+2}) \in \Delta_{m+1} \right\}, \ \ \text{ for any } x \in \F_m,$$
	where $ \Delta_{m+1}$ is the $(m+1)$-simplex. 

\end{definition}

\begin{lemma}[$\mathbf{C}$-Operator on a Regular Rectangular Set]\label{remark:Drectagle} For any regular rectangular set  $\mathcal{X}$ and  $f \in \ell^{\infty}_{S}(\X)$, if $x \in \F_m$ with $m > 0$, then
	$$
	\mathbf{C} f(x) = \mathbf{C}_{\X \mid \F_m} f(x).
	$$
\end{lemma}

\begin{proof}[Proof of Lemma~\ref{remark:Drectagle}] Suppose that $\X$ is a regular rectangle in $\RR^k$. Let $\F_m$ be a face of $\X$ with dimension $m$ such that $x \in \F_m$.  The result follows from the following facts:
	
	First, $x \mapsto \mathbf{C}f(x)$ is a convex function and lies below $x \mapsto f(x)$ on $\X$, so that $x \mapsto \mathbf{C}f(x)$ is a convex function and lies below $x \mapsto f(x)$ on $\mathcal{F}_m \subset \X$. By definition, $\mathbf{C}_{\X \mid \F_m} f$ is the convex minorant of $f$ restricted on $\mathcal{F}_m$, i.e., the largest possible convex function lying below $f$ restricted on $\mathcal{F}_m$. Therefore, it must be that $\mathbf{C}_{\X \mid \F_m} f\geq \mathbf{C}f(x)$ for all $x\in \mathcal{F}_m$.
	
	%$\mathbf{L}_{\X} f(\xi) \geq\mathbf{L}_{\X \mid \F_m}f$ because $\mathbf{L}_{\X} f(\xi) =\sup_{x\in \X} \{x'\xi-f(x)\}$, $\mathbf{L}_{\X \mid \F_m} f=\sup_{x\in \F_m} \{x'\xi-f(x)\}$ and $\F_m \subset \X$.

	%$\mathbf{C}_{\X \mid \F_m} f$ is the convex hull of $f$ restricted on $\F_m$.
	%It should be no less than the convex hull of $f$ on $\X$ restricted
	%on $\F_m$. Therefore, $\mathbf{C}f \leq \mathbf{C}_{\X \mid \F_m} f$.
	
	%\inote{Ye: I do not follow the last part of previous paragraph. Please clarify. also the argument seems to be repeat in (A.3), which I do not follow either||| Ye: elaborated in the above}

	Second, by statement (2) of Lemma~\ref{lemma:DLF-basic}, for any $x\in \F_m$, there exist $d\leq k+2$ points
	$x_1,\ldots,x_d$ and $\alpha_i>0$, $1\leq i \leq d$, $\sum_{i=1}^d
	\alpha_i=1$, such that $\mathbf{C}f (x_i)=f(x_i), \sum_{i=1}^d
	\alpha_i x_i=x$, and $\sum_{i=1}^d \alpha_i f(x_i)=f(x)$. It must be that $x_i \in \F_m$, $1\leq i \leq d$, since $x\in \F_m$.

	Third, by definition of greatest convex minorant, on the face $\F_m$, $f(x)\geq  \mathbf{C}_{\X \mid \F_m}f(x)$ for any $x\in \F_m$. Since $ \mathbf{C}_{\X \mid \F_m}f(x)$ is the convex minorant of $f(x)$ restricted on $\F_m$, and $\mathbf{C}f$ is a convex function on $\F_m$, it follows that $\mathbf{C}_{\X \mid \F_m}f(x) \geq \mathbf{C}f(x)$ for any $x\in \F_m$. Therefore,
	\begin{equation}\label{eq:chain1}f(x)\geq \mathbf{C}_{\X \mid \F_m}f(x)\geq \mathbf{C} f(x)\end{equation} for
	all $x\in \F_m$.
	
	Fourth, for each $x_i$, $i=1,2,\ldots,d$, we know that $f(x_i) = \mathbf{C}f(x_i)$. Applying equation (\ref{eq:chain1}), it must be that
	$f(x_i)=\mathbf{C}_{\X \mid \F_m} f(x_i)=\mathbf{C} f (x_i)$. Therefore,
	\begin{equation}\label{eq:chain2}
		\mathbf{C} f (x)=\sum_{i=1}^d \alpha_i f(x_i) =
		\sum_{i=1}^d \alpha_i \mathbf{C}_{\X \mid \F_m} f(x_i)\geq \mathbf{C}_{\X \mid \F_m}f(x),
	\end{equation}
	where the inequality follows from convexity of
	$x \mapsto \mathbf{C}_{\X \mid \F_m}f(x)$.
	
	%\inote{Ye: I do not follow why $f(x_j)=\mathbf{C}_{\X \mid \F_m} f(x_j)$. Please clarify - Ye: I made some changes}
	
	Combining inequalities (\ref{eq:chain1}) and (\ref{eq:chain2}), we conclude that $\mathbf{C}_{\X \mid \F_m}f (x)=\mathbf{C} f(x)$.
\end{proof}

Before stating the main proof of Theorem~\ref{composition}, we require a lemma to show that $\mathbf{M}$ maps a function in $ \ell^{\infty}_{S}(\X)$ to $ \ell^{\infty}_{S}(\X)$.

\begin{lemma}\label{lemma:M-S}
	Suppose $\mathcal{X}=[0,1]^k$. The rearrangement operator maps any function $f\in  \ell^{\infty}_{S}(\X)$ to $ \ell^{\infty}_{S}(\X)$.
\end{lemma}

\begin{proof}
	First, it is easy to see that for any $f_1,f_2\in  \ell^{\infty}_{S}(\X)$ and $a,b\geq 0$, $af_1+bf_2\in  \ell^{\infty}_{S}(\X)$. Therefore, to show that $\mathbf{M}$ maps a function $f \in  \ell^{\infty}_{S}(\X)$ to $ \ell^{\infty}_{S}(\X)$, it suffices to show that $\mathbf{M}_\pi$ maps a function $f \in  \ell^{\infty}_{S}(\X)$ to $ \ell^{\infty}_{S}(\X)$, since $\mathbf{M} f = \sum_{\pi \in \Pi}  \mathbf{M}_{\pi} f/|\Pi|$. Denote $\pi = (\pi_1,\ldots,\pi_k)$, so $\mathbf{M}_\pi = \mathbf{M}_{\pi_1}\circ \cdots \circ \mathbf{M}_{\pi_k}$. For any function $f\in\ell^{\infty}_{S}(\X)$ and $j=1,2,\ldots,k$, we would like to prove that $\mathbf{M}_{j}f\in \ell^{\infty}_{S}(\X)$. If the statement above is true, then it follows that $\mathbf{M}_\pi f =  \mathbf{M}_{\pi_1}\circ \cdots \circ \mathbf{M}_{\pi_k}f \in \ell^{\infty}_{S}(\X)$. Consequently, the conclusion of the lemma is true.
	
	%\inote{Ye: does the composition of two lower semicontinuous functions yield a lower semicontinuous function?}
	
	Second, we prove that  for any function $f\in\ell^{\infty}_{S}(\X)$ and $j=1,2,\ldots,k$, $\mathbf{M}_{j}f\in \ell^{\infty}(\X)$. Without loss of generality, we can assume $j=1$. By definition, $$\mathbf{M}_1 f(x(1),x(-1)) =  \inf\left\{y \in \RR : \int_{\X(1)} 1\{f(t, x(-1)) \leq y\} dt \geq x(1)  \right\}.$$ For any $x(1)\in [0,1]$ and $x(-1)\in [0,1]^{k-1}$,
	$
	\int_{\X(1)} 1\{f(t, x(-1)) \leq y_{max}\} dt = 1\geq x(1),
	$
	and
	$
	\int_{\X(1)} 1\{f(t, x(-1)) \leq (y_{min}-\epsilon)\} dt =0< x(1),
	$
	where $y_{max}=\sup_{x\in \mathcal{X}} f(x)$, $y_{min}=\inf_{x\in \X}f(x)$ and $\epsilon>0$ can be any arbitrarily small constant. Since $f\in \ell^{\infty}_{S}(\X)$, $y_{min}$ and $y_{max}$ exist. Therefore, $ \inf\left\{y \in \RR : \int_{\X(1)} 1\{f(t, x(-1)) \leq y\} dt \geq x(1)  \right\}$ must be well defined and bounded by $y_{max}$ from above and by $y_{min}-\epsilon$ from below. We conclude that $\mathbf{M}_1 f\in \ell^\infty(\X)$.
	
	Third, we show that $\mathbf{M}_1f \in\ell^{\infty}_{S}(\X)$  if $f\in\ell^{\infty}_{S}(\X)$.
	We prove this by contradiction: suppose that $\mathbf{M}_1 f$ is not lower semi-continuous at a point $x^0\in \X$. There must exist a sequence $x^1,\ldots,x^n,\ldots$ in $\mathcal{X}$ and a constant $\epsilon>0$ such that $\|x^n-x^0\|\rightarrow 0$ as $n\rightarrow \infty$ and $\mathbf{M}_1 f(x^n)\leq \mathbf{M}_1 f(x^0)-\epsilon$ for all $n\geq 1$. Let $y^0 := \mathbf{M}_1 f(x^0)$. By definition of $\mathbf{M}_1 f$, it must be that $\int_{\X(1)} 1\{f(t, x^0(-1)) \leq y^0-\frac{\epsilon}{2}\} dt < x^0(1) $, and  $\int_{\X(1)} 1\{f(t, x^n(-1)) \leq y^0-\epsilon\} dt \geq  x^n(1) $ for all $n\geq 1$. For any $t\in [0,1]$, since $f\in\ell^{\infty}_{S}(\X)$ and $(t, x^{n}(-1))\rightarrow (t, x^0(-1))$, $\lim_{n\rightarrow \infty} f(t, x^n(-1))\geq f(t, x^0(-1))$. Therefore, there exists $N$ large enough such that $f(t, x^n(-1))\geq f(t, x^0(-1))-\frac{\epsilon}{2}$ for all $n\geq N$. Consequently, $1\{f(t, x^n(-1)) \leq y^0-\epsilon\} \leq 1\{f(t, x^0(-1)) \leq y^0-\frac{\epsilon}{2}\}$ for all $n\geq N$. Then, $$  {\lim\sup_{n\rightarrow \infty}   1\{f(t, x^n(-1)) \leq y^0-\epsilon\} \leq 1\{f(t, x^0(-1)) \leq y^0-\frac{\epsilon}{2}\}}$$ holds for all $t$.
	%Denote $g(x_1) = 1$ for all $x_1\in [0,1]$. It is easy to see that $g(x_1')\geq 1\{f(x_1', x^n_{-1}) \leq y^0-\epsilon\}$ and $g(x_1')\geq 1\{f(x_1', x^0_{-1}) \leq y^0-\frac{\epsilon}{2}\}$ for all $x_1'\in [0,1]$.
	%, so $g(\cdot)$ is an envelop function of $f(\cdot, x^0_{-1})$ and $f(\cdot, x^n_{-1})$ for all $n\geq 1$.
	By reverse Fatou's Lemma,
	\begin{multline}\label{eq:fatou}
		\lim\sup_{n\rightarrow \infty }\int_{\X(1)} 1\{f(t, x^n(-1)) \leq y^0-\epsilon\} dt\leq \int_{\X(1)} \lim\sup_{n\rightarrow \infty } 1\{f(t, x^n(-1)) \leq y^0-\epsilon\} dt \\\leq  \int_{\X(1)} 1\{f(t, x^0(-1)) \leq y^0-\frac{\epsilon}{2}\} dt.
	\end{multline}
	However, $\int_{\X(1)} 1\{f(t, x^0(-1)) \leq y^0-\frac{\epsilon}{2}\} dt<x(1)^0$, while $\lim\sup_{n\rightarrow \infty }\int_{\X(1)} 1\{f(t, x^n(-1)) \leq y^0-\epsilon\} dt\geq \lim\sup_{n\rightarrow \infty}  x^n(1)=x^0_1$. Hence,
	$$\lim\sup_{n\rightarrow \infty }\int_{\X(1)} 1\{f(t, x^n(-1)) \leq y^0-\epsilon\} dt \geq x^0_1>\int_{\X(1)} 1\{f(t, x^0(-1)) \leq y^0-\frac{\epsilon}{2}\} dt,$$ which contradicts (\ref{eq:fatou}).
	Therefore, we conclude that $\mathbf{M}_1 f \in\ell^{\infty}_{S}(\X)$ if $f \in\ell^{\infty}_{S}(\X)$.
	%\inote{Ye: can we simplify the proof using that $x_1 \mapsto \mathbf{M}_1f (x_1, x_{-1})$ is  nondecreasing and left-continuous. This follows because it is a quantile function}
\end{proof}

We now start the proof of Theorem~\ref{composition}.

(1) We first show that $\mathbf{CM}$ satisfies the reshaping property (1) of  Definition~\ref{def:SE_operators}.

%The second and third properties of  Definition~\ref{def:SE_operators} and the $d_p$-distance contraction for any $p \in [1,\infty]$ follow from Theorems \ref{DFL} and~\ref{RA}.

% \inote{Ye: Elaborate more on the other properties - later|||Ye: add (2)-(4)}

We know that $\mathbf{CM}f = \mathbf{C} (\mathbf{M} f)$. By Lemma~\ref{lemma:M-S}, for any $f\in \ell^{\infty}_{S}(\X)$, $\mathbf{M}f\in  \ell^{\infty}_{S}(\X)$. Consequently, $\mathbf{M}f\in  \ell^{\infty}_{S}(\X)\cap \ell^{\infty}_{M}(\X)$.

We use induction to prove that $\mathbf{C} f \in \ell^{\infty}_{CM}(\X)$ for any $f \in  \ell^{\infty}_{S}(\X)\cap\ell^{\infty}_{M}(\X)$, where $\X \subset \RR^k$ is a regular rectangular set. Without loss of generality, assume that $\X=[0,1]^k$. Since $\mathbf{C} f \in \ell^{\infty}_{C}(\X)$ by Theorem~\ref{DFL}, we only need to show that $\mathbf{C} f \in \ell^{\infty}_{M}(\X)$.

For dimension $k=1$, $\X$ is a closed interval. We prove that $\mathbf{C}f$ is nondecreasing. Assume, by contradiction, that there exists a pair of points $x,x' \in \X$ such that $x<x'$ and $\mathbf{C}f(x)>\mathbf{C}f(x')$. Let $\underline{x}$ be the left end point of the interval $\X$. By convexity, $\mathbf{C}f(\underline{x})\geq
\mathbf{C}f(x)>\mathbf{C}f(x')$. By Lemma~\ref{remark:Drectagle},
$\mathbf{C}f(\underline{x})=\mathbf{C}_{\X \mid \F_0} f(\underline{x})=  f(\underline{x})$. By statement (2) of Lemma~\ref{lemma:DLF-basic},  there exist $x_1,\ldots,x_d\in \mathcal{X}$ and $\alpha_1,\ldots,\alpha_d > 0 , \sum_{j=1}^d \alpha_j = 1$  such that
$\mathbf{C}f(x') = \mathbf{C} f (x')=\sum_{j=1}^d \alpha_j   f(x_j)$. Since $ f$ is nondecreasing, we have $\sum_{j=1}^d \alpha_j  f(x_j)\geq\sum_{j=1}^d
\alpha_j  f(\underline{x})= f(\underline{x})$, which contradicts that $\mathbf{C}f(\underline{x})>\mathbf{C}f(x')$.  Hence, for any $x<x'$, it must be that $\mathbf{C}f(x)\leq \mathbf{C}f(x')$. We conclude that $x \mapsto \mathbf{C}f(x)$ is  nondecreasing.

Suppose that $x \mapsto \mathbf{C}f(x)$ is nondecreasing for $(k-1)$-dimensional regular rectangles, $k\geq 2$. Let $\X$ be a $k$-dimensional rectangle.
Assume, by contradiction, that there exists $x\leq x'$ ($x\neq x'$) such that $\mathbf{C}f(x)>\mathbf{C}f(x')$. Consider the radial originated from $x'$ that passes through  $x$, denoted as $L$. $L$ can be written as $\{z \in \RR^k : z=\gamma x'+(1-\gamma)x, \gamma\leq 1\}$. Therefore, there exists a $\gamma_0\leq 0$ such that $\gamma x'+(1-\gamma)x\in \X\cap L$ if and only if $1\geq \gamma \geq \gamma_0$. Denote $\underline{l} = \gamma_0 x'+(1-\gamma_0)x$.  By convexity of $x \mapsto \mathbf{C}f(x)$, it must be that
\begin{equation}\label{eq:CM-1}\mathbf{C}f(\underline{l})\geq
	\mathbf{C}f(x)>\mathbf{C}f(x').\end{equation}
By statement (2) of Lemma~\ref{lemma:DLF-basic}, there are $d$
points $x_1,\ldots,x_d\in \X$ and $\alpha_1,\ldots,\alpha_d>0$, $\sum_{i=1}^d\alpha_i=1$, such that $\sum_{j=1}^d \alpha_j x_i = x'$ and  $\sum_{j=1}^d \alpha_j f(x_j)=\mathbf{C}f(x')$. The point $\underline{l}$ must be on a $k-1$ dimensional face of $\X$, denoted by $\F_{k-1}$. Since $\X=[0,1]^k$, $\F_{k-1}$ can be expressed as $A_1\times A_2\times\ldots\times A_k$, where $A_i=[0,1]$ for $i\in i(\F_{k-1})$, $i(\F_{k-1})\subset\{1,2,\ldots,k\}$, $|i(\F_{k-1})|=k-1$, and $A_i = \{0\}$ or $\{1\}$ if $i\notin i(\F_{k-1})$. Without loss of generality, we can assume that $i(\F_{k-1})=\{1,2,\ldots,k-1\}$. Denote $A_k=\{w\}$, so $w=0$ or $1$.

Let $s$ be the projection mapping from $\X=[0,1]^k$ to $\F_{k-1}$, so $s: (x(1),\ldots,x(k))\mapsto (x(1),\ldots,x(k-1),w)$ for any $x(1),\ldots,x(k)\in [0,1]$. Since $\underline{l} \in \F_{k-1}$, it must be that $s(\underline{l}) = \underline{l}$. If $w=0$,  $s(x)\leq x$ for any $x\in \X$. Therefore, $s(x_1)\leq x_1,\ldots, s(x_{d})\leq x_d$. Then, since $x \mapsto f(x)$ is nondecreasing, $f(x_i)\geq f(s(x_i))$ for all $i=1,2,\ldots,d$. If $w=1$, then any point $x \in \F_{k-1}$ satisfies $x(k)=1$, including $\underline{l}$. Since $x'\geq \underline{l}$, the $k^{th}$ entry of $x'$ must equal to $1$. By $\sum_{j=1}^d \alpha_j x_j = x'$, $x_j\in [0,1]^k$, it must be that the $k^{th}$ entry of $x_j$ equals to $1$ for all $i=1,2,\ldots,d$. Therefore, $s(x_j) = x_j$, and $s(x') = x'$. Therefore, regardless of the value of $w$, $x_j \geq s(x_j)$, $x'\geq s(x')$.
Since $x'\geq \underline{l}$, it must be that $\underline{l}\leq s(x') = \sum_{j=1}^d \alpha_j s(x_j)$. By Lemma 9, $ \mathbf{C}_{\X|\F} f(\underline{l})=\mathbf{C}f(\underline{l})$ and by (\ref{eq:CM-1}),
\begin{multline*}
	\mathbf{C}_{\X|\F}f(\underline{l}) =\mathbf{C}f(\underline{l}) > \mathbf{C}f(x') =\sum_{i=1}^d \alpha_i f(x_i) \geq \sum_{i=1}^d \alpha_i f(s(x_i)) \\ \geq\sum_{i=1}^d \alpha_i \mathbf{C}_{\X|\F}  f(s(x_i))\geq   \mathbf{C}_{\X|\F}f(s(x')),
\end{multline*}
where the second inequality holds by monotonicity of $x \mapsto f(x)$, the third inequality by $\mathbf{C}_{\X|\F}$ being the convex minorant of $f$, and the fourth by convexity of $x \mapsto \mathbf{C}_{\X|\F} f (x)$.
Therefore, \begin{equation}\label{eq:CM-2}
	\mathbf{C}_{\X|\F} f(\underline{l})> \mathbf{C}_{\X|\F}f(s(x')).
\end{equation}

By induction, $\mathbf{C}_{\X|\F}f$ restricted on the $k-1$ dimensional regular rectangle $\F$ is nondecreasing. Since $s(x')\geq \underline{l}$, it must be that $\mathbf{C}_{\X|\F}f(s(x'))\geq \mathbf{C}_{\X|\F}f(\underline{l})$, which contradicts (\ref{eq:CM-2}). Hence, the induction is complete. $x \mapsto \mathbf{C}f(x)$ is nondecreasing if $x \mapsto f(x)$ is nondecreasing. Therefore, for any $f\in\ell^\infty_{S}(\X)$, $\mathbf{CM}f$ is monotonically increasing.

%\inote{Ye: I do not follow the previous 2 paragraphs. Please clarify}

We next show that $\mathbf{CM}$ satisfies the rest of the properties of Definition~\ref{def:SE_operators} and distance reduction.

(2) To show invariance, note that if $f \in \ell^{\infty}_{CM}(\X)$, then $\mathbf{M} f = f$ by Theorem~$\ref{RA}$, and therefore $\mathbf{CM} f = \mathbf{C} (\mathbf{M} f) =   \mathbf{C} f = f $ by definition of $\mathbf{CM}$ and Theorem~\ref{DFL}.

(3) Similarly, $\mathbf{CM}$ is order preserving because if $f\geq g$ then $\mathbf{M} f\geq \mathbf{M}g$ by Theorem~\ref{RA}, and therefore $\mathbf{CM}f = \mathbf{C}(\mathbf{M}f) \geq \mathbf{C}(\mathbf{M}g) = \mathbf{CM}g$ by  definition of $\mathbf{CM}$ and Theorem~\ref{DFL}.

% $\mathbf{C}$ and $\mathbf{M}$ are both order preserving By Lemmas \ref{DFL} and~\ref{RA}
%
%, $\mathbf{C}$ and $\mathbf{M}$ are both order preserving. Then, if $f\geq g$, it follows that $\mathbf{M}(f)\geq \mathbf{M}(g)$, and $\mathbf{C}(\mathbf{M}(f))\geq \mathbf{C}(\mathbf{M}(g))$. That said,
%$\mathbf{CM}(f)\geq \mathbf{CM}(g)$.

(4) Finally, $\mathbf{CM}$ is a $d_\infty-$distance contraction because $$d_\infty(\mathbf{CM}f , \mathbf{CM}g) = d_\infty(\mathbf{C}[\mathbf{M}f] , \mathbf{C} [\mathbf{M}g]) \leq d_\infty(\mathbf{M}f, \mathbf{M}g)\leq d_\infty(f,g),$$ where the first equality follows from definition of $\mathbf{CM}$, the first inequality by Theorem~\ref{DFL}, and the second inequality by Theorem~\ref{RA}.
%
%By Lemmas \ref{DFL} and~\ref{RA}, $\mathbf{C}$ and $\mathbf{M}$ are both $d_\infty-$ distance contration. Therefore,
%$d_\infty(\mathbf{CM}(f) , \mathbf{CM}(g))\leq d_\infty(\mathbf{M}(f), \mathbf{M}(g))\leq d_\infty(f,g)$. That said, $\mathbf{CM} $ is $d_\infty-$distance contraction.
%%(3) Statement (3) follows immediately from statement (2).
\qed

%\end{proof}

\subsection*{Proof of Theorem~\ref{QC}} 
We start with a lemma establishing that the operator $\mathbf{Q}$ is well-defined. 
\begin{lemma}[Properties of Operator $\mathbf{Q}$]\label{lemma:Q-def}
	For any convex  and compact set $\X \subset \RR^k$, the operator $\mathbf{Q}$ defined in \eqref{def:Q} is well-defined in that  the minimum of the set $\left\{ y \in \RR : x \in \mathrm{conv}[\mathcal{I}_f(y)] \right\}$ exists for all $x\in \X$ and $\mathbf{Q} f \in \ell^\infty_S(\X)$ for any $f \in  \ell^\infty_S(\X)$.
\end{lemma}

\begin{proof}
	%Let $\mathcal{I}_f(y) = \{x \in \X: f(x)\leq y\}$ be the lower contour set of $f$ at  level $y$. 
	Define $\mathcal{Q}_f(x):=\{y \in \RR : x\in \mathrm{conv}[\mathcal{I}_f(y)] \}$. We first show that $\min  \mathcal{Q}_f(x)$ exists.
	Obviously, $\inf \mathcal{Q}_f(x) \geq f(x)>-\infty$, because  $\mathcal{I}_f(y) \subset \mathrm{conv}[\mathcal{I}_f(y)]$ and $x\in \mathcal{I}_f(f(x))$. Hence, $\mathcal{Q}_f(x)$ is bounded from below. Let $y^0 := \inf \mathcal{Q}_f(x)$. 
	We need to show that $\min \mathcal{Q}_f(x) = y^0$, i.e,  there exists a sequence of $y^1,\ldots,y^i,\ldots$ such that $y^i\in \mathcal{Q}_f(x)$, and $y^i\rightarrow y^0$ as $i\rightarrow \infty$. 
	%For any $y'>y$, $\mathcal{I}_f(y)\subset \mathcal{I}_f(y')$. Therefore, if $x\in Conv(\mathcal{I}_f(y))$, $x$ must be in the set $Conv(\mathcal{I}_f(y'))$. Therefore, if $y\in \mathcal{Q}_f(x)$, then for any $y'>y$, $y'\in \mathcal{Q}_f(x)$. \inote{Ye: where is the previous property used? Ye: - not used later, i think we can erase it}
	%Let $y_0 = inf(\mathcal{Q}_f(x))$, so $y_0\geq f(x)$ exists. 
	For each $y^i$, $x\in \mathrm{conv}[\mathcal{I}_f(y^i)]$. Hence, by Carath\'eodory's theorem, there exist $d_x \leq k+1$ points $X^i:=(x_1^i,\ldots,x_{d_x}^i)$ where $x_j^i \in \mathcal{I}_f(y^i), j=1,2,\ldots,d_x$, and $\alpha^i:=(\alpha_1^i,\ldots,\alpha_{d_x}^i)\in \Delta_{d_x-1}$ where $\Delta_{d_x-1}$ is the $(d_x-1)$-simplex, such that $x=\sum_{j=1}^{d_x} \alpha_j^i x_{j}^i$. Since $\mathcal{X}$ and $\Delta_{d_x-1}$ are both compact, there must exist a subsequence of $(X^i,\alpha^i)$ that converges to a limit point $(X^0,\alpha^0)$ where $X^0 = (x_1^0, \ldots, x_{d_x}^0)$, $x_j^0 \in\mathcal{X}$, $j = 1\ldots,d_x$, and  $\alpha^0\in \Delta_{d_x-1}$. For simplicity, let us just assume that $(X^i,\alpha^i)$ converges to $(X^0,\alpha^0)$. Consequently, $x=\sum_{j=1}^{d_x} \alpha_j^0 x_{j}^0$. 
	%\inote{Ye: where is the previous property used? Ye: it is used in the next sentence} 
	By $f \in  \ell^\infty_S(\X)$ and $\lim_{i \to \infty} x_j^i = x_j^0$,  $f(x_j^0)\leq \liminf_{i \to \infty} f(x_j^i) \leq \liminf_{i\rightarrow \infty} y_i= y^0$, where the second inequality follows from $f(x_j^i)\leq y^i$ for each $x_j^i$ by definition of $\mathcal{I}_f(y^i)$. Hence, $x_j^0\in \mathcal{I}_f(y^0)$, for all $j$. Since $x=\sum_{j=1}^{d_x} \alpha_j^0 x_{j}^0$, and $x_j^0\in \mathcal{I}_f(y^0)$, it must be that $x\in \mathrm{conv}[\mathcal{I}_f(y^0)]$. Therefore, $y^0\in \mathcal{Q}_f(x)$. We conclude that $\min(\mathcal{Q}_f(x)) = y^0$ because $y^0 = \inf \mathcal{Q}_f(x)$ by definition.
	
	We next show that $\mathbf{Q} f \in \ell^\infty_S(\X)$ for any $f \in  \ell^\infty_S(\X)$. It is easy to see that $\mathbf{Q} f  \in \ell^\infty(\X)$ because $\inf_{x\in \X} f(x)\leq \mathbf{Q} f(x)\leq \sup_{x\in \X} f(x)$.
	%\inote{Ye: where/how is the previous property used? -Ye: add one more sentence here} Therefore, $\mathbf{Q} f  \in \ell^\infty(\X)$ must be bounded from above and below.  
	We prove the result by contradiction. Suppose that $\mathbf{Q} f \not\in \ell^\infty_S(\X)$, i.e., there exists $x^0\in \mathcal{X}$, and a sequence $x^1,\ldots,x^n,\ldots\in \mathcal{X}$ such that $x^i\rightarrow x^0$ and $\liminf_{i\rightarrow\infty} \mathbf{Q} f(x^i) < y^0 - \epsilon$ for some constant $\epsilon>0$ and $y^0 := \mathbf{Q} f(x_0)$. Since $\mathcal{X}$ is compact, there must exist a subsequence of $\{x^{i}\}_{i=1}^\infty$, denoted as $\{x^{n_i}\}_{i=1}^\infty$, such that for all $i$,
	\begin{equation}\label{eq:r1}
		y^i:=\mathbf{Q} f(x^{n_i})\rightarrow c < y^0 - \epsilon.
	\end{equation} 
	For simplicity, we can assume that $x^{n_i} = x^i$ for all $i=1,2,\ldots$.  Similar to the proof above, each $x^i$ can be written as $\sum_{j=1}^{d_x} \alpha_{j}^i x_j^i$, where $\alpha^i = (\alpha_j^i,\ldots,\alpha_{d_x}^i)\in \Delta_{d_x-1}$ and $x_j^i\in \mathcal{I}_f(y_i)$. By compactness of $\mathcal{X}\times \Delta_{d_x-1}$, there exist subsequences of $X^i:=(x_1^i,\ldots,x_{d_x}^i)$ and $\alpha^i$, $i=1,2,\ldots$, such that they converge to $X^0:=(x_1^0,\ldots,x_{d_x}^0)$ and $\alpha^0\in \Delta_{d_x-1}$. Again, for simplicity, we can assume that the subsequences are the sequences $X^i$ and $\alpha^i$. Since $x^i\rightarrow x^0$, it is easy to see that 
	$x^0=\sum_{j=1}^{d_x}\alpha^0_j x_j^0$.
	By $f \in  \ell^\infty_S(\X)$, for $j = 1,\ldots, d_x$,
	\begin{equation}\label{eq:r2}
		f(x_j^0)\leq \liminf_{i\to \infty} f(x_j^i).
	\end{equation}
	Moreover,  $f(x_j^i)\leq y^i$ for all $i$ because $x_j^i\in \mathcal{I}_f(y_i)$. Combining this result with \eqref{eq:r1} and \eqref{eq:r2} yields that $f(x_i^0)\leq \liminf_{i\to \infty} f(x_j^i) \leq\liminf_{i\to \infty} y^i< y^0-\epsilon$ for all $i=1,2,\ldots,d_x$. Hence, $x^0\in \mathcal{I}_f(y^0-\epsilon)$ as $x^0 = \sum_{j=1}^{d_x}\alpha^0_j x_j^0$. By definition of $\mathbf{Q}f$, it must be that $\mathbf{Q}f(x^0)\leq y^0-\epsilon<y^0$, which leads to a contradiction with $y^0 = \mathbf{Q} f(x^0) =  \min \mathcal{Q}_f(x^0)$.
	%it must be that 
	%\begin{equation}\label{eq:contra1}
	%x^0\in \mathcal{I}_f(y^0-\epsilon).
	%\end{equation}
	%But we know that $\mathbf{Q} f(x^0) = y^0 = \min \{y : x\in \textrm{conv}(\mathcal{I}_f(y))\}$, which contradicts with (\ref{eq:contra1}). Hence, $\mathbf{Q} f$ must be lower semi-continuous.
\end{proof}

%\begin{proof}[Proof of Lemma 5]
We now proceed to prove Theorem~\ref{QC}.  We first show that $\mathbf{Q}$ satisfies the three properties of  Definition~\ref{def:SE_operators}.

(1) For any  $f \in \ell^{\infty}_S(\X)$, the lower contour set of $\mathbf{Q}f$ at level $y$ is defined as $\mathcal{I}_{\mathbf{Q}f}(y) = \{x \in \X : \mathbf{Q}f(x) \leq y\} := \mathrm{conv}[\mathcal{I}_f(y)]$, where $\mathcal{I}_f(y)$ is the lower contour set of $f$ at level $y$. Since  $\mathcal{I}_{\mathbf{Q}f}(y)$ is convex for any $y$, $\mathbf{Q}f\in \ell^{\infty}_Q(\mathcal{X})$.

(2) If $f \in  \ell_Q^{\infty}(X)$, then $\mathrm{conv}[\mathcal{I}_f(y)] = \mathcal{I}_f(y)$ for any $y \in \RR$. Thus, the lower contour set of $f$ agrees with the lower contour set of $\mathbf{Q}f$ at any level $y$, which implies that $f=\mathbf{Q}f$.

(3) If $f\geq g$, then $\mathcal{I}_f(y)\supset \mathcal{I}_g(y)$ at any level $y \in \RR$. If follows that $\mathrm{conv}[\mathcal{I}_f(y)] \supseteq \mathrm{conv}[\mathcal{I}_g(y)]$, which means that the level set of $\mathbf{Q}f$ contains the level set of $\mathbf{Q}g$ at any level $y$, i.e., $\mathbf{Q}f \geq \mathbf{Q}g$.

We next show that $\mathbf{Q}$ is $d_\infty$-distance contraction. For any $f,g\in \ell_S^\infty(\X)$, let $\epsilon:=\|f-g\|_\infty$. Then, $g(x)-\epsilon\leq f(x) \leq g(x)+\epsilon$. It is easy to see that $\mathbf{Q} (g+c) = \mathbf{Q}g + c$ for any constant $c$. By order preserving property of $\mathbf{Q}$, $\mathbf{Q}g(x)-\epsilon\leq \mathbf{Q}f(x)\leq \mathbf{Q}g(x)+\epsilon$. It follows that 
$$||\mathbf{Q} g - \mathbf{Q}f||_\infty\leq \epsilon = ||f-g||_\infty.$$\qed
%this proof is very simple...
%write a short proof for it

%\end{proof}

\subsection*{Proof of Theorem~\ref{composition2}}

%\begin{proof}[Proof of Lemma 6]
Without loss of generality, we can assume that the domain $\mathcal{X} = [0,1]^k$. For a vector $w\in \mathbb{R}^k$, denote $w(i)$ as the $i^{th}$ entry of $w$.

%We show that $\mathbf{QM}$ satisfies the first property of  Definition~\ref{def:SE_operators}.

%\inote{Ye: Elaborate more on second and third properties}

(1) We first prove that $\mathbf{QM}$ is reshaping.

For any $f\in \ell^{\infty}_S(\X)$, $\mathbf{QM} f \in \ell_Q^{\infty}(\X)$ by Theorem~\ref{QC}. Therefore, we only need to show that $\mathbf{QM} f \in \ell_M^{\infty}(\X)$. Let $g:=\mathbf{M}  f$. By Theorem~\ref{RA}, $g \in \ell_M^{\infty}(\X)$, so that for any $y\in \mathbb{R}$, the lower contour set $\mathcal{I}_g(y)$  satisfies:
\begin{equation}\label{eq:lm6}
	\textrm{ For any } x\in \mathcal{I}_g(y) \textrm{ and } x'\in \mathcal{X} \textrm{ such that } x'\leq x,  x'\in \mathcal{I}_g(y).
\end{equation}
Therefore, we need to prove that for any $x, y$ such that  $x\in\mathrm{conv}[\mathcal{I}_g(y)]$,  $x'\in \mathrm{conv}[\mathcal{I}_g(y)]$ for any $x' \in \X$ such that $x' \leq x$.

First, we show the following:
\begin{equation}\label{conlm6}
	\textrm{ If } x'=x-t e_i \textrm{ for some } t\geq 0, \textrm{ then \ } x'\in \mathrm{conv}[\mathcal{I}_g(y)],
\end{equation}
where $e_i$ is defined as the $i^{th}$ standard unit vector, $i=1,2,\ldots,k$. Without loss of generality, we can simply assume that $i=1$, so $x'$ and $x$ are the same for all entries except for the first one. By assumption that $\mathcal{X}=[0,1]^k$, we know that the first entry of $x'$, denoted as $x'(1)$, must be non-negative. Since $x\in \mathrm{conv}[\mathcal{I}_g(y)]$,
%and $\mathscr{C}_0(\mathcal{I}_g(y))$ is the convexification of $\mathcal{I}_g(y)$,
by Carath\'eodory's theorem, there exists a finite set of points $x_1,\ldots,x_{d_x}$ such that $d_x \leq k+1$, $x_j\in \mathcal{I}_g(y)$, $j=1,2,\ldots,d_x$, and $(\alpha_1,\ldots,\alpha_{d_x}) \in \Delta_{d_x-1}$, such that $\sum_{j=1}^{d_x} \alpha_j x_j =x$. Define $\widetilde{x}_j = (0,x_{j}(2),\ldots,x_{j}(k))$ as a vector which is constructed by replacing the first entry of $x_j$ with $0$. Therefore, $\widetilde{x}:=\sum_{j=1}^{d_x} \alpha_j \widetilde{x}_j  = (0,x(2),\ldots,x(k))$ is a vector such that $\widetilde{x}\leq x'\leq x$. Therefore, there must exist $x_1^*,\ldots,x_{d_x}^*$ such that $x_j^* = (x_{j}^*(1), x_{j}(2),\ldots,x_{j}(k) )$ with $x_{j}^*(1)\in [0,x_{j}(1)]$ such that $\sum_{j=1}^{d_x} \alpha_j x_{j}^*(1) = x'(1)\in [0,x(1)]$. By construction, $x_j^*\in \mathcal{X}$. Since $x_j^*\leq x_j\in \mathcal{I}_g(y)$, (\ref{eq:lm6}) implies that $x_j^*\in \mathcal{I}_g(y)$. It follows that $\sum_{j=1}^{d_x} \alpha_j x_j^* = x'$, and therefore $x'\in  \mathrm{conv}[\mathcal{I}_g(y)]$.

Now, for any $x'\in \mathcal{X}$ such that $x'\leq x$, denote $v:=x-x'\geq 0$. Since $x\in  \mathrm{conv}[\mathcal{I}_g(y)]$, it follows that $x-v(1) e_1\in  \mathrm{conv}[\mathcal{I}_g(y)]$, and then that $(x-v(1)e_1) - v(2)e_2\in  \mathrm{conv}[\mathcal{I}_g(y)]$, \ldots. Therefore, after applying (\ref{conlm6}) for $k$ times, $x' = x-v(1)e_1-v(2)e_2-\cdots-v(k)e_k \in  \mathrm{conv}[\mathcal{I}_g(y)]$.
By (\ref{def:Q}), $\mathbf{Q} g(x):=\min\{y \in \RR : x \in\mathrm{conv}[\mathcal{I}_g(y)]\}$.  Let  $y':=\mathbf{Q} g(x)$ so that $x\in \mathrm{conv}[\mathcal{I}_g(y')]\}$. Then, for any $x'\in \mathcal{X}$ such that $x'\leq x$,  $x'\in \mathrm{conv}[\mathcal{I}_g(y'))\}$. That implies $\mathbf{Q} g (x')\leq y' = \mathbf{Q} g(x)$. Therefore, we conclude that $\mathbf{Q} g = \mathbf{QM} f$ is nondecreasing.
%That said, $\mathscr{Q}\circ \mathscr{M}$ is quasi-convex and monotonic enforcing with respect to $L^{\infty}(\mathcal{X})$.

(2) Next, we can show that $\mathbf{QM}$ satisfies the rest of the properties of Definition~\ref{def:SE_operators} using the same argument as in the proof of Theorem~\ref{composition}, replacing  $\mathbf{C}$ with $\mathbf{Q}$. We omit it for brevity.

(3) Finally, since $\mathbf{M}$ and $\mathbf{Q}$ are both $d_\infty$-distance contractions by Theorems \ref{RA} and~\ref{QC}, the composite map $\mathbf{QM}$ is also a $d_\infty$-distance contraction.
%
%(2) We show that $\mathbf{QM}$ has invariance.
%
%By Theorems \ref{RA} and~\ref{QC}, we know that $\mathbf{Q}, \mathbf{M}$ have neurality. Therefore, for any $f\in \ell_{QM}^\infty(\X)$, we have $\mathbf{QM} f= \mathbf{Q} \mathbf{M}f = \mathbf{Q} f = f$. That said, $\mathbf{QM}$ has invariance.
%
%(3) The second and third properties follow from Theorems \ref{RA} and~\ref{QC}.
%
%By Theorems \ref{RA} and~\ref{QC}, we know that $\mathbf{Q}, \mathbf{M}$ are order preserving. Therefore, for any $f, g\in \ell_{QM}^\infty(\X), f\geq g$, we have $\mathbf{M} f \geq \mathbf{M}g$. Hence, $\mathbf{Q} \mathbf{M} f\geq \mathbf{Q}\mathbf{M} g$. That said, $\mathbf{QM}$ is order preserving.
%
\qed

%\end{proof}

%\inote{Ye: we need a proof for Theorem~\ref{cRR}. I have defined the composition by applying the range operator first because I think it should make the proof easier. In (i) and (iii) we need to show that if $f \in \ell_S^{\infty}(\X)$ then $\mathbf{R} f \in \ell_S^{\infty}(\X)$.}

\subsection*{Proof of Theorem~\ref{cRR}}
We first show that if $f \in \ell_S^{\infty}(\X)$, then $\mathbf{R} f \in \ell_S^{\infty}(\X)$. This result is used in  parts (i) and (iii) to ensure that we apply the $\mathbf{C}$ and $\mathbf{Q}$ operators to lower semi-continuous functions. By  $f\in \ell_S^\infty(\X)$, for any sequence $x^1,...,x^n,...\in \X$ such that $\lim_{i\rightarrow \infty} x^i=x^0\in \X$, $f(x^0) \leq \liminf_{i\rightarrow \infty} f(x^i)$. Therefore, for any $\epsilon>0$, there exists $N$ large enough such that for any $i>N$, $f(x^i)>f(x^0)-\epsilon$. It follows that for any $i>N$, $\mathbf{R} f(x^i) > \mathbf{R} f(x^0) - \epsilon$. Hence, $\liminf_{i\rightarrow \infty} \mathbf{R} f(x^i)\geq \mathbf{R} f(x^0)$, i.e., $\mathbf{R} f \in \ell_S^{\infty}(\X)$.
%\inote{Ye: I think the previous argument is not correct as it does not hold when $f$ is continuous. Does not the result follow more directly from the operator $\mathbf{R}$ being order preserving? - it does not follow directly from the order preserving, as $\mathbf{R} f(x^i)$ can be smaller than $ \mathbf{R} f(x^0)$, but not that smaller. We need to show that the limit of $\mathbf{R} f(x^i)$ is $\geq \mathbf{R} f(x^0)$. I checked, there are some small errors. The proof should be correct now - please double check.}  

We now proceed to prove each of the parts of the Lemma. 

Part (i): $\mathbf{C} \mathbf{R} f\in \ell_C^{\infty}(\X)$ by the definition of $\mathbf{C}$ applied to $\mathbf{R} f$ and $\mathbf{R} f \in  \ell_S^{\infty}(\X)$. Moreover, by statement (2) of Lemma~\ref{lemma:DLF-basic}, there exist $d\leq k+2$ points $x_1,\ldots,x_d \in \mathcal{X}$ and $\alpha_1>0,\ldots,\alpha_d>0$, $\sum_{i=1}^d \alpha_i=1$ such that $\mathbf{C} \mathbf{R} f(x) = \sum_{i=1}^d \alpha_i \mathbf{R} f(x_i)$, where $x=\sum_{i=1}^d \alpha_i x_i$. Therefore, $\mathbf{C} \mathbf{R} f\in \ell_R^{\infty}(\X)$ because $\mathbf{R} f(x_i)\in [\underline{f}, \bar{f}]$ for all $x_i\in \X$. It is easy to see that $\mathbf{C}\mathbf{R}$ also satisfies invariance and order preservation because it is a composition of two operators that satisfy these properties. Indeed, if $f\in \ell_{CR}^\infty(\X)$ then  $\mathbf{C}\mathbf{R} f= \mathbf{C} (\mathbf{R} f) = \mathbf{C} f= f$, and if $g\geq f$, $g,f \in \ell_S^{\infty}(\X)$, then $\mathbf{R} g \geq \mathbf{R} f$, $\mathbf{R}g, \mathbf{R}f \in \ell_S^{\infty}(\X)$, and $\mathbf{C}(\mathbf{R} g)\geq \mathbf{C}(\mathbf{R} f)$.
%
%Since $\mathbf{C}$ and $\mathbf{R}$ both have order preserving and invariance, therefore, 
%(a) $\mathbf{C}\mathbf{R} f= \mathbf{C} (\mathbf{R} f) = \mathbf{C} f= f$, if $f\in \ell_{CR}^\infty$,
%
%(b) for any $g\geq f$, $\mathbf{R} g \geq \mathbf{R} f$, and $\mathbf{C}\mathbf{R} g\geq \mathbf{C}\mathbf{R} f$.
%
Hence, $\mathbf{C} \mathbf{R}$ is $\ell^\infty_{CR}$-enforcing with respect to $\ell_S^{\infty}(\X)$.  By Theorems \ref{RR} and~\ref{DFL}, both $\mathbf{C}$ and $\mathbf{R}$ are $d_\infty$-distance contractions. Therefore, the composite map $\mathbf{C} \mathbf{R}$ must be a $d_\infty$-distance contraction.

%by definition of  $\mathbf{R}$, $\underline{f} \leq \mathbf{R} f(x)\leq \bar{f}$ for any function $f\in \ell_S^\infty(\X)$, and  therefore $\mathbf{R} f(x) \in \ell_R^\infty(\X) $. Since $\mathbf{C} \mathbf{R} f$ is the convex minorant of $\mathbf{R} f$, $\mathbf{C} \mathbf{R} f\in \ell_C^{\infty}(\X)$. Moreover, by statement (2) of Lemma~\ref{lemma:DLF-basic}, there exist $d\leq k$ points $x_1,\ldots,x_d \in \mathcal{X}$ and $\alpha_1>0,\ldots,\alpha_d>0$, $\sum_{i=1}^d \alpha_i=1$ such that $\mathbf{C} \mathbf{R} f(x) = \sum_{i=1}^d \alpha_i \mathbf{R} f(x_i)$, where $x=\sum_{i=1}^d \alpha_i x_i$. Moreover, $\mathbf{C} \mathbf{R} f(x)\in [\underline{f}, \bar{f}]$ because $\mathbf{R} f(x_i)\in [\underline{f}, \bar{f}]$ for all $x_i\in \X$. Therefore, $\mathbf{C} \mathbf{R} f\in \ell_R^{\infty}(\X)$. Hence, $\mathbf{C} \mathbf{R}$ is $\ell^\infty_{CR}-enforcing$ with respect to $\ell_S^{\infty}(\X)$. 
%
%By Theorems \ref{RR} and~\ref{DFL}, both $\mathbf{C}$ and $\mathbf{R}$ are $d_\infty-$ distance contraction. Therefore, the composite map must be $d_\infty-$ distance contraction.

Part (ii):  $\mathbf{R} f\in \ell_R^{\infty}(\X)$ by the definition of $\mathbf{R}$. The $\mathbf{M}$ operator is the average of sorting operators, where each sorting operator does not change the range of the function. Therefore, $\mathbf{M}\mathbf{R} f \in  \ell_{MR}^{\infty}(\X)$. As in part (i), it is easy to see that $\mathbf{M}\mathbf{R}$ also satisfies invariance and order preservation because it is a composition of two operators that satisfy these properties.
%Since $\mathbf{M}$ and $\mathbf{R}$ both have order preserving and invariance, therefore, 
%(a) $\mathbf{M}\mathbf{R} f= \mathbf{M} (\mathbf{R} f) = \mathbf{M} f= f$, if $f\in \ell_{MR}^\infty$,
%
%(b) for any $g\geq f$, $\mathbf{R} g \geq \mathbf{R} f$, and $\mathbf{M}\mathbf{R} g\geq \mathbf{M}\mathbf{R} f$.
%
Hence, $\mathbf{M}\mathbf{R}$ is $\ell^{\infty}_{MR}$-enforcing with respect to $\ell^{\infty}(\X)$. Since $\mathbf{M}$ and $\mathbf{R}$ are both $d_p$-distance contractions for any $p\geq 1$, it must be that $\mathbf{M}\mathbf{R}$ is $d_p$-distance contraction for any $p\geq 1$.

Part (iii):  let $f \in \ell_S^\infty(\X) \cap  \ell_R^\infty(\X)$. By definition $\mathbf{Q} f(x) = \min\{y \in \mathbb{R}: x\in \mathrm{conv}[\mathcal{I}_f(y)]\}$, so that $\bar{f} \in \mathcal{Q} _f(x):=\{y \in \mathbb{R}: x \in \mathrm{conv}[\mathcal{I}_f(y)]\}$ and $\mathbf{Q} f(x) \leq \bar{f}$. Moreover, for any $y<\underline{f}$, $\mathcal{I}_f(y)=\emptyset$, so that $y\notin \mathcal{Q} _f(x)$ and $\mathbf{Q} f(x)\geq \underline{f}$. Consequently, $\mathbf{Q} f \in \ell^\infty_R(\X)$. On the other hand, if $f \in \ell_S^\infty(\X)$  then $\mathbf{Q} f \in \ell_S^\infty(\X)$ by Theorem~\ref{QC}. Combining the previous results, $\mathbf{Q} f  \in \ell_S^\infty(\X) \cap \ell^\infty_R(\X)$. Since $\mathbf{R} f \in \ell_S^\infty(\X) \cap \ell^\infty_R(\X)$ for any $f \in \ell_S^\infty(\X)$, then $\mathbf{Q} \mathbf{R} f \in \ell_S^\infty(\X) \cap \ell^\infty_R(\X)$. As in part (i), it is easy to see that $\mathbf{Q}\mathbf{R}$ also satisfies invariance and order preservation because it is a composition of two operators that satisfy these properties.
%Since $\mathbf{Q}$ and $\mathbf{R}$ both have order preserving and invariance, therefore, 
%(a) $\mathbf{Q}\mathbf{R} f= \mathbf{Q} (\mathbf{R} f) = \mathbf{Q} f= f$, if $f\in \ell_{QR}^\infty$,
%
%(b) for any $g\geq f$, $\mathbf{R} g \geq \mathbf{R} f$, and $\mathbf{Q}\mathbf{R} g\geq \mathbf{Q}\mathbf{R} f$.
%
Hence, $\mathbf{Q}\mathbf{R}$ is $\ell^{\infty}_{QR}$-enforcing with respect to $\ell_S^{\infty}(\X)$. By Theorems \ref{RR} and~\ref{QC}, both $\mathbf{R}$ and $\mathbf{Q}$ are $d_\infty$-distance contractions. Therefore, the composite map $\mathbf{Q} \mathbf{R}$ must be a $d_\infty$-distance contraction.\qed

\subsection*{Proof of Theorem~\ref{transformation}} We need to show the 3 properties of Definition~\ref{def:SE_operators}. (1) By definition, $\widetilde{f}:=\mathbf{O}(h\circ f) \in \ell_1^\infty(\X)$. Therefore, $\mathbf{O}_h(f)=h^{-1}\circ f\in \ell_{h,1}^\infty(\X) $ and reshaping  holds. (2) If $f\in \ell^\infty_{h,1}(\X)$, then $h\circ f\in \ell_1^\infty(\X)$ and $\mathbf{O}(h\circ f) = h\circ f$ by the invariance property of $\mathbf{O}$. Hence, $\mathbf{O}_h = h^{-1}\circ h\circ f =f$. (3) Since $h$ is a real-valued bijection, it must be that $h$ and $h^{-1}$ are both strictly increasing or strictly decreasing.  We can assume that they are strictly increasing without loss of generality. For any $f,g\in \ell_{h,0}^{\infty}(\X)$ such that $f\geq g$, 
$h\circ f\geq h\circ g$, $\mathbf{O}(h\circ f)\geq \mathbf{O}(h\circ g),$
and $h^{-1}\circ \mathbf{O}(h\circ f)\geq h^{-1}\mathbf{O}(h\circ g)$. Therefore, $\mathbf{O}_h$ is order-preserving.

To show contractivity, let $\widetilde{f}:=\mathbf{O}_h(f)$ and $\widetilde{g}:=\mathbf{O}_h(g)$. Then, 
$\rho_h(\widetilde{f},\widetilde{g}) = \rho(h\circ \widetilde{f}, h\circ \widetilde{g}) = \rho(\mathbf{O}(h\circ f), \mathbf{O}(h\circ g))$. Since $\mathbf{O}$ is a $\rho$-distance contraction, $\rho(\mathbf{O}(h\circ f), \mathbf{O}(h\circ g))\leq \rho(h\circ f,h\circ g) = \rho_h(f,g)$. Hence, $\mathbf{O}_h$ is a $\rho_h$-distance contraction. \qed

\subsection*{Proof of Corollary~\ref{CI}}

% [Inference]
% \section{Design of the Monte-Carlo experiment}

%\begin{proof}[Proof of Theorem 1]
(1) We show that the event $\{f_l \leq f_0 \leq f_u\}$ implies the event $\{\mathbf{O} f_l \leq f_0 \leq \mathbf{O} f_u\}$ by the properties of the $\mathbf{O}$-operator. Indeed, by order preservation, $\{f_l \leq f_0 \leq f_u\} $  implies that $\{\mathbf{O} f_l \leq \mathbf{O} f_0 \leq \mathbf{O} f_u\}$, which is equivalent to $\{\mathbf{O} f_l \leq f_0 \leq \mathbf{O} f_u\}$ because $\mathbf{O} f_0  = f_0$ by invariance.

%By definition, with probability $1-\alpha$, $f_u\geq f\geq f_l$. By the order preserving property, $\mathscr{T}(f_u)\geq \mathscr{T}(f)\geq \mathscr{T}(f_l)$. Therefore, $\mathscr{T}(f)\in [\mathscr{T}(f_l),\mathscr{T}(f_u)]$ with probability at least $1-\alpha$.

(2) The result follows from $\rho(\mathbf{O} f_0, \mathbf{O} f) \leq \rho(f_0, f)$ by $\rho$-distance contraction and  $\mathbf{O} f_0  = f_0$ by invariance.

%(2) Since the operator is $\mathcal{F}$-enforcing, by Definition~\ref{def:SE_operators}, for any two functions $g,h$, it must be that
%\begin{equation}\label{eq:tm1}
%\rho(\mathscr{T}(g),\mathscr{T}(h))\leq\rho(g,h).
%\end{equation}
%Applying (\ref{eq:tm1}) to $(f_l,f)$, $(f_u,f)$ and $(f_l,f_u)$, we get the results in statement (2) of this theorem.

(3) The result follows directly by $\rho$-distance contraction.
\qed
%\end{proof}

%\vspace{-.2in}
\bigskip

% comment out for JMLR because JMLR uses plainnat
%\bibliographystyle{abbrv}
\bibliography{ref_app}
	
\end{document}